\documentclass[draftclsnofoot,12pt,onecolumn]{IEEEtran}        
\usepackage{amsmath, amsthm, amssymb}
\usepackage{epsfig}
\usepackage{latexsym}
\usepackage{graphicx}
\usepackage{cite}
\usepackage{epsfig}
\usepackage{subfigure}
\usepackage{color}

\newtheorem{theorem}{Theorem}
\newtheorem{lemma}{Lemma}
\newtheorem{proposition}{Proposition}
\newtheorem{definition}{Definition}
\newtheorem{corollary}{Corollary}  
\newtheorem{remark}{Remark}

\title{Analytic Expressions and Bounds for Special Functions and Applications in Communication Theory}

\author{
Paschalis~C.~Sofotasios, Theodoros~A.~Tsiftsis, ~Yury~A.~Brychkov, Steven~Freear, ~Mikko Valkama, ~and George~K.~Karagiannidis

\thanks{P. C. Sofotasios was with the School of Electronic and Electrical Engineering, University of Leeds, LS2 9JT Leeds, UK.  He is now with the Department of Electronics and Communications Engineering, Tampere University of Technology, 33101 Tampere, Finland and with the Department of Electrical and Computer Engineering, Aristotle University of Thessaloniki, 54124 Thessaloniki, Greece  \, (e-mail: p.sofotasios@ieee.org).  }

\thanks{T. A. Tsiftsis is with the Department of Electrical Engineering, Technological Educational Institute of Central Greece, 35100 Lamia, Greece and with the Department of Electrical and Computer Engineering, Aristotle University of Thessaloniki, 54124 Thessaloniki, Greece \, (e-mail: tsiftsis@teilam.gr).}

\thanks{Yu. A. Brychkov is with the Dorodnicyn Computing Center of the Russian Academy of Sciences, 119333 Moscow, Russia \, (e-mail: brychkov@ccas.ru).}

\thanks{ S. Freear is with the School of Electronic and Electrical Engineering, University of Leeds, LS2 9JT Leeds, UK \, (e-mail: s.freear@leeds.ac.uk).}

\thanks{M. Valkama is with the Department of Electronics and Communications Engineering, Tampere University of Technology, FI-33101 Tampere, Finland \, (e-mail: mikko.e.valkama@tut.fi).}

\thanks{G. K. Karagiannidis is with the Department of Electrical and Computer Engineering, Aristotle University of Thessaloniki, 54124 Thessaloniki, Greece and with the Department of Electrical and Computer Engineering, Khalifa University, PO Box 127788
Abu Dhabi, UAE  \, (e-mail: geokarag@ieee.org).}
}

\begin{document}

\maketitle

\begin{abstract}
This work is devoted to the derivation of novel analytic expressions and bounds for a family of special functions that are   useful in wireless communication theory. These functions are the well-known Nuttall $Q{-}$function, the incomplete Toronto function,  the Rice $Ie$-function and the incomplete Lipschitz-Hankel integrals.  
 Capitalizing on the offered results, useful identities are additionally derived between the above functions and the Humbert, $\Phi_{1}$, function as well as for specific cases of the Kamp${\it \acute{e}}$ de F${\it \acute{e}}$riet function. These functions  can be considered useful mathematical tools that can be employed in   applications relating to the analytic performance evaluation of   modern wireless communication systems such as cognitive radio, cooperative and free-space optical communications as well as radar, diversity  and multi-antenna systems. As an example, new closed-form expressions are derived for the outage probability over non-linear generalized fading channels, namely, $\alpha{-}\eta{-}\mu$, $\alpha{-}\lambda{-}\mu$ and $\alpha{-}\kappa{-}\mu$ as well as for specific cases of the $\eta{-}\mu$ and $\lambda{-}\mu$ fading channels. Furthermore, simple   expressions are presented for the channel capacity for the truncated channel inversion with fixed rate   and the corresponding optimum cut-off signal-to-noise ratio for single-and multi-antenna communication systems over Rician  fading channels.   The accuracy and validity of the derived expressions is justified through extensive comparisons with respective numerical results.
\end{abstract}


\section{Introduction}

It is widely known that special functions constitute invaluable mathematical tools in most fields of natural sciences and engineering. In the broad area of wireless communications, their utilization often allows the derivation of useful expressions for important performance measures such as error probability, channel capacity and higher-order statistics (HOS). The computational realization of such expressions is typically straightforward since the majority of special functions, that are used in digital communications, are included as built-in functions in popular mathematical software packages such as MAPLE, MATLAB and MATHEMATICA. Among others, the Marcum $Q{-}$function, $Q_{m}(a,b)$, the Nuttall $Q{-}$function, $Q_{m,n}(a,b)$, the Rice $Ie{-}$function, $Ie(k,x)$, the incomplete Toronto function (ITF), $T_{B}(m,n,r)$, and the incomplete Lipschitz-Hankel integrals (ILHIs), $Ze_{m,n}(x;a)$, were proposed several decades ago \cite{J:Marcum_1, J:Marcum_2, J:Marcum_3, J:Swerling, B:Proakis_Book, J:Shnidman1989, J:Helstrom1992, Simon_Marcum, B:Alouini, J:Nuttall, J:Simon_2, J:Hatley, J:Sagon, J:Fisher, J:Rice_1, B:Roberts, J:Rice_2, J:Tan, J:Agrest1971, B:Maksimov, J:Miller1989, J:Dvorak} and have been largely involved in communication theory and in the analytic performance evaluation of wireless communications systems \cite{B:Alouini_2005, J:Loskot2009, J:Cheng, J:Tellambura2013, J:Paris2009_EL, J:Jimenez2010, J:Paris_Hoyt_2010, J:Labao, C:Alsusa, J:Farina, J:Tulino2010, J:Paris2012_CommL, J:Paris2013_TCOM, J:Baricz2009_EL, J:Affes2009}  and the references therein.

More specifically, the $Q_{m}(a,b)$ function was proposed by Marcum in \cite{J:Marcum_1, J:Marcum_2} and became widely known in digital communications by applications relating to wireless transmission over fading or non-fading media \cite{J:Nuttall, B:Proakis_Book, J:Shnidman1989, J:Helstrom1992, Simon_Marcum, J:Simon_2, B:Alouini, B:Alouini_2005}. Its basic properties and identities were reported in \cite{J:Nuttall} and  several upper and lower bounds   were proposed in \cite{B:Alouini, J:Ferrari2002_T-IT, C:Kam2006_ISIT, C:Li2006, J:Zhao2008_EL, J:Kam2008, C:Sun2008_Globecom, J:Baricz2008, J:Baricz, J:Karagiannidis, J:Baricz2009_JMAA, J:Abreu, J:Kam201_T-IT, J:Baricz2010_IT, C:Kam2011, J:Koh2013_IET}. Furthermore, semi-analytic representations and approximations were given in \cite{C:Kam2006_VTC, C:Ding2008, J:Andras2011, J:Marcum_New_1} while various properties were investigated in \cite{J:Kam2008, C:Sun2008_Globecom, J:Karagiannidis, J:Baricz2010_IT, J:Brychkov2012,B:Vasilis_PhD}. Exact analytic expressions for the special cases that $m$ is a non-negative integer and  half-integer were derived in \cite{J:Helstrom1992} and \cite{J:Karagiannidis, B:Vasilis_PhD}, respectively, whereas  general expressions in terms of the confluent Appell function were derived in \cite{J:Brychkov2012} and also in \cite{J:Lozano2013} in the context of deriving closed-form expressions for the bivariate Nakagami${-}m$ distribution and the distribution of the minimum eigenvalue of correlated non-central Wishart matrices.

In the same context, the  $Q_{m, n}(a, b)$ function  was firstly proposed in \cite{J:Nuttall} and constitutes a generalization  of the Marcum $Q{-}$function. It is defined by a semi-infinite integral representation and it can be expressed in terms of the $Q_{m}(a, b)$ function and the modified Bessel function of the first kind, $I_{n}(x)$, for the special case that the sum of its indices is a real odd positive integer i.e. $(m + n +1){/}2 \in \mathbb{N}$.  Establishment of further properties, monotonicity criteria and the derivation of lower and upper bounds along with a closed-form expression for the case that $m \pm 0.5 \in \mathbb{N}$ and $n \pm 0.5 \in \mathbb{N} $ were reported in \cite{J:Karagiannidis, J:Baricz, J:Karasawa, J:Brychkov2013}.  Likewise, the incomplete Toronto function is a  special function, which was proposed by Hatley in \cite{J:Hatley}. It  constitutes a generalization of the  Toronto function, $T(m,n,r)$, and includes the $Q_{m}(a, b)$ function as a special case. Its definition is  given by a finite integral  while alternative representations include two  infinite series that were proposed in \cite{J:Sagon}. The incomplete Toronto function has been also useful in wireless communications as it has been employed in applications relating to statistical analysis, radar systems, signal detection and estimation as well as in error probability analysis \cite{J:Fisher, J:Marcum_3, J:Swerling}.

The Rice $Ie{-}$function is also a special function of similar analytic representation to the Marcum and Nuttall $Q{-}$functions. It was firstly proposed by S. O. Rice in \cite{J:Rice_1} and has been applied in investigations relating to zero crossings,  angle modulation systems, radar pulse detection and error rate analysis of differentially encoded systems \cite{B:Roberts, J:Rice_2, J:Tan, J:Pawula}. It is typically defined by a finite integral while alternative representations include two infinite series which involve the modified Struve function and an expression in terms of the Marcum $Q{-}$function \cite{J:Rice_2, J:Tan, J:Pawula, B:Tables}. Finally, the $Ze_{m,n}(x;a)$ integrals constitute a general class of incomplete cylindrical functions that have been encountered in analytic solutions of numerous problems in electromagnetic theory \cite{B:Maksimov, J:Dvorak}${-}$and the references therein. Their general representation is given in a non-infinite integral form and it accounts accordingly for the Bessel function of the first kind, $J_{n}(x)$, the Bessel function of the second kind, $Y_{n}(x)$, and their modified counterparts, $I_{n}(x)$ and $K_{n}(x)$, respectively. In the context of wireless communication systems, the ILHIs have been  utilized in the OP over generalized multipath fading channels as well as in the error rate analysis of MIMO systems under imperfect channel state information (CSI) employing adaptive modulation, transmit beamforming and maximal ratio combining (MRC), \cite{J:Jimenez2010, J:Paris2012_CommL, J:Paris}.

Nevertheless, in spite of the undoubted importance of the $Q_{m,n}(a,b)$, $T_{B}(m,nr)$, $Ie(k,x)$  functions and $Ze_{m,n}(x;a)$ integrals, they are all neither available in tabulated form nor are included as built-in functions in widely used mathematical software packages. As a consequence, their utilization becomes rather intractable and laborious to handle both algebraically and computationally. Motivated by this, analytic results on these special functions and integrals were reported in \cite{B:Sofotasios, C:Sofotasios_1, C:Sofotasios_2, C:Sofotasios_3, C:Sofotasios_4, C:Sofotasios_5, C:Sofotasios_6, C:Sofotasios_7, C:Sofotasios_8, C:Sofotasios_9, C:Sofotasios_10}. In the same context, the present work is devoted to elaborating  substantially on these results aiming to derive a comprehensive mathematical framework that consists of numerous  analytic expressions and bounds for the above special functions and integrals.    The offered results have a versatile  algebraic representation and can be useful in applications relating to natural sciences and engineering, including conventional and emerging wireless communications.  

In more details, the  contributions of the present paper are listed below:

\begin{itemize}

\item[$\bullet$] Closed-form expressions and simple polynomial approximations are derived for  the $Q_{m,n}(a,b)$, $T_{B}(m,n,r)$, $Ie(k,x)$ functions  and the $Ie_{m,n}(x;a)$ integrals\footnote{This work considers only the $Ie_{m,n}(x;a)$ case i.e. the $I_{n}(x)$ function-based $Ze_{m,n}(x;a)$ integrals. However, the offered analytic expressions can be readily extended for the case of $Je_{m, n}(x;a)$, $Ye_{m, n}(x;a)$ and $Ke_{m, n}(x;a)$  with the aid of the standard identities of the Bessel functions.}. These expressions are valid for all values of the involved parameters  and can readily reduce to exact  infinite series representations. 

\item[$\bullet$] Closed-form upper bounds are derived for the respective truncation errors of the proposed polynomial and series representations. 

\item[$\bullet$] Simple closed-form expressions are derived for specific cases of  the $T_{B}(m,n,r)$ function and the $Ie_{m,n}(x;a)$ integrals. 

\item[$\bullet$] Capitalizing on the derived expressions, generic closed-form upper and lower bounds are derived for the $T_{B}(m,n,r)$ function  and the $Ie_{m,n}(x;a)$ integrals. 

\item[$\bullet$] Simple closed-form upper and lower bounds are proposed for the $Ie(k,x)$ function  which under certain range of values become accurate approximations. 

\item[$\bullet$] Simple closed-form upper bounds are proposed  for the $Q_{m,n}(a,b)$, $T_{B}(m,n,r)$ functions and $Ie_{m,n}(x;a)$ integrals which for certain range of values can serve as particularly tight approximations. 

\item[$\bullet$]  Novel closed-form identities are deduced relating  specific cases of the Kamp$\acute{e}$ de F${\it \acute{e}}$riet (KdF) and Humbert, $\Phi_1$,  functions with the above special functions. These identities are useful because although $\Phi_{1}$ and particularly KdF functions are rather generic functions that are capable of representing  numerous other special functions, yet, they are currently neither explicitly tabulated nor  built-in functions in popular mathematical software packages such as MATLAB, MAPLE and MATHEMATICA.

\item[$\bullet$] The offered results are applied in the context of digital communications for deducing respective analytic expressions for: $i)$ the outage probability (OP) over non-linear generalized fading, namely, $\alpha{-}\eta{-}\mu$, $\alpha{-}\lambda{-}\mu$ and $\alpha{-}\kappa{-}\mu$ fading channels; $ii)$ the OP for $\eta{-}\mu$ and $\lambda{-}\mu$ fading channels for the special case that the value of $\mu$ is integer or half-integer; $iii)$ the channel capacity for the truncated channel inversion with fixed rate (TIFR) adaptive transmission technique of single-and multi-antenna systems over  Rician fading channels; $iv)$ the optimum cut-off SNR for the aforementioned  TIFR scenario in the case of single-input single-output (SISO), multiple-input single-output (MISO) and single-input multiple-output (SIMO) systems. 
\end{itemize}

To the best of the Authors' knowledge, the offered results have not been previously reported in the open technical literature. 
The remainder of this paper is organized as follows: New expressions  are derived for the Nuttall $Q{-}$function in Sec. II. Sec. III and Sec. IV are devoted to the derivation of closed-form expressions and bounds for the  ITF and Rice $Ie{-}$function, respectively. Analytic results for the ILHIs are derived in Section V while simple identities for special cases of the KdF and Humbert $\Phi_1$ functions are proposed in Sec. VI. Finally, applications in the context of wireless communications along with the necessary discussions are provided in Section VII while closing remarks are given in Section VIII.

\section{New Closed-Form Expressions and Bounds for the Nuttall $Q{-}$function}

 \begin{definition}
For $m, n, a, b \in \mathbb{R}^{+} $, the Nuttall $Q{-}$function is defined by the following semi-infinite integral representation \cite[eq. (86)]{J:Nuttall}, 
 
\begin{equation}  \label{Nuttall_1}
Q_{m, n}(a,b) \triangleq \int_{b}^{\infty} x^{m} e^{-\frac{x^2+a^2}{2}} I_{n} (ax) {\rm d}x.  
\end{equation}
\end{definition}

The Nuttall $Q{-}$function constitutes a generalization of the well-known  Marcum $Q{-}$function. The normalized Nuttall $Q{-}$function is expressed as,
 
 \begin{equation}
\mathcal{Q} _{m, n} (a, b) = \frac{Q_{m, n} (a, b)}{a^{n}}
\end{equation}
 which for the special case that  $m = 1$ and $n = 0$, reduces to the   Marcum $Q{-}$function, namely, 
 
 \begin{equation}  \label{Marcum_new}
Q_{1}(a,b)  \triangleq  \int_{b}^{\infty} x e^{-\frac{x^2+a^2}{2}} I_{0} (ax) {\rm d}x  
\end{equation}
 and thus,  $Q_{1, 0}(a, b) = \mathcal{Q} _{1, 0} (a, b) =Q_{1} (a, b) = Q(a,b)$.  In addition, for the special case that $n = m - 1$ it follows that, 
 
  \begin{align}  \label{Marcum_general_new}
  \mathcal{Q} _{m, m-1} (a, b) &= Q_{m} (a, b)  \\
  &= \frac{1}{a^{m-1}}\int_{b}^{\infty} x^{m} e^{-\frac{x^2+a^2}{2}} I_{m-1} (ax) {\rm d}x  
\end{align} 
 and thus,
 $Q_{m, m - 1} (a, b) = a^{m - 1} Q_{m} (a, b)$. Likewise, when $m$ and $n$ are positive integers, the following recursion formula is valid\cite[eq. (3)]{J:Simon_2}, 
 
\begin{equation} \label{Nuttall_Recursion}
Q_{m, n} (a, b) =  a Q_{m-1, n+1} (a, b) +   b^{m - 1} e^{ - \frac{a^{2} + b^{2}}{2}}  I_{n} (ab) +     (m + n - 1) Q_{m - 2, n} (a, b)  
\end{equation}

\noindent 
along with the finite series representation  in \cite[eq. (8)]{J:Simon_2}. Nevertheless, the validity of this series is not general because it is restricted to the special case that the sum of $m$ and $n$ is an odd positive integer i.e. $m + n \in \mathbb{N}$.

\subsection{A Closed-Form Expression in Terms of the Kamp${\it \acute{e}}$ de F${\it \acute{e}}$riet Function}

\begin{theorem}
For $m, n, a \in \mathbb{R}$ and $b \in \mathbb{R}^{+}$, the   Nuttall $Q{-}$function can be expressed as follows, 

\begin{equation} \label{KdF_1}
Q_{m,n}(a,b) = \frac{ a^{n} \Gamma\left(  \frac{m + n + 1}{2} \right) \, _{1}F_{1} \left( \frac{m + n + 1}{2}, 1 + n, \frac{a^{2}}{2} \right)}{n!e^{\frac{a^{2}}{2}} 2^{ \frac{n - m + 1}{2} }}  - \frac{ a^{n} b^{m + n + 1} F_{1,1}^{1,0} \left( _{\frac{m + n + 3}{2} : n + 1, -: }^{ \frac{m + n + 1}{2}: -, - : }  \frac{a^{2} b^{2}}{4}, - \frac{b^{2}}{2} \right) }{n! (m + n + 1) 2^{n} e^{ \frac{a^{2}}{2} }} 
\end{equation} 
where $\Gamma(a)$, $\,_{1}F_{1}(a,b,x)$ and $F_{.,.}^{.,.}(_{.}^{.} : .,.)$ denote the (complete)  Gamma function, the Kummer confluent hypergeometric function and the KdF function, respectively \cite{B:Abramowitz, J:Bailey, KdF_1, B:Exton, J:Srivastava, J:Watson, KdF_2, B:Prudnikov}.  
\end{theorem}

\begin{proof}
The proof is provided in Appendix A. 
\end{proof}

\subsection{A Simple Polynomial Representation}

In spite of the general usefulness of \eqref{KdF_1}, its presence in integrands along with other elementary and${/}$or special function can lead to intractable integrals due to the absence of relatively simple representations and properties for the $F_{.,.}^{.,.}(_{.}^{.} : .,.)$ function. Therefore, it is evident that a simple approximative formula that is valid for all values of the involved parameters is additionally useful.

\begin{proposition}

For $m, n, a \in \mathbb{R}$ and $b \in \mathbb{R}^{+}$, the $Q_{m,n}(a,b)$ function can be accurately approximated as follows,

\begin{equation} \label{Nuttall_Polynomial} 
Q_{m,n}(a,b) \simeq \sum_{l = 0}^{p} \frac{a^{n + 2l } \, \Gamma(p+l) p^{1 - 2l} \, \Gamma \left( \frac{m + n + 2l + 1}{2}, \frac{b^{2}}{2} \right) }{  l! (n+l)! 2^{\frac{n - m + 2l + 1}{2}}\,  (p-l)!  e^{\frac{a^{2}}{2}}   } 
\end{equation}
which for the special case that $ (m + n + 1){/}2 \in \mathbb{N}$, it can be equivalently expressed as,

\begin{equation} \label{Nuttall_Integer_Indices}
Q_{m,n}(a,b) \simeq \sum_{l = 0}^{p} \sum_{k=0}^{L} \frac{\mathcal{A} \, a^{n + 2l} b^{2k} \Gamma(p+l) p^{1 - 2l} \Gamma(L + l + 1)}{ l! k! \Gamma(n + l  +1) (p-l)! 2^{ l + k}  e^{- \frac{a^{2} + b^{2}}{2}} } 
\end{equation}
where 

\begin{equation} \label{coefficient_1}
L = \frac{m + n - 1}{2} + l  
\end{equation}
and

\begin{equation} \label{coefficient_2} 
\mathcal{A} = a^{n}  2^{ \frac{m - n - 1}{2}} e^{- \frac{a^{2} + b^{2}}{2}}
\end{equation}
with $p$ denoting the corresponding  truncation term, $\Gamma(a,x) $ is the  upper incomplete Gamma function \cite{B:Tables} whereas,

\end{proposition}

\begin{proof}
The proof is provided in Appendix B. 
\end{proof}

\begin{remark}

The  coefficients of the series in \cite[eq. (19)]{J:Gross_2} differ from the series in \cite[eq. (8.445)]{B:Tables} by the terms $ p^{1 - 2l} \Gamma(p + l) {/} (p - l)!$. Therefore, as $p \rightarrow \infty$, these terms vanish and as a result \eqref{Nuttall_Polynomial} and \eqref{Nuttall_Integer_Indices} reduce to the following exact infinite series representations, 

\begin{equation} \label{Nuttall_Polynomial_infinite}
Q_{m,n}(a,b) = \sum_{l = 0}^{\infty} \frac{a^{n + 2l } \, e^{- \frac{a^{2}}{2}}    \, \Gamma \left( \frac{m + n + 2l + 1}{2}, \frac{b^{2}}{2} \right) }{  l! \Gamma(n+l+1) 2^{\frac{n - m + 2l + 1}{2}} } 
\end{equation}
and

\begin{equation} \label{Nuttall_Integer_Indices_infinite}
Q_{m,n}(a,b) =\sum_{l = 0}^{\infty} \sum_{k=0}^{\frac{m + n - 1}{2} + l} \frac{ a^{n + 2l} b^{2k}  2^{ \frac{m - n - 1}{2}} \Gamma(\frac{m + n + 1}{2} + l)}{ l! k! \Gamma(n + l  +1)   2^{ l + k}  e^{ \frac{a^{2} + b^{2}}{2}} } 
\end{equation}
respectively to \eqref{Nuttall_Polynomial} and \eqref{Nuttall_Integer_Indices}.
\end{remark}

\begin{remark}
By setting $n = m-1$ in \eqref{Nuttall_Polynomial} and recalling that $\mathcal{Q}_{m,n}(a,b) = Q_{m,n}(a,b){/}a^{n}$ and $\mathcal{Q}_{m,m-1}(a,b) = Q_{m}(a, b)$, a new simple approximation is deduced for the Marcum $Q{-}$function,   

\begin{equation} \label{Marcum_Polynomial} 
Q_{m}(a,b) \simeq \sum_{l = 0}^{p} \frac{a^{2l } \,   \Gamma(p+l) p^{1 - 2l} \, \Gamma \left( m + l, \frac{b^{2}}{2} \right) }{  l! \Gamma(m+l) 2^{l}\,  (p-l)! e^{ \frac{a^{2}}{2}} }
\end{equation}
which for $ m \in \mathbb{N}$ it can reduce to,

\begin{equation} \label{Marcum_Integer_Indices}
Q_{m}(a,b) \simeq \sum_{l = 0}^{p} \sum_{k=0}^{m+l-1} \frac{ \Gamma(p+l) p^{1 - 2l} a^{2l} b^{2k} }{l! k! (p-l)!  2^{l+k}e^{\frac{a^{2} + b^{2}}{2}}}. 
\end{equation}
Based on Remark $1$, as $p \rightarrow \infty$, equations \eqref{Marcum_Polynomial} and \eqref{Marcum_Integer_Indices}   become exact infinite series, namely, 

\begin{equation} \label{Marcum_Polynomial} 
Q_{m}(a,b) = \sum_{l = 0}^{\infty} \frac{a^{2l } \,    \Gamma \left( m + l, \frac{b^{2}}{2} \right) }{  l! \Gamma(m+l) 2^{l}\,   e^{ \frac{a^{2}}{2}} }
\end{equation}
and

\begin{equation} \label{Marcum_Integer_Indices}
Q_{m}(a,b) = \sum_{l = 0}^{\infty} \sum_{k=0}^{m+l-1} \frac{   a^{2l} b^{2k} }{l! k!   2^{l+k}} e^{-\frac{a^{2} + b^{2}}{2}}
\end{equation} 
respectively. 
\end{remark}

\subsection{Truncation Error}

The proposed expressions  converge  rather quickly and  their accuracy is proportional to the value of $p$. However, determining the involved truncation error analytically is particularly  advantageous for ensuring certain accuracy levels when applied in analyses related to wireless communications.

\begin{lemma}
For $m, n, a \in \mathbb{R}$ and $b \in \mathbb{R}^{+}$, the following inequality can serve as a closed-form upper bound for the truncation error of the $ Q _{m,n}(a,b)$ function in \eqref{Nuttall_Polynomial},

\begin{equation}  \label{Nuttall_Truncation} 
\epsilon_{t}  \leq \sum_{k = 0}^{\lceil n \rceil_{0.5} - 1}\frac{ (-1)^{\lceil n \rceil_{0.5}}  \Gamma(2\lceil n \rceil_{0.5} - k - 1) \mathcal{I}_{\lceil m\rceil_{0.5},\lceil n \rceil_{0.5}}^{k} (a, b) }{k! \Gamma(\lceil n \rceil_{0.5} - k) (2a)^{-k} \sqrt{\pi} 2^{ \lceil n \rceil_{0.5} - \frac{1}{2}} a^{2\lceil n \rceil_{0.5} - 1} }   - \sum_{l = 0}^{p} \frac{   p a^{n + 2l } \,  \Gamma(p+l)  \Gamma \left( \frac{m + n + 2l + 1}{2}, \frac{b^{2}}{2} \right) }{  l!p^{2l}  (n+l)! 2^{\frac{n - m + 2l - 1}{2}}\,  (p-l)! e^{\frac{a^{2}}{2}}  }  
\end{equation}

\noindent
where,  

\begin{equation} \label{Nuttall_Coefficient}
\begin{split}
\mathcal{I}^{k}_{m, n}&(a,b) = \sum_{l = 0}^{m - n + k}  \binom{m - n + k}{l}  (-1)^{k} 2^{\frac{l - 1}{2}}   a^{ k  + m - l} \\
& \times   \left\lbrace (-1)^{m - n - l - 1} \Gamma \left( \frac{l + 1}{2}, \frac{(b + a)^{2}}{2} \right)  -   [{\rm sgn}(b - a)]^{ l + 1}\gamma \left( \frac{l + 1}{2}, \frac{(b - a)^{2}}{2} \right)  +  \Gamma\left( \frac{l + 1}{2} \right)   \right\rbrace 
\end{split}
\end{equation} 

\noindent
where $\gamma(a, x)$ is the lower incomplete Gamma function and 

\begin{equation}
\lceil x \rceil_{0.5} \triangleq  \lceil x - 0.5 \rceil + 0.5
\end{equation}
 with  $\lceil . \rceil$ denoting the integer ceiling function.  
\end{lemma}

\begin{proof}
The truncation error of \eqref{Nuttall_Polynomial} is expressed by definition as follows:

\begin{align}  \label{Tr_e_1}
\epsilon_{t}  &=  \sum_{l = p+1}^{\infty} \frac{a^{n + 2l } \,  \Gamma(p+l) p^{1 - 2l} \, \Gamma \left( \frac{m + n + 2l + 1}{2}, \frac{b^{2}}{2} \right) }{  l! (n+l)! 2^{\frac{n - m + 2l + 1}{2}}\, (p-l)!  e^{ \frac{a^{2}}{2}}   } \\
&=   \underbrace{\sum_{l = 0}^{\infty} \frac{a^{n + 2l }   \Gamma(p+l) p^{1 - 2l} \, \Gamma \left( \frac{m + n + 2l + 1}{2}, \frac{b^{2}}{2} \right) }{  l! \, e^{\frac{a^{2}}{2}} (n+l)! 2^{\frac{n - m + 2l + 1}{2}}\,  (p-l)! } }_{\mathcal{I}_{1}} - \sum_{l = 0}^{p} \frac{a^{n + 2l } \Gamma(p+l) p^{1 - 2l} \, \Gamma \left( \frac{m + n + 2l + 1}{2}, \frac{b^{2}}{2} \right) }{  l! (n+l)! \, e^{ \frac{a^{2}}{2}}   2^{\frac{n - m + 2l + 1}{2}}\,  (p-l)! }. 
\end{align}

\noindent 
Given that the $\mathcal{I}_{1}$ series is infinite and based on the proposed series in \cite{J:Gross_2},  the terms  

$$ \frac{\Gamma(p + l) p^{1 - 2l}}{ \Gamma(p - l + 1)}$$
 vanish which yields, 

\begin{align} \label{Tr_e_2}
\mathcal{I}_{1} &= \sum_{l = 0}^{\infty} \frac{a^{n + 2l } \, e^{- \frac{a^{2}}{2}}  \, \Gamma \left( \frac{m + n + 2l + 1}{2}, \frac{b^{2}}{2} \right) }{  l! \Gamma(n+l+1) 2^{\frac{n - m + 2l + 1}{2}} } \\
&= Q_{m, n}(a, b).  
\end{align}

\noindent
It is recalled here that according to \cite[eq. (19)]{J:Karagiannidis},

\begin{equation} \label{Nuttall_UB}
Q_{m, n} (a, b) \leq Q_{\lceil m\rceil_{0.5}, \lceil n \rceil_{0.5}} (a, b)
\end{equation}
Therefore, by substituting \eqref{Nuttall_UB} in \eqref{Tr_e_2} and  then in \eqref{Tr_e_1} one obtains the following inequality, 

\begin{equation} \label{Tr_e_3}
\epsilon_{t}   \leq Q_{\lceil m\rceil_{0.5}, \lceil n \rceil_{0.5}} (a, b)    - \sum_{l = 0}^{p} \frac{a^{n + 2l } \, e^{- \frac{a^{2}}{2}}  \Gamma(p+l) p^{1 - 2l} \, \Gamma \left( \frac{m + n + 2l + 1}{2},\frac{b^{2}}{2} \right) }{  l! \Gamma(n+l+1) 2^{\frac{n - m + 2l + 1}{2}}\,  \Gamma(p-l+1) }. 
\end{equation}

\noindent
The upper bound for the $ Q _{m, n} (a, b)$ function can be expressed in closed-form with the aid of   \cite[Corollary $1$]{J:Karagiannidis}. Therefore, by substituting  in \eqref{Tr_e_3} yields \eqref{Nuttall_Truncation}, which completes the proof.    
\end{proof}

\begin{remark} 

For the specific case that $n = m - 1$ and given that $\mathcal{Q}_{m,m-1}(a,b) = Q_{m}(a,b)$, the following upper bound is obtained for the truncation error of the Marcum $Q{-}$function representations in \eqref{Marcum_Polynomial} and \eqref{Marcum_Integer_Indices}, 

\begin{equation} \label{Truncation_Marcum_Polynomial}
\begin{split}
\epsilon_{t} \leq &    \sum_{l=1}^{m - \frac{1}{2}} \sum_{k=0}^{l-1} \frac{(-1)^{l} b^{k} (l-k)_{l-1}   \left[1 - (-1)^{k} e^{2ab} \right] }{k! \sqrt{\pi} 2^{ l -2k - \frac{1}{2}} a^{2l - k - 1}e^{   \frac{(a + b)^{2}}{2}} }  + Q(b+a)  \\
 &+ Q(b-a)  -  \sum_{l = 0}^{p} \frac{a^{2l } \,   \Gamma(p+l) p^{1 - 2l} \, \Gamma \left( m + l, \frac{b^{2}}{2} \right) }{ l!  \Gamma(m+l) 2^{l}\,  (p-l)! e^{ \frac{a^{2}}{2}} }  
\end{split} 
\end{equation}
where $Q(x)$ denotes the one dimensional Gaussian $Q{-}$function \cite{B:Abramowitz}. By  following the same methodology as in Lemma $1$, a respective upper bound can be also deduced for the truncation error of the infinite series  in Remark $1$, namely, 

\begin{equation}  \label{Nuttall_Truncation} 
\epsilon_{t}   \leq \sum_{k = 0}^{\lceil n \rceil_{0.5} - 1}\frac{ (-1)^{\lceil n \rceil_{0.5}}  \Gamma(2\lceil n \rceil_{0.5} - k - 1) \mathcal{I}_{\lceil m\rceil_{0.5},\lceil n \rceil_{0.5}}^{k} (a, b) }{k! \Gamma(\lceil n \rceil_{0.5} - k) (2a)^{-k} \sqrt{\pi} 2^{ \lceil n \rceil_{0.5} - \frac{1}{2}} a^{2\lceil n \rceil_{0.5} - 1} }  - \sum_{l = 0}^{p} \frac{     a^{n + 2l } \,    \Gamma \left( \frac{m + n + 2l + 1}{2}, \frac{b^{2}}{2} \right) }{  l!   (n+l)! 2^{\frac{n - m + 2l - 1}{2}}\,    e^{\frac{a^{2}}{2}}  }.   
\end{equation}  
\end{remark}

\subsection{A Tight Upper Bound and Approximation}


\begin{proposition}
For $a,b,m,n \in \mathbb{R}^{+}$ and for the special cases that either $b \rightarrow 0$ or $\, a, m, n \geq \frac{3}{2} b$, the following closed-form upper bound for the  Nuttall $Q{-}$function is valid,

\begin{equation} \label{Nuttall_Novel_Bound}
Q_{m,n}(a,b) \leq \frac{a^{n} \Gamma\left( \frac{m + n + 1}{2} \right)  }{n!\,2^{\frac{n - m + 1}{2}}\,e^{\frac{a^{2}}{2}} } \,_{1}F_{1} \left( \frac{m + n + 1}{2}, n + 1, \frac{a^{2}}{2} \right) 
\end{equation}
which becomes an accurate  approximation when $a, m, n \geq \frac{5}{2} b$.
%
\end{proposition}

\begin{proof}
Given that \eqref{Nuttall_Polynomial} becomes an exact infinite series as $p \rightarrow \infty$ and  with the aid of the monotonicity property  $\Gamma(a, x) \leq \Gamma(x)$, $a \in \mathbb{R}^{+}$, the $ Q_{m,n}(a,b)$ can be   upper bounded as follows, 

\begin{equation} \label{Nuttall_Approximation_1}
Q_{m,n}(a,b)  \leq   \underbrace{\sum_{l = 0}^{\infty} \frac{a^{n+2l } \, e^{- \frac{a^{2}}{2}}  \, \Gamma \left( \frac{m + n + 2l + 1}{2} \right) }{  l! \Gamma(n+l+1) 2^{\frac{n - m + 2l + 1}{2}} }}_{\mathcal{I}_{2}}.   
\end{equation}

\noindent
By expressing each Gamma function as,

\begin{equation}
\Gamma(x+l) = (x)_{l} \Gamma(x)  
\end{equation}
and carrying out some algebraic manipulations one obtains,  

\begin{equation} \label{Nuttall_Approximation_2}
\mathcal{I}_{2} = \frac{\Gamma\left( \frac{m + n + 1}{2} \right) e^{\frac{-a^{2}}{2}}  }{n! 2^{\frac{n - m + 1}{2}} } \sum_{l = 0}^{\infty} \frac{\left( \frac{m + n + 1}{2} \right)_{l} a^{n + 2l}  }{ \, l!\,  (n + 1)_{l} 2^{l} }.  
\end{equation}

\noindent
The above infinite series can be expressed in terms of  Kummer's  hypergeometric function in \cite[eq. (9.14.1)]{B:Tables}. Therefore, by performing the necessary change of variables and substituting \eqref{Nuttall_Approximation_2} into \eqref{Nuttall_Approximation_1} yields \eqref{Nuttall_Novel_Bound}, which completes the proof.     
\end{proof}
\noindent
It is noted here that similar expressions to the $Q_{m,n}(a,b)$ function can be also deduced for the $\mathcal{Q}_{m,n}(a,b)$ function by applying the identity 

\begin{equation} \label{Identity_New}
 Q_{m,n}(a,b) = a^{n} \mathcal{Q}_{m,n}(a,b) 
 \end{equation}
which corresponds to dividing  equations \eqref{Nuttall_Polynomial}, \eqref{Nuttall_Integer_Indices}, \eqref{Nuttall_Truncation}, and \eqref{Nuttall_Novel_Bound}  by $a^{n}$.

%
%
%
%
%
%
\begin{table}[h!]
\centering
\caption{Accuracy of  proposed expressions for the $Q_{m, n}(a, b)$ function. } 
\begin{tabular}{|c|c|c|c|c|}
\hline FUNCTION&EXACT&Eq. \eqref{KdF_1}&Eqs. \eqref{Nuttall_Polynomial}, \eqref{Nuttall_Integer_Indices}&Eqs. \eqref{Nuttall_Novel_Bound} \\ 
\hline \hline $ Q_{0.7,0.3}(0.6,0.4)$&$0.6956$&$0.6956$&$0.6956$&$0.7458$  \\ 
\hline $Q_{1.6,1.4}(0.6,0.4)$&$0.2890$&$0.2890$& $0.2890$&$0.2898$  \\ 
\hline $ Q_{1.2,1.8}(0.6,0.4)$&$0.1295$&$0.1295$& $0.1295$&$0.1299$  \\ 
\hline $Q_{0.7,0.3}(0.9,0.4)$&$0.7580$&$0.7580$& $0.7580$&$0.8035$  \\ 
\hline $Q_{1.6,1.4}(0.6,1.3)$&$0.2360$&$0.2360$& $0.2360$&$0.2898$  \\ 
\hline $Q_{1.2,1.8}(2.0,2.0)$&$0.5380$&$0.5380$& $0.5380$&$0.7403$  \\ 
\hline  
\end{tabular} 
\end{table}
The behaviour of the offered results is depicted in Table I along with respective results from numerical integrations for comparisons. The polynomial series was truncated after $20$ terms and one can notice the excellent agreement between analytical and numerical results. This is also verified by the value of the corresponding absolute relative error,

\begin{equation}
\epsilon_{r} \triangleq  \, \frac{\mid  Q_{m,n}(a,b) - \tilde{Q}_{m,n}(a,b) \mid }{ Q_{m,n}(a,b)} 
\end{equation}
 which is typically smaller than $\epsilon_{r} < 10^{-11}$. It is also shown that   \eqref{Nuttall_Novel_Bound}   appears to be rather accurate particularly for high values of $a$.

\begin{figure}[h!] 
\centering
\subfigure[$\mathcal{Q}_{m,n}(a,b)$ in \eqref{KdF_1} and \eqref{Nuttall_Polynomial} ] 
{ \includegraphics[width=12cm, height=9cm]{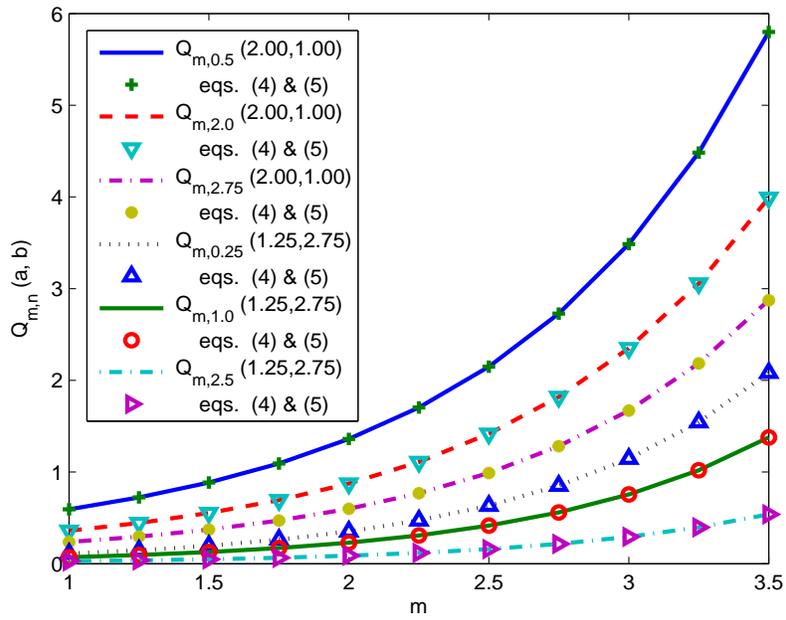} }
\subfigure[$\mathcal{Q}_{m,n}(a,b)$ in \eqref{Nuttall_Novel_Bound}]  
{ \includegraphics[width=12cm, height=9cm]{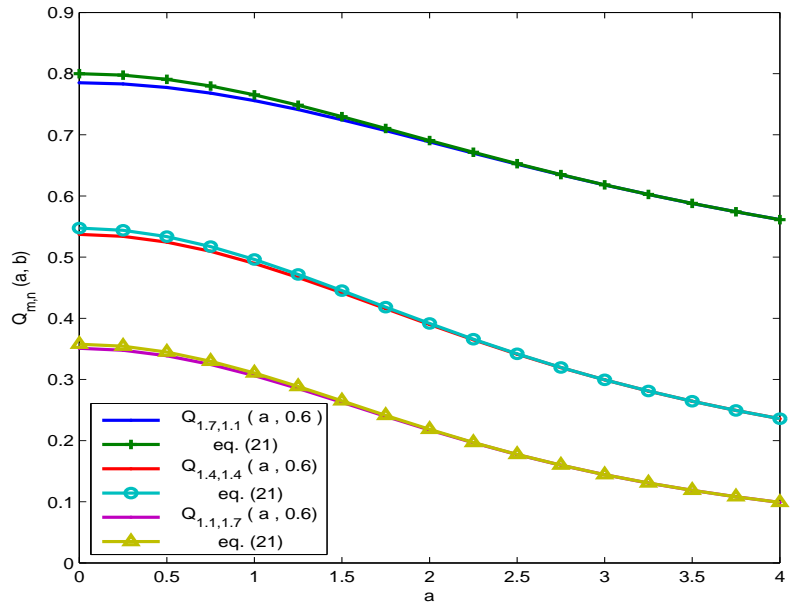} }
\caption{\small Behaviour and accuracy of the normalized Nuttall $Q{-}$function in the proposed equations \eqref{KdF_1}, \eqref{Nuttall_Polynomial} $\&$ \eqref{Nuttall_Novel_Bound}. }
\end{figure}
The behavior of \eqref{KdF_1} and \eqref{Nuttall_Polynomial} is also illustrated in Fig. $1{\rm a}$ for arbitrary values of the involved parameters   whereas Fig. $1{\rm b}$ depicts the accuracy of \eqref{Nuttall_Novel_Bound}. It is clearly observed that \eqref{Nuttall_Novel_Bound} upper bounds the $\mathcal{Q}_{m,n}(a,b)$  tightly and becomes a rather accurate approximation as $a$ increases asymptotically. Moreover, both \eqref{Nuttall_Polynomial} and \eqref{Nuttall_Novel_Bound} are tighter than the closed-form bounds proposed in \cite{J:Karagiannidis}, since they are in adequate match with the respective theoretical results for most cases.

\section{New Closed-Form Expressions for the Incomplete Toronto Function }

\begin{definition}
For $m, n, r, B \in \mathbb{R}^{+}$, the incomplete Toronto function is defined as follows, 

\begin{equation} \label{ITF_Definition}
T_{B}(m,n,r) \triangleq 2r^{n-m+1} e^{-r^{2}} \int_{0}^{B} t^{m-n} e^{-t^{2}}I_{n}(2rt) {\rm d}t. 
\end{equation}
\end{definition}

\noindent 
The ITF has been also  a useful special function in wireless communications. When $B \rightarrow \infty$,  it reduces to the (complete) Toronto function, 

\begin{equation} 
T(m,n,r) \triangleq 2r^{n-m+1} e^{-r^{2}} \int_{0}^{\infty} t^{m-n} e^{-t^{2}}I_{n}(2rt) {\rm d}t 
\end{equation}
while for the specific case that $n = (m-1){/}2$ it is expressed in terms of the Marcum $Q{-}$function
namely, 

\begin{equation} \label{ITF_Marcum}
T_{B}\left(m,\frac{m-1}{2},r\right) = 1 - Q_{\frac{m+1}{2}}\left(r\sqrt{2},B\sqrt{2}\right).
\end{equation}
Alternative representations  include two infinite series in \cite{J:Sagon}; however, to the best of the Authors' knowledge no study has been reported in the open technical literature for the   convergence and truncation of these series.

\subsection{Special Cases}

\begin{theorem}

For  $r \in \mathbb{R}$, $B \in \mathbb{R}^{+}$, $m \in \mathbb{N}$, $n \pm 0.5 \in \mathbb{N}$ and $m > n$, the following closed-form expression is valid for the incomplete Toronto function, 

\begin{equation}  \label{ITF_Closed}
\begin{split}
T_{B}(m, n, r) &= \sum_{k = 0}^{n - \frac{1}{2}} \sum_{l = 0}^{L} \frac{ \Gamma \left(n + k + \frac{1}{2}\right) \, (L - k)! 2^{-2k}  r^{-2k -l)}  }{2 \sqrt{\pi} \, k! \, l! \Gamma\left( n - k + \frac{1}{2} \right) (L-k-l)!  }  \\
& \quad  \times  \left\lbrace  (-1)^{m - l}  \gamma \left[ \frac{l + 1}{2}, (B + r)^{2} \right]  + (-1)^{k}  \gamma \left[ \frac{l + 1}{2}, (B- r)^{2} \right] \right\rbrace   
\end{split}
\end{equation}
where $L = m - n - \frac{1}{2}$.

\end{theorem}

\begin{proof}
The proof is provided in Appendix C. 
\end{proof}

In the same context, a similar closed-form expression can be derived for the case that  $m - 2n$ is an odd positive integer. To this end, it is firstly essential to algebraically link the incomplete Toronto function with the Nuttall $Q{-}$function, which  is provided in the following Lemma. 

\begin{lemma}
For $m, n, r \in \mathbb{R}$ and $B \in \mathbb{R}^{+}$, the $T_{B}(m,n,r)$ function can be algebraically related to the $Q_{m,n}(a, b)$ function   by the following representation, 

\begin{equation} \label{Nuttall_Toronto}
T_{B}(m, n, r)   =  \frac{\Gamma\left( \frac{m + 1}{2} \right) \, _{1}F_{1} \left( n + \frac{1 - m}{2}, n + 1, - r^{2} \right) }{n! r^{m - 2n - 1}} -  r^{n - m +1} 2^{ \frac{ n - m + 1}{2}}  Q_{m - n, n} (\sqrt{2} r, \sqrt{2} B)
\end{equation}  
\end{lemma}

\begin{proof}
The proof is provided in Appendix D. 
\end{proof}
Lemma $2$ is subsequently employed in the proof of the following theorem.

\begin{theorem}
For $r \in \mathbb{R}$, $B \in \mathbb{R}^{+}$, $m \in \mathbb{Z}^{+}$, $n \in \mathbb{N}$, $m >2n$ and $\frac{ m}{2} - n \notin \mathbb{N}$ , the following closed-form expression is valid for the incomplete Toronto function, 

\begin{equation}  \label{m+2n_1}
\begin{split}
T_{B}(m, n, r) &=   \frac{ \Gamma\left( \frac{m + 1}{2} \right)  }{   r^{m - 2n - 1} n! }  \, _{1}F_{1} \left( n + \frac{1 - m}{2}, n + 1, - r^{2} \right) \\
  & -  \sum_{l = 1}^{\frac{m-1}{2} - n} \sum_{j = 0}^{\frac{m-1}{2} - n - l} \frac{r^{n + l} 2^{2l + 2j} \Gamma\left( \frac{m+1}{2} \right) }{\Gamma(l) r^{m} \left( \frac{m-1}{2} - j - n \right)_{1 - l} }   \frac{B^{n + l + 2j + 1} I_{n+l-1}(2 r B)}{\Gamma(n + l + j + 1)}  \\
& - \sum_{l = 1}^{\frac{m+1}{2} - n} \frac{r^{2n + 2l - m - 1} \Gamma\left( \frac{m+1}{2} - n \right)}{\Gamma(l) \Gamma(n + l) \left( \frac{m+1}{2}\right)_{1 - n - l} } \left\lbrace Q_{1}\left(\sqrt{2}r, \sqrt{2}B \right) + \sum_{i=1}^{n+l-1} \frac{b^{i} I_{i}(2rB)}{r^{i} e^{r^{2} + B^{2}}} \right\rbrace   
\end{split}
\end{equation}
where $Q_{1}(a,b) = Q(a,b) $ denotes  the Marcum $Q{-}$function of the first order. 
\end{theorem}  

\begin{proof}
By utilizing \eqref{Nuttall_Toronto} and \cite[eq. (8)]{J:Simon_2} and after   basic algebraic manipulations it  follows that, 

\begin{equation} \label{m+2n_2}
\begin{split} 
T_{B}(m, n, r) =& \frac{\Gamma\left( \frac{m + 1}{2} \right)\, _{1}F_{1}\left(n + \frac{1 - m}{2}, n + 1, - r^{2} \right)  }{n! r^{m - 2n - 1}}  - \sum_{l = 1}^{\frac{m+1}{2} - n} \frac{r^{2(n + l)} \Gamma\left( \frac{m + 1}{2} - n \right) Q_{n+l}\left( \sqrt{2}r, \sqrt{2}B \right) }{ r^{m + 1} 2^{ \frac{1 - l}{2}} \Gamma(l) \Gamma(n +l) \left( \frac{m+1}{2}\right)_{1 - n - l}} \\
& - \sum_{l = 1}^{\frac{m - 1}{2} - n} \sum_{j = 0}^{\frac{m - 1}{2} - n - l} \frac{\Gamma\left(\frac{m-1}{2} - j - n \right) b^{n + l + 2j + 1} }{r^{m - n - l}\Gamma(l) \Gamma(1 - n - l - j) }  \Gamma\left( \frac{m + 1}{2} \right) e^{-\frac{r^{2} + B^{2}}{2}} I_{n + l - 1}(2 r B). 
\end{split}
\end{equation}
Given that $n \in \mathbb{N}$, the $Q_{m}(a,b)$ function can be equivalently expressed in terms of the $ Q_{1}(a,b)$ function according to \cite[eq. (12)]{J:Simon_2}. To this effect, by performing the necessary variable transformation and substituting in \eqref{m+2n_2} yields \eqref{m+2n_1}, which completes the proof. 
\end{proof}

\subsection{ Closed-Form Bounds}

\begin{lemma}
For $m, n, B \in \mathbb{R^{+}}, r \in \mathbb{R}$ and $m \geq n$, the following  inequalities can serve as upper and lower bounds to the incomplete Toronto function,

\begin{equation} \label{ITF_CFUB}
T_{B}(m, n, r) \leq T_{B}(\lceil m \rceil, \lfloor n \rfloor_{0.5}, r   )
\end{equation}
and

\begin{equation} \label{ITF_CFLB}
T_{B}(m, n, r) \geq T_{B}(\lfloor m \rfloor, \lceil n \rceil_{0.5}, r   )  
\end{equation}

\noindent
where  $\lfloor . \rfloor$ denotes the integer floor function.  
\end{lemma}

\begin{proof}
Based on the  monotonicity properties of the Toronto function, $T_{B}(m, n, r)$ is strictly increasing  w.r.t  $m$ and strictly decreasing w.r.t   $n$. Furthermore, two half-integer rounding operators were given in \cite[eq. (18)]{J:Karagiannidis}, namely, 

\begin{equation}
\lfloor n \rfloor_{0.5} = \lfloor n + 0.5 \rfloor - 0.5
\end{equation}
and 

\begin{equation}
\lceil n \rceil_{0.5} =  \lceil n - 0.5 \rceil + 0.5. 
\end{equation}
%
%
By also recalling that \eqref{ITF_Closed}  is valid for $m \in \mathbb{N}$ and $n \pm 0.5 \in \mathbb{N}$, it follows that  $T_{B}(\lceil m \rceil, \lfloor n \rfloor_{0.5}, r )$ and $T_{B}(\lfloor m \rfloor, \lceil n \rceil_{0.5}, r   )$ can be expressed in closed-form for any value of $m$, $n$, $r$, $B$ and can hence serve as a closed-form bounds for $T_{B}(m,n,r)$. Thus, by applying the above floor and ceiling functions in \eqref{ITF_Closed}, equations \eqref{ITF_CFUB} and \eqref{ITF_CFLB} are obtained, which completes the proof. 
\end{proof}

\subsection{A Closed-Form Expression in Terms of the Kamp${\it \acute{e}}$ de F${\it \acute{e}}$riet Function }

A more generalized closed-form expression for the ITF, that does not impose any restrictions to the involved parameters, can be derived in terms of the KdF function. 

\begin{theorem}
For $m, n, r \in \mathbb{R} $, $B \in \mathbb{R}^{+}$ and $m + n > -1$,  the incomplete Toronto function can be expressed as follows, 

\begin{equation} \label{ITF/KdF_1}
T_{B}(m, n, r) =  \frac{2 r^{2n - m + 1} B^{m + 1}}{n! (m + 1) }  e^{r^{2}} F_{1, 1}^{1, 0} \left(^{ \frac{m + 1}{2} : - , - : }_{ \frac{m + 3}{2} : n + 1 , - : } r^{2} B^{2} , - B^{2} \right).  
\end{equation}  
\end{theorem}

\begin{proof}
The proof is provided in Appendix E. 
\end{proof}

\subsection{A Simple Polynomial Representation}

The proposed specific and generalized expressions  are rather  useful in  applications relating to wireless communications. However,  a relatively simple and  general representation for the ITF is additionally necessary for cases that the parameters  of the $T_{B}(m, n, r)$ are required to be  unrestricted  and  the algebraic representation of the function must be rather simple.

\begin{proposition}

For $m, n, r \in \mathbb{R}$ and $B \in \mathbb{R}^{+}$, the following  polynomial approximation holds for the $T_{B}(m,n,r)$ function,

\begin{equation} \label{ITF_Series}
T_{B}(m,n,r) \simeq \sum_{k = 0}^{p} \frac{\Gamma(p + k) r^{ 2(n + k) - m + 1} \gamma\left( \frac{m + 1}{2} + k, B^{2} \right) }{k! p^{2k - 1}\Gamma(p - k + 1) \Gamma(n + k + 1) e^{r^{2}} }. 
\end{equation}
\end{proposition}

\begin{proof}
The proof follows from Theorem $2$ and Proposition $1$ and with the aid of  \cite[eq. (19)]{J:Gross_2} and \cite[eq. (8.350.1)]{B:Tables}. 
\end{proof}

\begin{remark} 
By recalling that \cite[eq. (19)]{J:Gross_2} reduces to \cite[eq. (8.445)]{B:Tables} when $p \rightarrow \infty$, it  immediately follows that \eqref{ITF_Series} becomes an    infinite series representation as $p \rightarrow \infty$, namely,  

\begin{equation} \label{ITF_Series_exact}
T_{B}(m,n,r) = \sum_{k = 0}^{\infty} \frac{  r^{ 2(n + k) - m + 1}     }{k!  (n + k )! e^{r^{2}}  } \gamma\left( \frac{m + 1}{2} + k, B^{2} \right)  
\end{equation}
which is exact.

\end{remark}

\subsection{A Closed-Form Upper Bound for the Truncation Error} 

A tight upper bound for the truncation error of \eqref{ITF_Series} can be derived in closed-form. 

\begin{lemma}
For  $m, n, r \in \mathbb{R}$, $B \in \mathbb{R}^{+}$ and $m > n$ the following closed-form inequality can serve as an upper bound for the truncation error in \eqref{ITF_Series}, 

\begin{equation} \label{ITF_Truncation_1}
\begin{split}
\epsilon_{t} \leq& \sum_{k = 0}^{\lfloor n \rfloor _{0.5} - \frac{1}{2}} \sum_{l = 0}^{L} \frac{ r^{-(2k +l)} \, \left(\lfloor n \rfloor _{0.5} + k - \frac{1}{2}\right)! \, (L - k)!  }{\, k! \, l!\left( \lfloor n \rfloor _{0.5} - k - \frac{1}{2} \right)! (L-k-l)!} \left\lbrace \frac{\gamma \left[ \frac{l + 1}{2}, (B + r)^{2} \right]}{(-1)^{\lceil m \rceil  - l} \,2^{2k+1}} + \frac{\gamma \left[ \frac{l + 1}{2}, (B- r)^{2} \right]}{(-1)^{k} \,  2^{2k+1}} \right\rbrace     \\
&   - \sum_{k = 0}^{p} \frac{\Gamma(p + k) r^{ 2(n + k) - m + 1} \gamma\left( \frac{m + 1}{2} + k, B^{2} \right) }{k! p^{2k - 1}\Gamma(p - k + 1) \Gamma(n + k + 1) e^{r^{2}} }.  
\end{split}
  \end{equation}
\end{lemma}

\begin{proof}
Since the corresponding truncation error is expressed as 

\begin{align}
\epsilon_{t} &\triangleq \sum_{p + 1}^{\infty} f(x)  \\
& = \sum_{l = 0}^{\infty} f(x) - \sum_{l=0}^{p} f(x)
\end{align}
and given that \eqref{ITF_Series} reduces to an exact infinite series as $p \rightarrow \infty$, it follows that,

\begin{equation} \label{ITF_Truncation_2}
\epsilon_{t} = \underbrace{\sum_{k = 0}^{\infty} \frac{r^{ 2(n + k) - m + 1} \gamma\left( \frac{m + 1}{2} + k, B^{2} \right) }{k!  \Gamma(n + k + 1) e^{r^{2}} } }_{\mathcal{I}_{7}}   - \sum_{k = 0}^{p} \frac{\Gamma(p + k) r^{ 2(n + k) - m + 1} \gamma\left( \frac{m + 1}{2} + k, B^{2} \right) }{k! p^{2k - 1}\Gamma(p - k + 1) \Gamma(n + k + 1) e^{r^{2}} }. 
\end{equation}

\noindent
It is noted that, 

\begin{equation}
\mathcal{I}_{7} = T_{B}(m, n, r). 
\end{equation}
To this effect  and with the aid of  \eqref{ITF_CFUB}, the $\epsilon_{t}$ can be upper bounded as follows:

\begin{equation} \label{ITF_Truncation_3}
\epsilon_{t}  \leq T_{B}(\lceil m \rceil, \lfloor n \rfloor_{0.5}, r)  - \sum_{k = 0}^{p} \frac{\Gamma(p + k) r^{ 2(n + k) - m + 1} \gamma\left( \frac{m + 1}{2} + k, B^{2} \right) }{k! p^{2k - 1}\Gamma(p - k + 1) \Gamma(n + k + 1) e^{r^{2}} }. 
\end{equation} 

\noindent
It is recalled here that  the $ T_{B}(\lceil m \rceil, \lfloor n \rfloor_{0.5}, r)$ function can be expressed in closed-form according to \eqref{ITF_Closed}. Therefore, by substituting in \eqref{ITF_Truncation_3} one obtains \eqref{ITF_Truncation_1}, which completes the proof.  
\end{proof}

\begin{remark}
By omitting the coefficients  

$$\frac{\Gamma(p + k) p^{1 - 2k}}{\Gamma(p - k + 1)}$$
 in the second term of \eqref{ITF_Truncation_3}  as $p \rightarrow \infty $, a closed-form upper bound can be  deduced for the truncation error of the   infinite series in \eqref{ITF_Series_exact}, namely,

\begin{equation} \label{ITF_Truncation_1}
\begin{split}
\epsilon_{t} \leq& \sum_{k = 0}^{\lfloor n \rfloor _{0.5} - \frac{1}{2}} \sum_{l = 0}^{L} \frac{ r^{-(2k +l)} \, \left(\lfloor n \rfloor _{0.5} + k - \frac{1}{2}\right)! \, (L - k)!  }{\, k! \, l!\left( \lfloor n \rfloor _{0.5} - k - \frac{1}{2} \right)! (L-k-l)!}   \left\lbrace \frac{\gamma \left[ \frac{l + 1}{2}, (B + r)^{2} \right]}{(-1)^{\lceil m \rceil  - l} \,2^{2k+1}} + \frac{\gamma \left[ \frac{l + 1}{2}, (B- r)^{2} \right]}{(-1)^{k} \,  2^{2k+1}} \right\rbrace   \\
& \qquad - \sum_{k = 0}^{p} \frac{ r^{ 2(n + k) - m + 1} \gamma\left( \frac{m + 1}{2} + k, B^{2} \right) }{k!  \Gamma(n + k + 1) e^{r^{2}} }. 
\end{split} 
  \end{equation}

\end{remark}

\subsection{A Tight Closed-form Upper Bound and Approximation}

Capitalizing on the algebraic representation of the $T_{B}(m, n, r)$ function, a simple closed-form upper bound is proposed which in certain cases becomes an accurate approximation.

\begin{proposition}
For $m, n, r \in \mathbb{R}$, $B \in \mathbb{R}^{+}$ and $m, n, r \leq \frac{B}{2}$, the following inequality holds,

\begin{equation} \label{ITF_Appr_1}
T_{B}(m, n, r) \leq  \frac{ \Gamma\left( \frac{m + 1}{2} \right)  \, _{1}F_{1} \left( \frac{m + 1}{2}, n + 1, r^{2} \right)    }{  r^{m - 2n - 1  }     \Gamma(n + 1) e^{r^{2}} }  
\end{equation}
which when $m, n, r \leq 2B$, it can serve as a tight closed-form approximation. 
%
\end{proposition}

\begin{proof}
The proof follows from Proposition 2  and with the aid of the monotonicity  identity $\gamma(a) \geq  \gamma(a,x)$. 
\end{proof}

\begin{table*} 
\centering
\caption{Accuracy of  proposed expressions for $T_{B}(m, n, r)$} 
\begin{tabular}{|c|c|c|c|c|}
\hline FUNCTION&EXACT&Eqs. \eqref{ITF_Closed}, \eqref{m+2n_1}, \eqref{ITF/KdF_1}&Eq. \eqref{ITF_Series}&Eq.   \eqref{ITF_Appr_1} \\ 
\hline \hline $T_{3}(2.0,0.5,2.0)$&$0.8695$&$0.8695$, n/a, $0.8695$&$0.8695$&$1.000$  \\ 
\hline $T_{3}(3.0,1.5,2.0)$&$0.7554$&$0.7554$, n/a, $0.7554$&$0.7554$&$0.8761$  \\ 
\hline $T_{5}(2.0,0.5,2.0)$&$0.9999$&$0.9999$, n/a, $0.9999$&$0.9999$&$1.0000$   \\ 
\hline $T_{5}(3.0,1.5,2.0)$&$0.8760$&$0.8760$, n/a, $0.8760$&$0.8760$&$0.8761$  \\ 
\hline $T_{4}(3.0,1.0,2.0)$&$0.9930$& n/a, $0.9930$, $0.9930$&$0.9930$&$1.000$  \\ 
\hline $T_{4}(5.0,2.0,2.0)$&$0.9865$& n/a, $0.9865$, $0.9865$&$0.9865$&$1.000$  \\ 
\hline  
\end{tabular} 
\end{table*}

The accuracy of the offered  expressions is demonstrated in Table II (top of the next page). The exact formulas are in full agreement with the respective numerical results which is also the case for the proposed polynomial approximation for truncation after $20$ terms. The corresponding relative error for \eqref{ITF_Series} is rather low, as it is typically $\epsilon_{r} < 10^{-5}$. Likewise,  \eqref{ITF_Appr_1}  is shown to be relatively tight while the overall involved relative error is proportional to the value of $r$ and is $\epsilon_{r} < 10^{-6}$ when $r < 1$.  
Figure $2$a also illustrates the behaviour of  \eqref{ITF_Closed}, \eqref{ITF/KdF_1} and \eqref{ITF_Series} along with respective numerical results while \eqref{ITF_Appr_1} is  depicted in Fig. $2$b   for three different scenarios.  It is evident that the  analytical results match their numerical counterparts in all cases, which indicates the  accuracy of the proposed expressions. 
\begin{figure}[h!] 
\centering 
\subfigure[$T_{B}(m,n,r)$ in \eqref{ITF_Closed}, \eqref{ITF/KdF_1} $\&$ \eqref{ITF_Series}] 
{ \includegraphics[width=12cm, height=9.0cm]{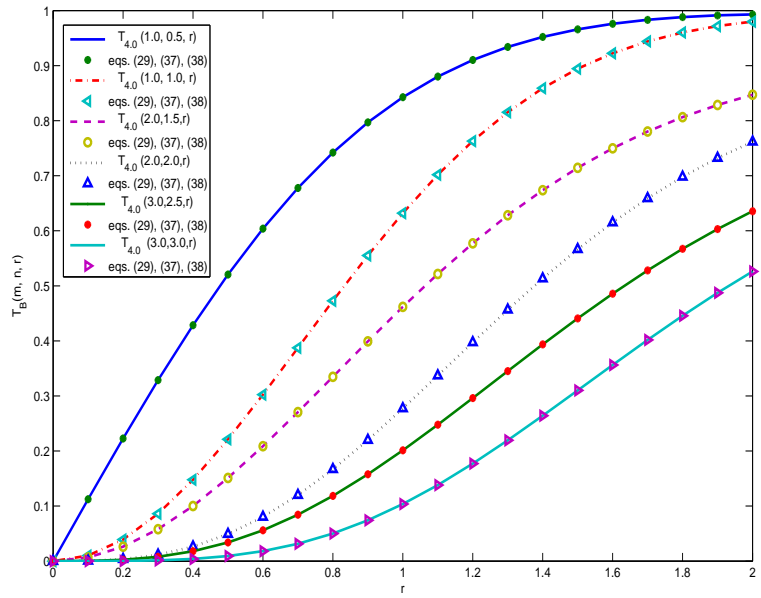} }
\subfigure[$T_{B}(m,n,r)$ in \eqref{ITF_Appr_1}] 
{ \includegraphics[width=12cm, height=9.0cm]{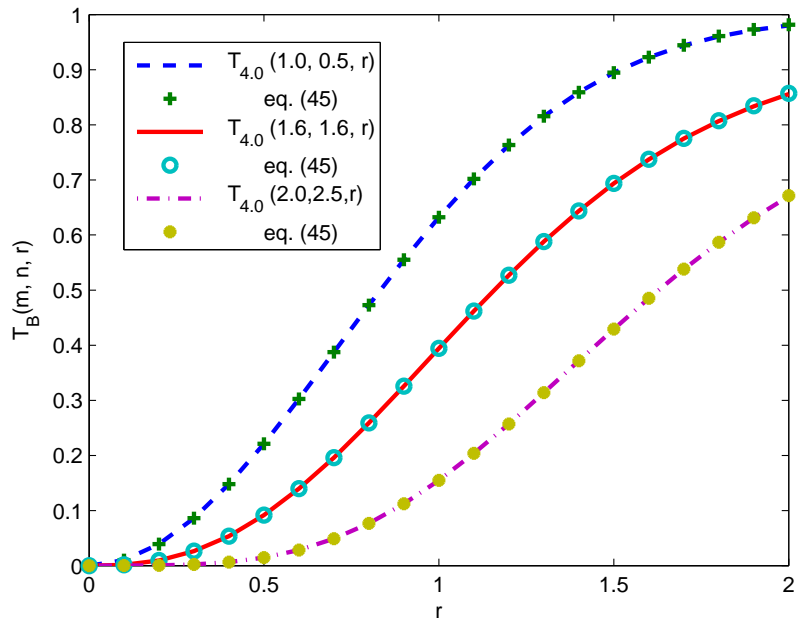} }
\caption{\small Behaviour and accuracy of the incomplete Toronto function  in equations  \eqref{ITF_Closed}, \eqref{ITF_Series} $\&$ \eqref{ITF_Appr_1}.  }
 \end{figure} 
%

\section{New Closed-Form Results for the Rice $Ie{-}$Function}

 \begin{definition}
 For $x \in \mathbb{R}^{+}$ and $0 \leq k \leq 1$, Rice $Ie{-}$function is defined by the following finite integral representation \cite{J:Tan, J:Pawula},

 \begin{equation}  \label{Rice_Definition}
Ie(k,x) \triangleq \int_{0}^{x} e^{-t} I_{0}(kt) {\rm d}t.   
\end{equation}
 \end{definition}
\noindent
The above representation can be alternatively expressed in trigonometric form as \cite{J:Tan}, 

\begin{equation} \label{Rice_Definition_2}
Ie(k,x) = \frac{1}{\sqrt{1 - k^{2}}} - \frac{1}{\pi} \int_{0}^{\pi} \frac{e^{-x(1 -k \rm cos {\theta})}}{1 - k \rm cos {\theta}} {\rm d}\theta. 
\end{equation}

An analytic expression in terms of the Marcum $Q{-}$function as well as two alternative series representations were reported in \cite{J:Pawula, J:Tan}. These series are infinite and are expressed in terms of the modified Struve function, $L_{n}(.)$, and the $\Gamma(.)$, $I_{n}(.)$ functions, respectively \cite{B:Tables}. Furthermore, they were shown to be complementary to each-other as \cite[eq. (2)]{J:Tan} converges relatively quickly when $x\sqrt{1 - k^{2}}$ is large and $kx$ is small, whereas \cite[eq. (3)]{J:Tan} converges relatively quickly when $x\sqrt{1 - k^{2}}$ is small and $kx$ is large. Therefore, it appears that   utilizing these series is rather inconvenient both analytically and numerically for the following reasons: $i)$ two infinite series  are required for computing the $Ie(k, x)$ function;  $ii)$ the $L_{n}(.)$ function is neither tabulated nor built-in in widely used mathematical software packages.

\subsection{Closed-form Upper and Lower Bounds}

The lack of simple expressions for the $Ie(k, x)$ function constitutes the derivation of tight upper and lower bounds advantageous. To this end, it is critical to primarily express $Ie(k, x)$ function alternatively.

\begin{lemma}
 For $x \in \mathbb{R}^{+}$ and $0 \leq k \leq 1$, the following  analytic representation is valid, 

\begin{equation} \label{Rice_Alternative}
 Ie(k,x) = 1 - e^{-x}I_{0}(kx) + k \int_{0}^{x} e^{-t} I_{1}(kt){\rm d}t. 
\end{equation}

\end{lemma}

\begin{proof}
By integrating \eqref{Rice_Definition} by parts one obtains,

\begin{equation} \label{Rice_Alternative_2}
Ie(k,x) = \left[ \int_{0}^{x} e^{-t} {\rm d}t \right] I_{0}(kt) - \int_{0}^{x} \left[\int_{0}^{x} e^{-t} {\rm d}t \right] \frac{{\rm d}\,  I_{0}(kt)}{{\rm d}t}{\rm d}t. 
\end{equation}

\noindent
Based on the basic principles of integration it follows straightforwardly that, 

\begin{equation} \label{Rice_Alternative_3a}
\int_{0}^{x} e^{-t} {\rm d}t = 1 - e^{-x}. 
\end{equation}
By also recalling that 

\begin{equation} \label{Rice_Alternative_3}
\frac{{\rm d}}{{\rm d}x} I_{n}(kx) =\frac{k}{2} \left[ I_{n-1}(kx) + I_{n+1}(kx) \right]  
\end{equation}

\noindent
and 

\begin{equation}
I_{-1}(x) \triangleq I_{1}(x)
\end{equation}
it follows that 

\begin{equation}
\frac{{\rm d}\, I_{0}(kt)}{{\rm d}t} = k\,I_{1}(kt). 
\end{equation}
Therefore, by substituting accordingly in \eqref{Rice_Alternative_2} one obtains  \eqref{Rice_Alternative}, which completes the proof. 
\end{proof}
Capitalizing on Lemma $5$, we derive closed-form upper and lower bounds for the $Ie(k, x)$ function.

\begin{theorem}
For $x \in \mathbb{R}^{+}$ and $0 \leq k \leq 1$, the following inequalities can serve as   upper and lower bounds for the Rice $Ie{-}$function, 

\begin{equation} \label{Rice_UB_1}
Ie(k,x) < 1  + \frac{\sqrt{k}}{\sqrt{2 (1 - k)}}      + \frac{\sqrt{2k}Q(\sqrt{2x}\sqrt{1+k})}{\sqrt{1+k}} -  \frac{ I_{0}(kx)}{e^{x}}   - \frac{\sqrt{k}}{\sqrt{2 (1 + k)}} - \frac{\sqrt{2k}Q(\sqrt{2x}\sqrt{1-k})}{\sqrt{1-k}} 
\end{equation}
and\footnote{Eq. \eqref{Rice_UB_1} can be also expressed in terms of the error function, ${\rm erf}(x)$, and the complementary error function, ${\rm erfc}(x) = 1 - {\rm erf}(x)$ with the aid of the identities: $Q(x) \triangleq  \frac{1}{2} {\rm erfc}\left(\frac{x}{\sqrt{2}}\right) = \frac{1}{2} - \frac{1}{2}{\rm erf}\left(\frac{x}{\sqrt{2}}\right)$.}

\begin{equation} \label{Rice_LB_1}
Ie(k,x) > \frac{ 2Q( b + a ) + 2Q( b - a ) - e^{-x} I_{0}(kx) - 1}{\sqrt{1-k^{2}}} 
\end{equation}
\noindent 
where   

\begin{equation}
a=\sqrt{x}\sqrt{1 + \sqrt{1 - k^{2}}}
\end{equation}
 and
 
\begin{equation}
b=\sqrt{x}\sqrt{1 - \sqrt{1 - k^{2}}}.  
\end{equation}
\end{theorem} 

\begin{proof}
The proof is provided in Appendix F. 
\end{proof}

\begin{remark}
 The authors in \cite{J:Karagiannidis} derived closed-form bounds for the $Q_{m}(a,b)$ function. By performing the necessary change of variables and substituting accordingly in \cite[eq. (2c)]{J:Pawula},    an alternative closed-form upper bound can be obtained. However, the algebraic representation of such a bound is significantly less compact and less convenient than \eqref{Rice_UB_1} both analytically and numerically. Likewise, a lower bound for the $Ie(k,x)$ function could be theoretically derived by following the same methodology as in Theorem 2. Nevertheless, this approach  leads to an integral representation whose analytic solution is divergent. 
\end{remark}

\subsection{A Closed-form Expression in terms of  Humbert Function}

\begin{theorem}
For $0 \leq k \leq 1$ and $x \in \mathbb{R^{+}}$, the following   expression is valid for the $Ie(k,x)$  function, 

\begin{equation} \label{Rice_Humbert}
Ie(k, x) = \frac{1}{\sqrt{1 - k^{2}}} - \frac{e^{- (1 + k)x}}{ 1 + k} \Phi_{1} \left( \frac{1}{2}, 1, 1, \frac{2k}{1 + k}, 2kx\right)   
\end{equation}
where $\Phi_{1} (a, b, c, x, y) $ denotes the Humbert series, or confluent Appell function of the first kind \cite{B:Tables, J:Humbert, J:Brychkov}. 
\end{theorem}

\begin{proof}
The proof is provided in Appendix G. 
\end{proof}

\subsection{A Simple Polynomial Representation}

Similar to the case of $ Q _{m,n}(a,b)$ and $T_{B}(m,n,r)$ functions, a simple representation for the Rice $Ie{-}$function  is advantageous for cases that   parameter generality and/or algebraic simplicity are required. 

\begin{proposition}
For $x, k \in \mathbb{R^{+}}$ and $0 \leq k \leq 1$, the following polynomial approximation  is valid for the $Ie(k,x)$ function, 

\begin{equation} \label{Rice_Polynomial}
Ie(k, x) \simeq \sum_{l = 0}^{L} \frac{\Gamma(L+l) L^{1 - 2l} k^{2l}\gamma( 1 + 2l,x)}{l!\Gamma(L-l+1)\Gamma(l+1)2^{2l}} 
\end{equation}
which as $L \rightarrow \infty$, it becomes an exact infinite series representation, 

\begin{equation} \label{Rice_Infinite_Series}
Ie(k, x) =\sum_{l = 0}^{\infty} \frac{ k^{2l}\gamma( 1 + 2l,x)}{l! \Gamma(l+1)2^{2l}}. 
\end{equation}
\end{proposition}
\begin{proof}
The proof follows immediately from Proposition $1$ and Proposition $3$. 
\end{proof}

\begin{table}[!h]
\centering
\caption{Accuracy of  proposed expressions for the $Ie (k, x)$ function} 
\begin{tabular}{|c|c|c|c|c|c|}
\hline  FUNCTION&EXACT&Eq. \eqref{Rice_UB_1}&Eq. \eqref{Rice_LB_1}&Eq. \eqref{Rice_Humbert}&Eq. \eqref{Rice_Polynomial}\\ 
\hline \hline $Ie(0.1, 0.1)$&$0.0952$&$0.0952$&$0.0631$&$0.0952$&$0.0952$  \\ 
\hline $Ie(0.1, 0.4)$&$0.3297$&$0.3328$&$0.2829$&$0.3297$&$0.3297$  \\ 
\hline $Ie(0.4, 0.4)$&$0.3303$&$0.3526$&$0.1384$&$0.3303$&$0.3303$  \\ 
\hline $Ie(0.6, 0.4)$&$0.3311$&$0.3696$&$0.0079$&$0.3311$&$0.3311$  \\ 
\hline $Ie(0.6, 0.8)$&$0.5993$&$0.6380$&$0.2630$&$0.5993$&$0.5993$  \\ 
\hline $Ie(0.8, 0.9)$&$0.6139$&$0.7400$&$0.1110$&$0.6139$&$0.6139$  \\ 
\hline  
\end{tabular} 
\end{table}

\subsection{A Closed-Form Upper Bound for the Truncation Error} 

The precise accuracy of  \eqref{Rice_Polynomial} can be quantified by an upper bound for the truncation error. 

\begin{lemma}
For $k, x \in \mathbb{R^{+}}$ and $0 \leq k \leq 1$,  the following closed-form inequality holds for the truncation error in \eqref{Rice_Polynomial},

\begin{equation} \label{Rice_Truncation}
\begin{split} 
\epsilon_{t} < & 1  +\sqrt{ \frac{k }{2 (1 - k)}} \left\lbrace 1 - 2\,  Q(\sqrt{2x}\sqrt{1-k}) \right\rbrace  - e^{-x} I_{0}(kx)   - \sum_{l = 0}^{L} \frac{\Gamma(L+l) k^{2l} \gamma( 1 + 2l,x)}{l! l! L^{2l - 1}  (L-l)!2^{2l}} \\
& - \sqrt{ \frac{k }{2 (1 + k)}} \left\lbrace 1 - 2\,  Q(\sqrt{2x}\sqrt{1+k}) \right\rbrace.         
\end{split}
\end{equation} 

\end{lemma}

\begin{proof}
When \eqref{Rice_Polynomial} is truncated after $L$ terms, the corresponding truncation error is given by,

\begin{align} \label{Rice_Truncation_1}
\epsilon_{t} &= \sum_{l = L+1}^{\infty} \frac{\Gamma(L+l) k^{2l}\gamma( 1 + 2l,x)}{l! L^{2l - 1}  (L-l)! l! 2^{2l}}  \\
&= \underbrace{\sum_{l= 0}^{\infty} \frac{\gamma( 1 + 2l,x)}{l!  l! 2^{2l} k^{-2l}} }_{\mathcal{I}_{9}} - \sum_{l = 0}^{L} \frac{\Gamma(L+l)  k^{2l}\gamma( 1 + 2l,x)}{l! (L-l)! l! L^{2l - 1} 2^{2l}}.  
\end{align}

\noindent
Since 

\begin{equation}
\mathcal{I}_{9} = Ie(k, x)
\end{equation}
equation  \eqref{Rice_Truncation_1} can be equivalently expressed as follows,

\begin{equation} \label{Rice_Truncation_2}
\epsilon_{t} = Ie(k, x)-  \sum_{l = 0}^{L} \frac{\Gamma(L+l) L^{1 - 2l} k^{2l}\gamma( 1 + 2l,x)}{l!\Gamma(L-l+1)\Gamma(l+1)2^{2l}}. 
\end{equation}

\begin{figure}[H!] 
\centering
\subfigure[$Ie(k,x)$ in \eqref{Rice_UB_1} and \eqref{Rice_LB_1} vs $x$ ] 
{\includegraphics[width=12cm, height=9cm]{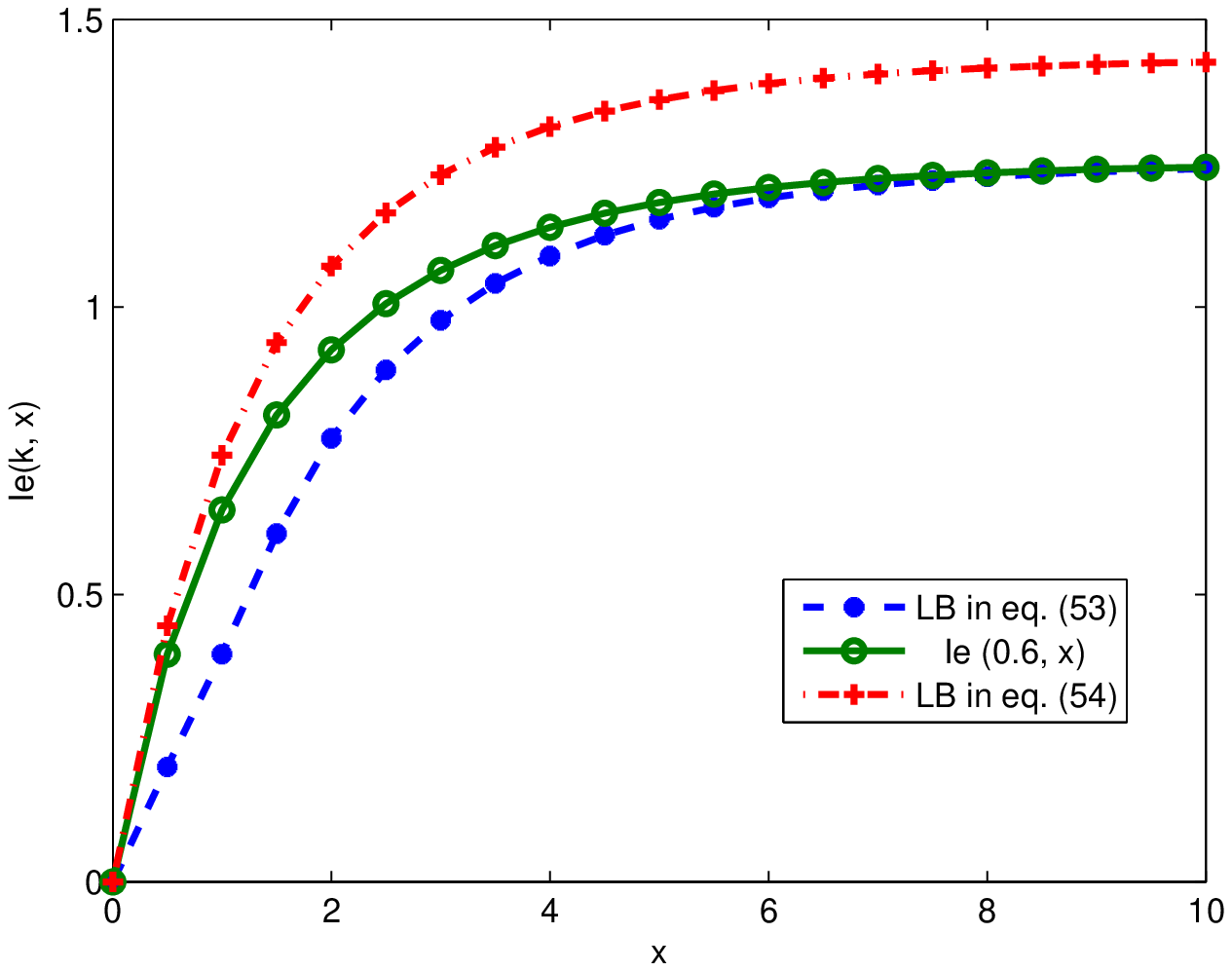} } 
\subfigure[$Ie(k,x)$ in \eqref{Rice_UB_1}, \eqref{Rice_LB_1}, \eqref{Rice_Humbert} and \eqref{Rice_Polynomial}   vs $k$]  
{ \includegraphics[width=12cm, height=9cm]{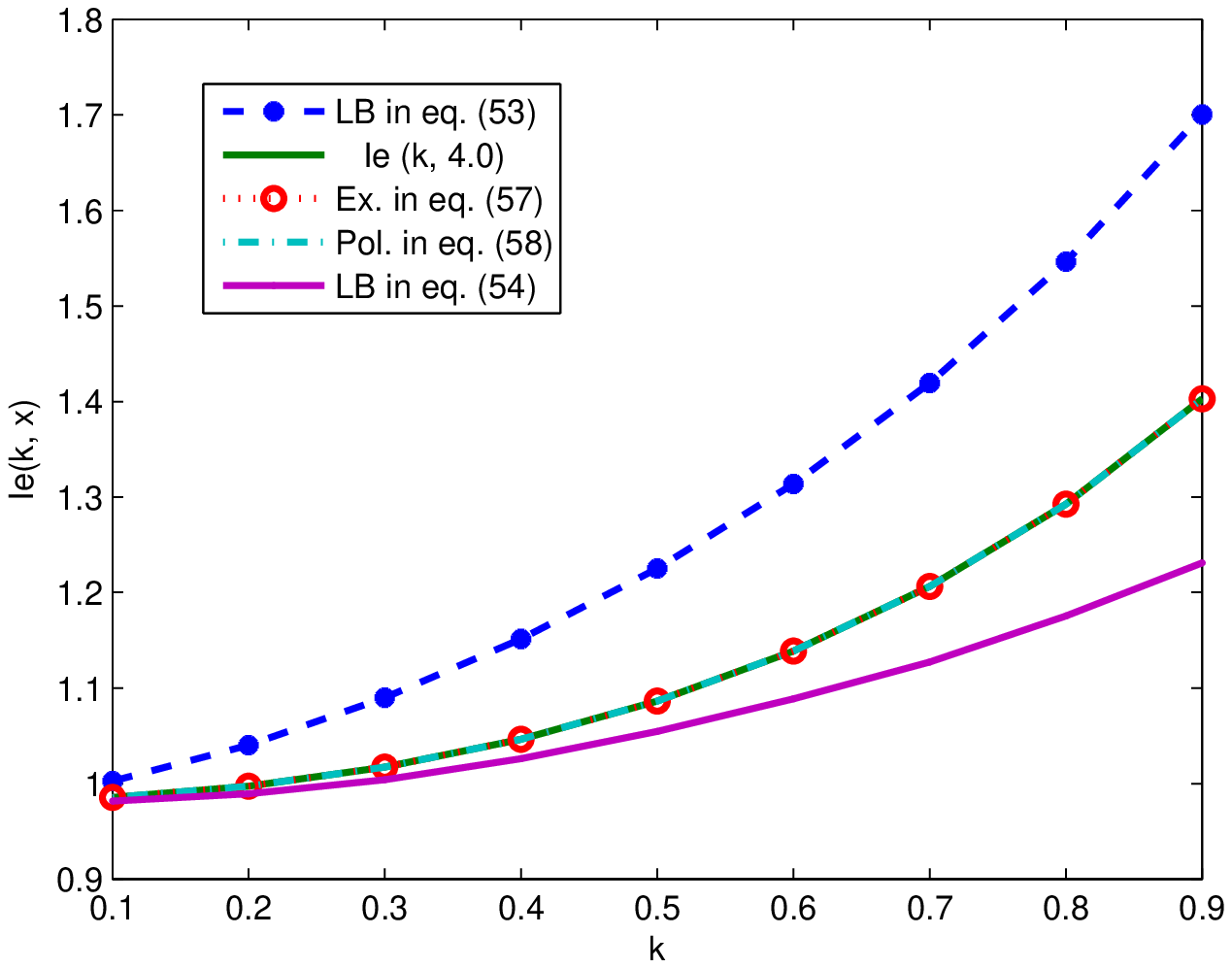} }
\subfigure[$Ie(k,x)$ in \eqref{Rice_UB_1}, \eqref{Rice_Humbert} and \eqref{Rice_Polynomial} vs $k$] 
{\includegraphics[width=12cm, height=9cm]{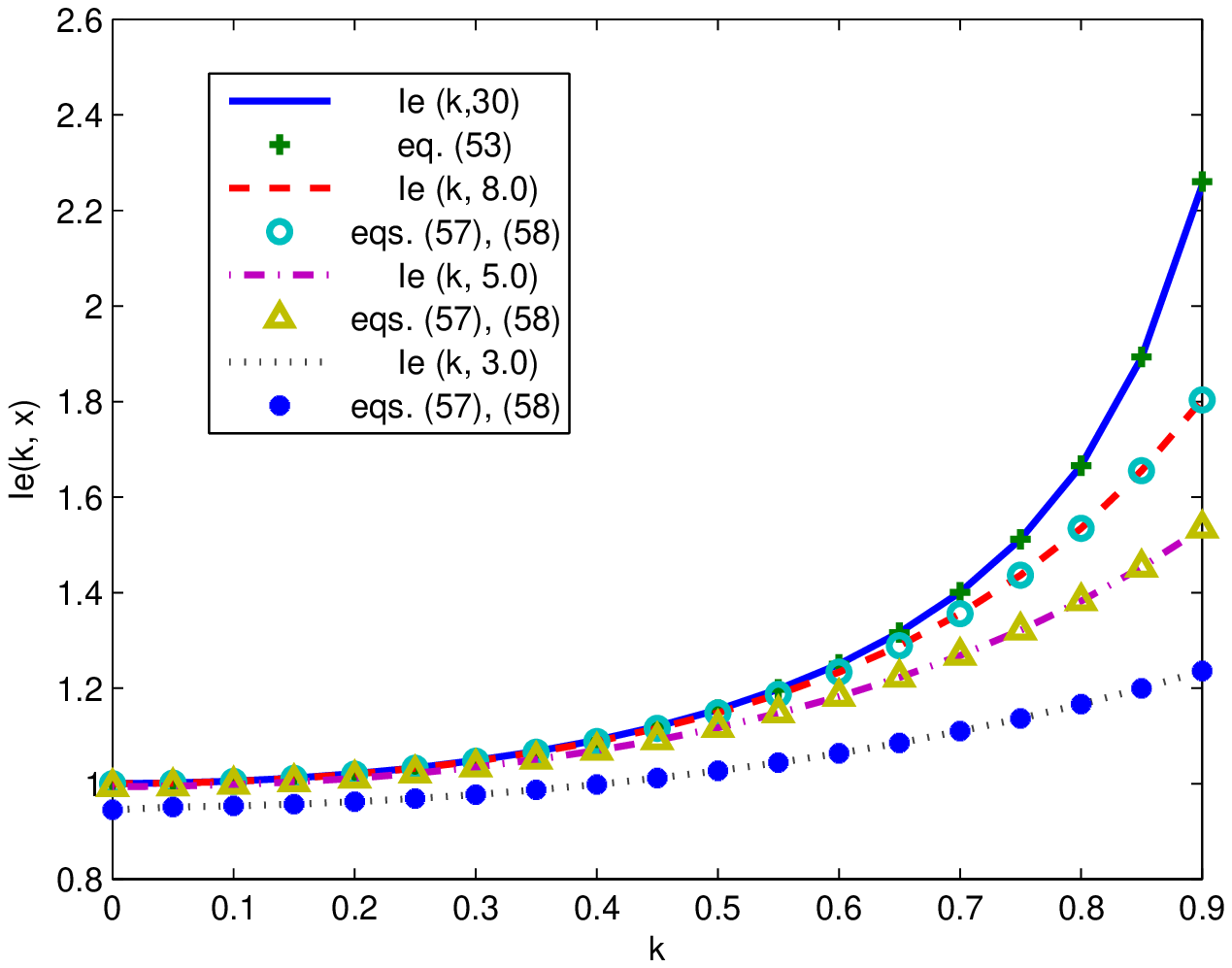}}
\caption{\small Behaviour of the $Ie(k,x)$ bounds, series, closed-form and approximation.}  
\end{figure}


\noindent
The $Ie(k,x)$ can be upper bounded with the aid of the closed-form upper bound in \eqref{Rice_UB_1}. As a result, the inequality in \eqref{Rice_Truncation} is deduced, which completes the proof.  
\end{proof}

\begin{remark}
By omitting the terms 

$$ \frac{\Gamma(L+l) L^{1 - 2l}}{\Gamma(L-l+1)}$$
 from the finite series term of \eqref{Rice_Truncation}, a similar closed-form upper bound is deduced for   \eqref{Rice_Infinite_Series}, namely, 

\begin{equation} \label{Rice_Truncation_new}
\begin{split} 
\epsilon_{t} &<  1  +\sqrt{ \frac{k }{2 (1 - k)}} \left\lbrace 1 - 2\,  Q(\sqrt{2x}\sqrt{1-k}) \right\rbrace  - e^{-x} I_{0}(kx) \\
&        - \sum_{l = 0}^{L} \frac{  \gamma( 1 + 2l,x)}{l! l! L^{2l - 1}   2^{2l}}  - \sqrt{ \frac{k }{2 (1 + k)}} \left\lbrace 1 - 2\,  Q(\sqrt{2x}\sqrt{1+k}) \right\rbrace    
\end{split} 
\end{equation} 
which is also tight.  
\end{remark}

Table III illustrates the behaviour of the derived expressions for the $Ie(k, x)$ function.  The proposed bounds are fairly tight for different values of $k$ and $x$ while it is clear that  \eqref{Rice_Humbert} and \eqref{Rice_Polynomial} are in excellent agreement with the respective exact numerical results.

Figures $3{\rm a}$ and $3{\rm b}$ also illustrate the behaviour of the bounds in  \eqref{Rice_UB_1} and \eqref{Rice_LB_1} versus $x$ and $k$, respectively. It is observed that the upper bound becomes tighter for small values of $x$ while for higher values of $x$ the lower bound appears to be tighter. Overall, it is observed that the lower bound is significantly tighter than the upper bound. This is also evident by Fig. $3 \rm{c}$ which indicates that the lower bound in \eqref{Rice_LB_1} becomes a remarkably accurate approximation to $Ie(k,x)$ for large values of $x$ as $\epsilon_{r} < 10^{-10}$ when $0 \leq k \leq 0.6$ and  $\epsilon_{r} < 10^{-5}$ when $ 0.6 < k \leq 1 $. This figure also depicts the behaviour of the closed-form expression in \eqref{Rice_Humbert} as well as the polynomial approximation in \eqref{Rice_Polynomial} which is shown to be in excellent agreement with the numerical results. This was achieved for truncation after $20$ terms which results to an involved error $\epsilon_{r} < 10^{-8}$.

\section{New Expressions for the Incomplete Lipschitz-Hankel Integrals}

\begin{definition}
For $m, a, n, x \in \mathbb{R}^{+}$, the general incomplete Lipschitz Hankel Integral is defined  by the following non-infinite integral representation, 

\begin{equation} \label{ILHIs_Definition}
Ze_{m,n}(x;a) \triangleq \int_{0}^{x} y^{m}e^{-ay} Z_{n}(y){\rm d}y 
\end{equation}
where $Z_{n}(x)$ can be one of the cylindrical functions $J_{n}(x)$, $I_{n}(x)$, $Y_{n}(x)$, $K_{n}(x)$, $H_{n}^{1}(x)$ or $H_{n}^{2}(x)$, \cite{J:Sagon, B:Maksimov}. 
\end{definition}
\noindent 
An alternative representation for the $I_{n}(x)$ based ILHIs was reported in \cite{J:Paris}, namely,  

\begin{equation}  \label{ILHIs_Alternative}
Ie_{m,n}(x;a)  = A_{m,n}^{0}(a) + e^{-ax} \sum_{i=0}^{m}\sum_{j=0}^{n+1} \frac{ B_{m,n}^{i,j}(a) }{[x^{i}I_{j}(x)]^{-1}}  +  \frac{Q_{1}\left(\sqrt{\frac{x}{a+\sqrt{a^{2}-1}}}, \sqrt{x}\sqrt{a+\sqrt{a^{2}-1}} \right)}{[A_{m,n}^{1}(a)]^{-1}} 
\end{equation}

\noindent 
where the set of coefficients $A_{m,n}^{l}(a)$ and $B_{m,n}^{i,j}(a)$ can be obtained recursively, \cite{J:Paris}. The above expression was employed in analytical investigations on error rate of MIMO systems under imperfect channel state information. Nevertheless, its algebraic representation is relatively inconvenient and laborious to handle analytically and numerically.


\subsection{Special Cases} 

A closed-form expression for $Ie_{m,n}(x;a)$ can be derived for the special case that $m$ and $n$ are positive half-integers. 

\begin{theorem}
For $a \in \mathbb{R}$, $x\in \mathbb{R}^{+}$, $m \pm  0.5 \in \mathbb{N}$, $n \pm 0.5 \in \mathbb{N}$ and  $m\geq n$, the following closed-form expression holds for the $Ie_{m, n}(x; a)$ integrals,

\begin{equation} \label{ILHIs_CF_1}  
\begin{split}
Ie_{m,n}(x;a) &=  \sum_{k=0}^{n - \frac{1}{2}} \frac{     \Gamma \left(n+k+\frac{1}{2} \right)  }{\sqrt{\pi}k! 2^{k + \frac{1}{2}} \Gamma \left(n-k+\frac{1}{2} \right) }    \qquad \quad  \qquad \qquad \\ 
& \quad \times  \left\lbrace  \frac{ (-1)^{k} \gamma \left(m-k+\frac{1}{2}, (a-1)x \right)}{ (a-1)^{m-k+\frac{1}{2}} } + \frac{ (-1)^{n + \frac{1}{2} } \gamma \left(m-k+\frac{1}{2}, (a + 1)x \right)}{ (a + 1)^{m-k+\frac{1}{2}} } \right\rbrace   
\end{split}
\end{equation}
\end{theorem}

\begin{proof}
The proof is provided in Appendix H. 
\end{proof}

Likewise, a closed-form expression is derived for the special case that  the sum of the indices $m$ and $n$ is a positive integer. 

\begin{theorem}
For  $m \in \mathbb{R}$, $ n  \in \mathbb{R}$, $x \in \mathbb{R}^{+}$, $a > 1$ and $m + n \in \mathbb{N}$, the following closed-form expression is valid for the $Ie_{m, n}(x; a)$ integrals,

\begin{equation} \label{m+n_Integer_1}
\begin{split}
  Ie_{m, n}(x;a) &= \frac{\Gamma(m + n + 1)}{2^{n} n! a^{m + n +1}} \, _{2}F_{1}\left( \frac{m + n + 1}{2}, \frac{m + n }{2} + 1; 1 + n; \frac{1}{a^{2}} \right) \\
  & -   \sum_{l = 0}^{m+n} \binom{m+n}{l} \frac{l! x^{m + n - l} e^{-x(1+a)} }{(1+a)^{l+1} 2^{n} n! }  \Phi_{1}\left(n + \frac{1}{2}, 1 + l, 1 + 2n; \frac{2}{1+a}, 2x \right)   
\end{split}
\end{equation}
where $\, _{2}F_{1}(a,b;c;x)$ denotes the Gauss hypergeometric function \cite{B:Tables}. 
\end{theorem}

\begin{proof}
The proof is provided in Appendix I. 
\end{proof}

In the same context, simple closed-form expressions can be derived for the specific cases  that $m = -n$ and $m=n=0$.

\begin{theorem}
For $ m \in \mathbb{R}$, $ n \in \mathbb{R}$, $ x \in \mathbb{R}^{+}$, $a > 1$ and $m = -n$, the following closed-form expression is valid for the $Ie_{m, n}(x;a)$ integrals, 

\begin{equation} \label{-n=n_1}
\begin{split}
  Ie_{-n, n}(x;a) &= \frac{\, _{2}F_{1} \left(n + \frac{1}{2}, 1; 1 + 2n ; \frac{2}{1 + a} \right)}{ (1 + a) n! 2^{n}}  -  \frac{ \Phi_{1} \left(n + \frac{1}{2}, 1, 1 + 2n; \frac{2}{1 + a}, 2x \right) }{2^{n} (1 + a) \Gamma(n + 1) e^{x (1 + a)} } 
\end{split}
\end{equation}
which for the specific case that $m = n = 0$ can be expressed as follows,

\begin{align} \label{ILHIs_Zeros}
Ie_{m = 0,n = 0}(x;a) &= Ie_{0,0}(x;a) \\
& = \frac{Q_{1}( \mathsf{b}, \mathsf{c}) - Q_{1}(\mathsf{c}, \mathsf{b})}{\sqrt{(a + 1)(a - 1)}}
\end{align}
where

\begin{equation}
\mathsf{b} = \sqrt{x}\sqrt{a + \sqrt{(a+1)(a-1)}}
\end{equation}
and
\begin{equation}
\mathsf{c}  = \sqrt{x}\sqrt{a - \sqrt{(a+1)(a-1)}}. 
\end{equation}

\end{theorem}

\begin{proof}
The proof is provided in Appendix J. 
\end{proof}

\subsection{Closed-form Upper and Lower Bounds}

Capitalizing on the derived closed-form expression for the $Ie_{m, n} (x;a)$ integrals in Theorem 7, tight closed-form upper and lower bounds can be readily deduced. 

\begin{lemma}
For $m, n, x, a \in \mathbb{R^{+}}$ and $m \geq n$, the following  inequalities can serve as upper and lower bounds to the $I_{n}(x)$ based incomplete Lipschitz Hankel integrals,

\begin{equation} \label{ILHIs_CFUB}
Ie_{m,n}(x; a) \leq Ie_{ \lceil m \rceil_{0.5}, \lceil n \rceil_{0.5}}(x;a) 
\end{equation}
and

\begin{equation} \label{ILHIs_CFLB}
Ie_{m,n}(x; a) \geq Ie_{\lfloor m \rfloor_{0.5}, \lfloor n \rfloor_{0.5}} (x;a). 
\end{equation}

\end{lemma}

\begin{proof}
The $Ie_{m,n}(x;a)$ integrals are monotonically increasing w.r.t $m$ and monotonically decreasing w.r.t $n$. 
By recalling the  two half-integer rounding operators in \cite[eq. (18)]{J:Karagiannidis} as well as  that \eqref{ILHIs_CF_1}  holds for $m \pm 0.5\in \mathbb{N}$ and $n \pm 0.5\in \mathbb{N}$, it becomes evident that  $Ie_{\lfloor m \rfloor_{0.5}, \lfloor n \rfloor_{0.5}}(x;a)$ and $Ie_{\lceil m \rceil_{0.5}, \lceil n \rceil_{0.5}}( x; a   )$ can be expressed in closed-form for any value of $m$, $n$, $r$ and $x$. As a result, equations \eqref{ILHIs_CFUB} and \eqref{ILHIs_CFLB} are   deduced and thus, completing the proof.  
\end{proof}

\subsection{A Simple Polynomial Representation}

The proposed expressions for the $Ie_{m, n}(x; a)$ integrals can be useful for applications related to  wireless communications. However,  a simpler and   more general  analytic expression is additionally  necessary for scenarios that require  unrestricted parameters and/or rather simple algebraic representation.  

\begin{proposition}
For $a, m, n \in \mathbb{R}$ and $x \in \mathbb{R^{+}}$, the following expression holds for the $Ie_{m,n}(x;a) $ integrals, 

\begin{equation} \label{ILHIs_Polynomial_1}
Ie_{m,n}(x; a)  \simeq \sum_{l=0}^{L} \frac{\Gamma(L+l)L^{1 - 2l}\gamma(m + n + 2l +1, ax)}{l!(L-l)!(n + l)!2^{n + 2l} a^{m+n+2l+1}} 
\end{equation}
which as $L \rightarrow \infty$ it reduces to the following exact infinite series representation, 

\begin{equation} \label{ILHIs_Infinite_Series_1}
Ie_{m,n}(x; a)  = \sum_{l=0}^{\infty} \frac{\gamma(m + n + 2l +1, ax)}{l! (n + l)!2^{n + 2l} a^{m+n+2l+1}}.
\end{equation}
\end{proposition}

\begin{proof}
The proof is provided in Appendix K. 
\end{proof}

\subsection{A Closed-form Upper Bound for the Truncation Error}

\begin{lemma}
For $a, m, n \in \mathbb{R}$ and $x \in \mathbb{R^{+}}$, the following inequality holds as an upper bound for the truncation error of \eqref{ILHIs_Polynomial_1}, 

\begin{equation}  \label{ILHIs_Truncation_1}
\begin{split}
\epsilon_{t} &\leq \sum_{k=0}^{\lceil n \rceil_{0.5}  - \frac{1}{2}} \frac{ 2^{-k - \frac{1}{2}} \Gamma \left(\lceil n \rceil_{0.5} +k+\frac{1}{2} \right)}{\sqrt{\pi}k!\Gamma \left(\lceil n \rceil_{0.5} -k+\frac{1}{2} \right)  } \\
&   \times \left\lbrace  \frac{(-1)^{k}  \gamma \left(m-k+\frac{1}{2}, (a-1)x \right) }{ (a-1)^{m-k+\frac{1}{2}}}  + \frac{ (-1)^{\lceil n \rceil_{0.5} + \frac{1}{2}} \gamma \left(m-k+\frac{1}{2}, (a+1)x \right) }{   (a+1)^{m-k+\frac{1}{2}}} \right\rbrace  \\
&  -  \sum_{l=0}^{L} \frac{\Gamma(L+l)L^{1 - 2l}\gamma(m + n + 2l +1, ax)}{l!(L-l)!(n + l)!2^{n + 2l} a^{m+n+2l+1}}. 
\end{split} 
\end{equation} 
\end{lemma} 

\begin{proof}
Since \eqref{ILHIs_Polynomial_1} reduces to \eqref{ILHIs_Infinite_Series_1}  as $L \rightarrow \infty$, the corresponding truncation error  is given by,

\begin{equation} \label{ILHIs_Truncation_2}
\epsilon_{t}  = \underbrace{\sum_{l=0}^{\infty} \frac{\gamma(m + n + 2l +1, ax)}{l!(n + l)!2^{n + 2l} a^{m+n+2l+1}} }_{\mathcal{I}_{10}} - \sum_{l=0}^{L} \frac{\Gamma(L+l)L^{1 - 2l}\gamma(m + n + 2l +1, ax)}{l!(L-l)!(n + l)!2^{n + 2l} a^{m+n+2l+1}}. 
\end{equation}

\noindent
Notably,  

\begin{equation}
\mathcal{I}_{10} = Ie_{m,n}(x;a)
\end{equation}
while $ Ie_{m,n}(x;a)$ can be upper bounded using $  Ie_{\lceil m \rceil _{0.5},\lceil n\rceil _{0.5} }(x;a)$. As a result, by substituting \eqref{ILHIs_CF_1} into \eqref{ILHIs_Truncation_2} one obtains   \eqref{ILHIs_Truncation_1}, which completes the proof. 
\end{proof} 

\begin{remark}
By omitting the terms 

$$ \frac{\Gamma(L+l)L^{1 - 2l}}{(L-l)!} $$
 in \eqref{ILHIs_Truncation_1}, a similar upper bound is also deduced for \eqref{ILHIs_Infinite_Series_1}, namely, 

\begin{equation}  \label{ILHIs_Truncation_1}
\begin{split}
\epsilon_{t} &\leq \sum_{k=0}^{\lceil n \rceil_{0.5}  - \frac{1}{2}} \frac{ 2^{-k - \frac{1}{2}} \Gamma \left(\lceil n \rceil_{0.5} +k+\frac{1}{2} \right)}{\sqrt{\pi}k!\Gamma \left(\lceil n \rceil_{0.5} -k+\frac{1}{2} \right)  } \\
&   \times \left\lbrace  \frac{(-1)^{k}  \gamma \left(m-k+\frac{1}{2}, (a-1)x \right) }{ (a-1)^{m-k+\frac{1}{2}}}  + \frac{ (-1)^{\lceil n \rceil_{0.5} + \frac{1}{2}} \gamma \left(m-k+\frac{1}{2}, (a+1)x \right) }{   (a+1)^{m-k+\frac{1}{2}}} \right\rbrace  \\
&  -  \sum_{l=0}^{L} \frac{ \gamma(m + n + 2l +1, ax)}{l! (n + l)!2^{n + 2l} a^{m+n+2l+1}}. 
\end{split} 
\end{equation}    
which is also rather tight. 
\end{remark}

\subsection{A Tight Closed-form Upper Bound and Approximation}


The algebraic representation of the $Ie_{m,n}(x;a)$ integrals allows the derivation of a simple upper bound  which in certain range of values becomes an accurate approximation. 

\begin{proposition}
For $m, n, a, x \in \mathbb{R^{+}}$ and $x, a > m, n$, the following inequality is valid for the ILHIs, 

\begin{equation} \label{ILHIs_Appr_1} 
Ie_{m,n}(x;a) \leq    \frac{ (n + 1)_{m}  \, _{2}F_{1}\left( \frac{m + n + 1}{2}, \frac{m +n}{2} + 1; n + 1; \frac{1}{a^{2}} \right) }{a^{m + n + 1} 2^{n}}    
\end{equation}
which for $a>3$ and $,x > 3$ becomes an accurate closed-form approximation.
%
\end{proposition}

\begin{proof}
The $\gamma(a,x)$ function can be upper bounded with the aid of the following $\Gamma(a)$ function property,

\begin{equation}
\Gamma(a) = \gamma(a, x = \infty).
\end{equation}  
To this effect, the $Ie_{m,n}(x;a)$ integrals can be upper bounded as follows:  

\begin{equation} \label{ILHIs_Appr_2}
Ie_{m,n}(x; a)  \leq \sum_{l=0}^{L} \frac{\Gamma(L+l)L^{1 - 2l}\Gamma(m + n + 2l +1)}{l!(L-l)!(n + l)!2^{n + 2l} a^{m+n+2l+1}}.
\end{equation}
As $L\rightarrow \infty$ and recalling that $x! = \Gamma(x+1)$ and $\Gamma(a,n) = (a)_{n}\Gamma(a)$ it immediately  follows that,

\begin{equation} \label{ILHIs_Appr_3}
Ie_{m,n}(x; a) \leq  \sum_{l=0}^{\infty} \frac{(m + n + 1)_{2l} \Gamma(m+n+1) 2^{-2l} }{l!(n+1)_{l} \Gamma(n+1) 2^{n} a^{m+n+2l+1}}.  
\end{equation}
Importantly, with the aid of the identity, 

\begin{equation}
(2x)_{2l} \triangleq 2^{2l} ( x )_{l} (x +  0.5 )_{l}
\end{equation}
equation \eqref{ILHIs_Appr_3} can be also expressed as,

\begin{equation} \label{ILHIs_Appr_5}
Ie_{m,n}(x; a) \leq \frac{(m+n)!}{(n)! a^{m +n +1} }  \sum_{l=0}^{\infty} \frac{\left( \frac{m+n+1}{2}\right)_{l} \left( \frac{m+n}{2} + 1 \right)_{l} }{ l! (n + 1)_{l} 2^{n + 2l} a^{2l} 2^{-2l} }. 
\end{equation}
\begin{table*}
\centering
\caption{Accuracy of  proposed expressions for the $Ie_{m,n}(x;a)$ integrals} 
\begin{tabular}{|c|c|c|c|c|}
\hline  FUNCTION&EXACT&Eqs. \eqref{ILHIs_CF_1}, \eqref{m+n_Integer_1}, \eqref{-n=n_1}, \eqref{ILHIs_Zeros}&Eq. \eqref{ILHIs_Polynomial_1}&Eq. \eqref{ILHIs_Appr_1}\\ 
\hline \hline $Ie_{0, 0}(3.2;1.7)$&$0.6974$&n/a, $0.6974$, n/a, $0.6974$&$0.6974$&$0.7274$  \\ 
\hline $Ie_{0, 0}(3.2;2.7)$&$0.3982$&n/a, $0.3982$, n/a, $0.3982$&$0.3982$&$0.3987$  \\ 
\hline $Ie_{0.5, 0.5}(3.2;1.7)$&$0.3615$&$0.3615$, $0.3615$, n/a, n/a&$0.3615$&$0.4222$  \\ 
\hline $Ie_{0.5, 0.5}(3.2;2.7)$&$0.1258$&$0.1258$, $0.1258$, n/a, n/a&$0.1258$&$0.1268$  \\ 
\hline $Ie_{-0.5, 0.5}(3.2;1.7)$&$0.5245$&n/a, $0.5245$, $0.5245$, n/a&$0.5245$&$0.5385$  \\ 
\hline $Ie_{-0.5, 0.5}(3.2;2.7)$&$0.3000$& n/a, $0.3000$, $0.3000$, n/a &$0.3000$&$0.3103$  \\ 
\hline  
\end{tabular} 
\end{table*}
The above series can be expressed in closed-form in terms of the Gaussian hypergeometric function $\,_{2}F_{1}(a,b;c;x)$. Hence, by substituting in \eqref{ILHIs_Appr_5} and performing some basic algebraic manipulations  \eqref{ILHIs_Appr_1} is deduced thus, completing the proof.  
\end{proof}

The accuracy of the derived analytic expressions for the $Ie_{m,n}(x;a)$ integrals  is depicted  in Table IV (top of the next page)  along with respective results from numerical integrations. One can notice the excellent agreement between analytical and numerical results while the proposed upper bound and approximation  appear to be quite accurate. Specifically,  truncating \eqref{ILHIs_Polynomial_1} after $30$ terms and for $a<2$ yields a relative error of $\epsilon_{r} < 10^{-4}$. It is also noticed that the tight upper bound for small values of $a$ becomes an accurate approximation as $a$ increases. This is additionally  evident by the involved relative error which can be as low as $\epsilon_{r} < 10^{-9}$.

In the same context, the accuracy of the proposed polynomial approximations for the above functions and integrals is depicted in Table  V (top of the next page) in terms of the involved relative error. Evidently, the value of $\epsilon_{r}$ is rather low for numerous different parametric scenarios which indicates the overall high accuracy of the proposed analytic expressions.

\begin{table*}
\centering
\caption{Absolute relative error for all proposed series representations} 
\begin{tabular}{|c|c||c|c|}
\hline FUNCTION & n = 30 & FUNCTION & n = 30  \\ 
\hline \hline $ Q_{1.1, 0.8}(1.7, 1.4)$ & $5.0\times 10^{-13}$&$Ie_{1.1, 0.8}(1.7; 1.4)$&$4.0\times 10^{-10}$ \\ 
\hline $Q_{1.1, 1.4}(1.9, 1.2)$ & $9.7\times 10^{-12}$ & $Ie_{1.1, 1.4}(1.9; 1.2)$&$9.4\times 10^{-11}$   \\
\hline $Q_{2.2, 0.9}(2.1, 1.9)$ & $1.9\times 10^{-13}$ &$Ie_{2.2, 0.9}(2.1; 1.9)$&$3.0\times 10^{-10}$    \\
\hline $Q_{0.9, 1.2}(0.6, 0.9)$ & $7.3\times 10^{-13}$ &$Ie_{0.9, 1.2}(0.6; 0.9)$&$9.1\times 10^{-11}$   \\ 
\hline $Q_{1.7, 1.7}(0.3, 0.2)$ & $1.8\times 10^{-13}$ &$Ie_{1.7, 1.7}(0.3; 0.2)$&$1.5\times 10^{-6}$   \\ 
\hline \hline $T_{3}(1.8,0.9,0.7)$&$7.5\times 10^{-10}$&$Ie(0.3, 1.8)$&$1.2\times 10^{-15}$  \\ 
\hline $T_{3}(1.1,1.9,1.2)$&$9.8\times 10^{-9}$&$Ie(0.3, 3.1)$&$1.5\times 10^{-15}$  \\
\hline $T_{4}(1.3,1.3,1.9)$&$2.1\times 10^{-9}$&$Ie(0.9, 1.2)$&$1.3\times 10^{-15}$  \\ 
\hline $T_{4}(2.7,2.7,2.7)$&$7.3\times 10^{-12}$&$Ie(0.9, 4.8)$&$1.4\times 10^{-15}$    \\ 
\hline   
\end{tabular} 
\end{table*}

\begin{figure}[] 
\centering
\subfigure[$Ie_{m,n}(x; a)$ in \eqref{ILHIs_CF_1}, \eqref{m+n_Integer_1}, \eqref{-n=n_1}, \eqref{ILHIs_Zeros} $\&$ \eqref{ILHIs_Polynomial_1} ]
{\includegraphics[width=12cm, height=9cm]{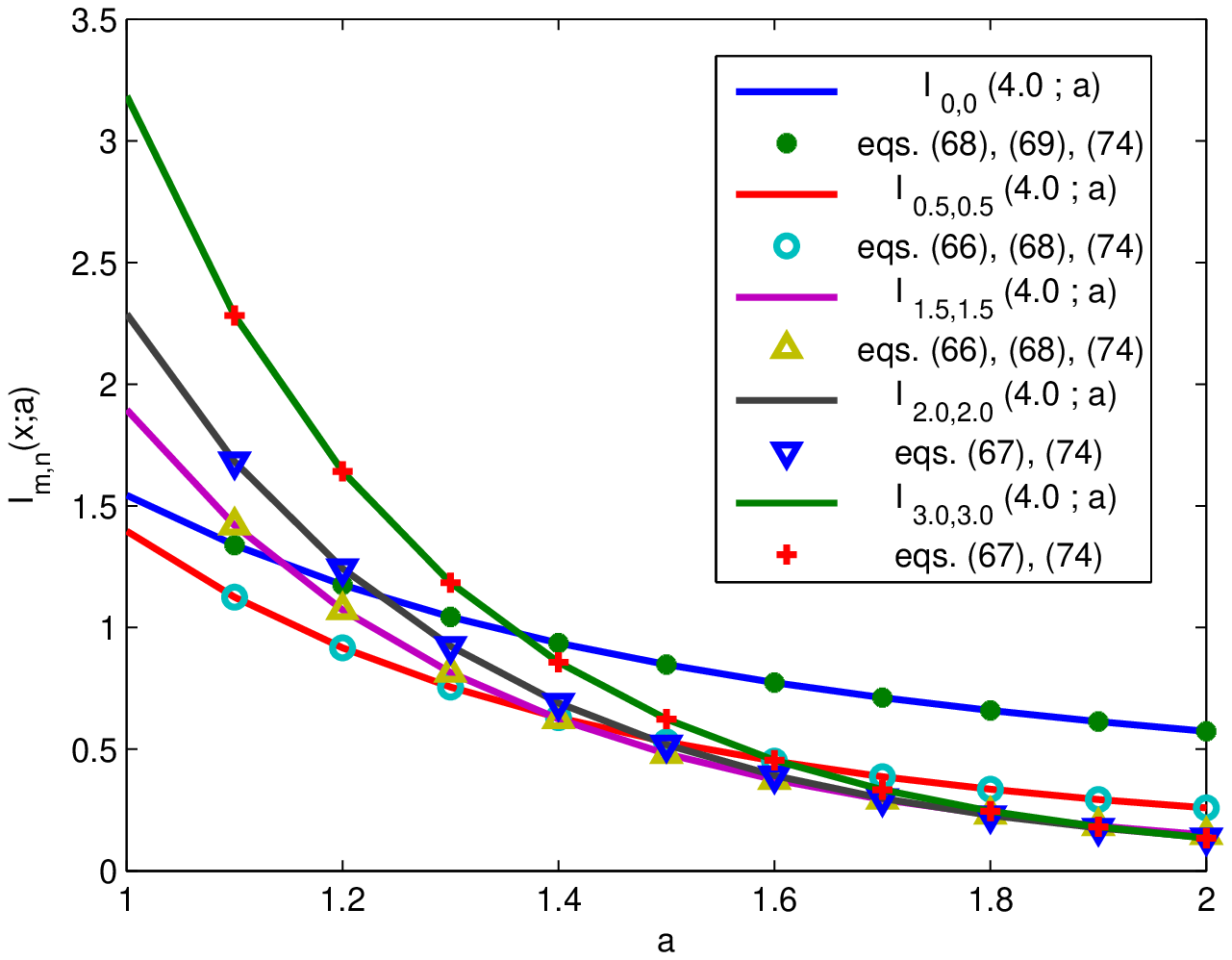} }
\subfigure[$Ie_{m,n}(x; a)$ in \eqref{ILHIs_Appr_1}]
{\includegraphics[width=12cm, height=9cm]{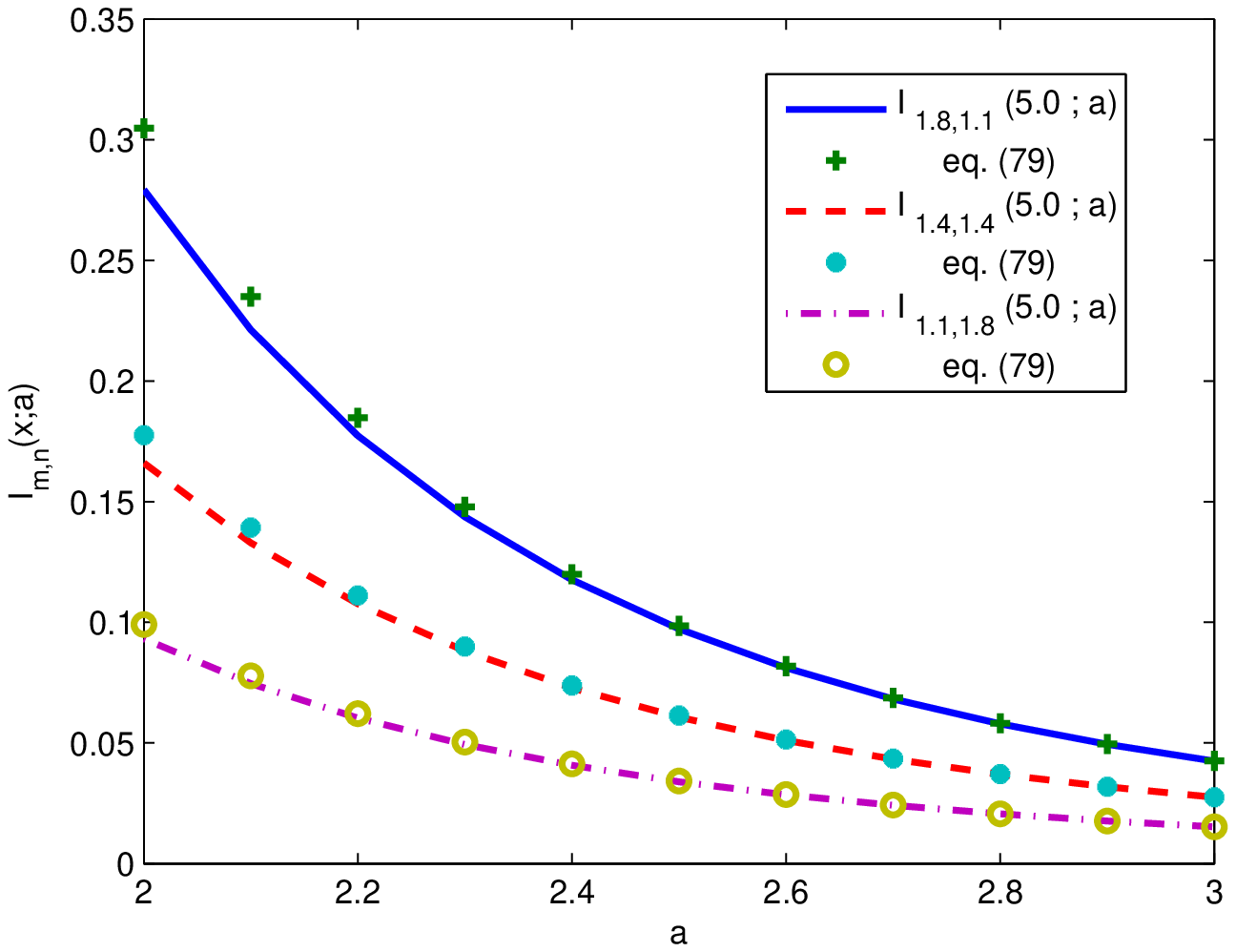}}
\caption{\small Behaviour of the proposed bounds, closed-forms, series and approximation for the $Ie_{m,n}(x;a)$ integrals.}
\end{figure}

The behaviour of the analytic expressions in \eqref{ILHIs_CF_1}, \eqref{m+n_Integer_1}, \eqref{-n=n_1}, \eqref{ILHIs_Zeros} $\&$ \eqref{ILHIs_Polynomial_1} is  illustrated in Fig. $4{\rm a}$ along with respective results from numerical integrations. One can notice the excellent agreement between analytical and numerical results. For \eqref{ILHIs_Polynomial_1}, this is achieved by truncating the series after $30$ terms which corresponds to a relative error of $\epsilon_{r} < 10^{-4}$ when $a<2$. Likewise, the accuracy of  \eqref{ILHIs_Appr_5} is illustrated in Fig. $4$b where it is observed that the tight upper bound for small values of $a$ becomes an accurate approximation as $a$ increases asymptotically. This is also evident by the involved relative error which can be as low as $\epsilon_{r} < 10^{-9}$.

\section{Closed-Form Expressions for Special Cases of the  Kampe de Feriet and the  Humbert $\Phi_1$ Functions}

The previous Sections were devoted to the derivation of novel analytic expressions for the $Q_{m,n}(a,b)$, $T_{B}(m,n,r)$, $Ie(k,x)$ functions and the $Ie_{m,n}(x;a)$ integrals. Capitalizing on the offered analytic results, useful closed-form expression can be readily  deduced for special cases of the  KdF and Humbert $\Phi_1$  special functions. It is noted here that these functions are rather general and particularly the KdF can represent the vast majority of special functions. As a result, relating expressions are rather necessary in unified representations of different special functions that are used in digital communications. 

\begin{corollary}
For  $x, y \in \mathbb{R}^{+}$ and $a > -\frac{1}{2}$,  $b > - 1$, the following  closed-form expression is valid,

\begin{align} \label{KdF_Special_1}
F_{1, 1}^{1, 0} \left(^{ a : - , - : }_{ a + 1: b , - : } x, - y \right) &= \mathcal{F}_{1} (a, a+1, b; x, -y)  \\
&=  \frac{ a  \Gamma(b) \, T_{\sqrt{y}}\left( 2a - 1, b- 1, \sqrt{\frac{x}{y}} \right)  }{ x^{b - a } y^{2a - b  }  e^{ - \frac{x}{y}} }  
\end{align}
where $\mathcal{F}_{1}(\cdot)$ denotes the following   infinite series representations, 

\begin{align} \label{KdF_Special_2}
\mathcal{F}_{1} (a, a+1, b; x, -y) &= \sum_{l = 0}^{\infty}\sum_{i = 0}^{\infty} \frac{a}{a + l + i} \frac{x^{l} }{l! } \frac{y^{i}}{i!} \\
& = \sum_{l = 0}^{\infty}\sum_{i = 0}^{\infty} \frac{(a)_{l+i}}{(a+1)_{l+i} (b)_{l}} \frac{x^{l} }{l! } \frac{y^{i}}{i!}.  
\end{align}
\end{corollary}

\begin{proof}
The proof follows from Theorem $4$ by setting 

\begin{equation} 
a = \frac{m + 1}{2}
\end{equation}
and $b = n+1$,  $x = r^{2}B^{2}$ and $y = B^{2}$
\end{proof}

\begin{corollary}
For $a, b \in \mathbb{R}$ and  $x, y \in \mathbb{R}^{+}$, the following  closed-form expression is valid,

\begin{align} \label{KdF_Special_3}
F_{1, 1}^{1, 0} \left(^{ a : - , - : }_{ a + 1: b , - : } x, - y \right) &= \mathcal{F}_{1} (a, a+1, b; x, -y)    \\
& =   \frac{a\Gamma(a) \,_{1}F_{1}\left(b - a, b, - \frac{x}{y} \right)}{y^{a} e^{-\frac{x}{y}}}  - \frac{a \Gamma(b) Q_{2a-b,b-1}\left( \sqrt{\frac{2x}{y}}, \sqrt{2y} \right) }{y^{a - \frac{b-1}{2}} e^{- \frac{x}{y} } x^{\frac{b-1}{2}} 2^{a - \frac{b+1}{2}}}  \nonumber  
\end{align}
where $\mathcal{F}_{1} (a, a+1, b; x, -y)$ is given in \eqref{KdF_Special_2}. 
\end{corollary}

\begin{proof}
The proof follows immediately by applying Lemma $2$ in Corollary $1$. 
\end{proof}

Likewise, closed-form expressions are deduced for special cases of the Humbert $\Phi_{1}$  function. 

\begin{corollary}
For $a \in \mathbb{R}$, $y \in \mathbb{R}^{+}$ and $-1 < x < 1$, the following closed-form expression holds, 

\begin{equation} \label{Humbert_Special_1}
\Phi_{1}(a, 1, 2a; x, y) = \,_{2}F_{1}(a, 1, 2a, x) e^{\frac{y}{x}}- \frac{2^{a + \frac{1}{2}}\Gamma\left(a + \frac{1}{2} \right) }{x e^{-\frac{y}{x}}} Ie_{\frac{1}{2} - a, a - \frac{1}{2} } \left( \frac{y}{2}; \frac{2}{x} - 1 \right).  
\end{equation}
\end{corollary}

\begin{proof}
The proof follows immediately from \eqref{-n=n_1} in from Theorem $9$ for $a = n + \frac{1}{2}$. 
\end{proof}

\begin{corollary}
For $y \in \mathbb{R}^{+}$ and $-1 < x < 1$, the following closed-form expression is valid, 

\begin{equation} \label{Humbert_Special_2}
\Phi_{1}\left( \frac{1}{2}, 1, 1; x, y \right)=  e^{ \frac{y}{x} }  \, _{2}F_{1}\left( \frac{1}{2}, 1, 1, x \right) -  e^{ \frac{y}{x} } \frac{ Q_{1}(b, c) - Q_{1}(c, b) }{\sqrt{1 - x}}   
\end{equation}
where 

\begin{equation} \label{Humbert_Special_3}
b = \sqrt{ \frac{y}{x} (1 + \sqrt{1 - x}) - \frac{y}{2} }
\end{equation}
and

\begin{equation} \label{Humbert_Special_4}
c = \sqrt{ \frac{y}{x} (1 - \sqrt{1 - x}) - \frac{y}{2} }. 
\end{equation}
\end{corollary}

\begin{proof}
The proof follows from \eqref{ILHIs_Zeros} in Theorem $9$ by setting $n = 0$ and $a = \frac{2}{x} - 1$.  
\end{proof}

\section{Applications in Wireless Communications Theory}

As already mentioned, the offered analytic results can be particularly useful in the broad area of wireless communications. To this end, they are indicatively employed in deriving  analytic expressions for applications relating to digital communications over fading channels. Novel closed-form expressions are derived for the OP over non-linear  generalized fading channels that follow the $\alpha{-}\eta{-}\mu$, $\alpha{-}\lambda{-}\mu$ and $\alpha{-}\kappa{-}\mu$ distributions. These fading models were proposed in \cite{Yacoub_NL_1, Yacoub_NL_2} and are distinct for their remarkable flexibility as they have been shown to provide accurate fitting in measurements that correspond to versatile realistic communication scenarios. This is clearly indicated in  \cite[Fig. 1]{Yacoub_NL_1} while it is also evident by the fact that these models include as special cases the well known $\alpha{-}\mu$, $\eta{-}\mu$ and $\kappa{-}\mu$ distributions and therefore, the Hoyt, Rice, Weibull, Nakagami${-}m$ and Rayleigh distributions \cite{Yacoub_a-m, Yacoub_h-m_k-m, Yacoub_l-m, Additional_1, Additional_2, Additional_3, Additional_4, Additional_5, Additional_6, Additional_7, Additional_8, Additional_9, Additional_10, Additional_11}. In addition,  closed-form expressions are additionally  deduced for specific cases of  OP over $\eta{-}\mu$ and $\lambda{-}\mu$ fading channels as well as for the truncated channel inversion with fixed rate transmission in both single and multi-antenna systems over  Rician fading channels.

\subsection{Outage Probability over $\alpha{-}\eta{-}\mu$  Fading Channels }

The $\alpha{-}\eta{-}\mu$ distribution is a particularly flexible fading model that provides accurate characterization of various multipath fading scenarios including modelling  of satellite links subject to strong atmospheric scintillation. Furthermore, it constitutes a generalization of $\eta{-}\mu$ distributions and thus, it includes as special cases the $\eta{-}\mu$, $\alpha{-}\mu$, Hoyt, Nakagami${-}m$ and Rayleigh distributions. In terms of physical interpretation of the involved parameters,  $\alpha$ denotes the   non-linearity parameter which accounts for the non-homogeneous diffuse scattering field, $\mu$ is related to the number of multipath clusters  and $\eta$ is the scattered-wave power ratio between the in-phase and quadrature components of each cluster of  multipath \cite{Yacoub_NL_1}.  
\begin{definition}
For  $\alpha, \eta, \mu, \rho\in \mathbb{R}^{+} $, the normalized envelope PDF for the $\alpha{-}\eta{-}\mu$ distribution is expressed as, 

\begin{equation} \label{OP_a-h-m_2}
p_{P}(\rho) = \frac{\alpha (\eta + 1)^{\mu + \frac{1}{2}} \sqrt{\pi} \mu^{\mu + \frac{1}{2}}  I_{\mu - \frac{1}{2}} \left( \frac{(\eta^{2} - 1) \mu \rho^{\alpha}}{2 \eta} \right)  }{\Gamma(\mu) \sqrt{\eta} (\eta - 1)^{\mu - \frac{1}{2}}    \rho^{1 - \alpha \left(\mu + \frac{1}{2} \right) } e^{\frac{(1 + \eta)^{2} \mu \rho^{a}}{2 \eta}}}. 
\end{equation}
\end{definition}
\begin{corollary}
For  $\alpha$, $\eta$, $\mu$, $\overline{\gamma}$, $\gamma_{th} \in \mathbb{R}^{+} $, the OP  over independent and identically distributed (i.i.d) $\alpha{-}\eta{-}\mu$ fading channels can be expressed as follows,

\begin{equation} \label{OP_a-h-m_1}
P_{\rm out} =  \frac{\sqrt{\pi} 2^{\mu + \frac{1}{2}} \eta^{\mu} }{\Gamma(\mu) (\eta - 1)^{2 \mu}}Ie_{ \mu +\frac{1}{2} + \frac{4 (1 - \alpha)}{\alpha^{2}}, \, \mu - \frac{1}{2} }\left( \frac{ \mu (\eta^{2} - 1)  \gamma_{th}^{\alpha{/}2}}{2 \eta \overline{\gamma}^{\, \alpha{/}2}} \, ; \, \frac{\eta + 1}{\eta - 1} \right)
\end{equation}
where  $\overline{\gamma}$ and $\gamma_{th}$ denote  the average SNR and the pre-determined SNR threshold, respectively.  
\end{corollary}
\begin{proof}
Based on the envelope PDF in \eqref{OP_a-h-m_2}, the PDF of the corresponding  SNR   per symbol is given by  \cite[eq. (1)]{Yacoub_NL_1},

\begin{equation} \label{OP_a-h-m_3}
p_{\gamma}(\gamma) =  \frac{\alpha (\eta + 1)^{\mu + \frac{1}{2}} \sqrt{\pi} \mu^{\mu + \frac{1}{2}} }{2 \Gamma(\mu) \sqrt{\eta} (\eta - 1)^{\mu - \frac{1}{2}}   }  \frac{\gamma^{\alpha \left(\mu + \frac{1}{2} \right) - 1}}{\overline{\gamma}^{\alpha \left(\mu + \frac{1}{2} \right) }} e^{-\frac{(1 + \eta)^{2} \mu }{2 \eta} \frac{\gamma^{\alpha{/}2}}{\overline{\gamma}^{\alpha{/}2}} }  I_{\mu - \frac{1}{2}} \left( \frac{(\eta^{2} - 1) \mu  }{2 \eta} \frac{\gamma^{\alpha{/}2}}{\overline{\gamma}^{\alpha{/}2}} \right). 
\end{equation}
It is also recalled that the OP over fading channels is defined as \cite[eq. (1.4)]{B:Alouini_2005}, namely, 

\begin{equation} \label{OP_a-h-m_4} 
P_{\rm out} \triangleq  F \left( \gamma_{th} \right) = \int_{0}^{\gamma_{th}} p_{\gamma}(\gamma) {\rm d} \gamma   
\end{equation}
where $F(\gamma)$ is the cumulative distribution function (CDF) of $\gamma$.  Thus, by substituting \eqref{OP_a-h-m_3} in \eqref{OP_a-h-m_4} and after performing necessary change of variables and basic algebraic manipulations yields, 

\begin{equation} \label{OP_a-h-m_5}
P_{\rm out} =     \frac{\alpha (\eta + 1)^{\mu + \frac{1}{2}} \sqrt{\pi} \mu^{\mu + \frac{1}{2}} }{2 \Gamma(\mu) \sqrt{\eta} (\eta - 1)^{\mu - \frac{1}{2}} \overline{\gamma}^{\alpha \left(\mu + \frac{1}{2} \right) }  }  \int_{0}^{\gamma_{th}}  \frac{\gamma^{\alpha \left(\mu + \frac{1}{2} \right) - 1}}{ e^{\frac{(1 + \eta)^{2} \mu }{2 \eta} \frac{\gamma^{\alpha{/}2}}{\overline{\gamma}^{\alpha{/}2}} }}  I_{\mu - \frac{1}{2}} \left( \frac{(\eta^{2} - 1) \mu  }{2 \eta} \frac{\gamma^{\alpha{/}2}}{\overline{\gamma}^{\alpha{/}2}} \right)  {\rm d}\gamma. 
\end{equation}
By setting 

\begin{equation}
 u =  \frac{ (\eta^{2} - 1) \mu \gamma^{\alpha{/}2}}{( 2 \eta \overline{\gamma}^{\alpha{/}2})} 
\end{equation} 
 and carrying out some  long but basic algebraic manipulations it follows that,  

\begin{equation} \label{OP_a-h-m_6}
P_{\rm out} =   \frac{\sqrt{\pi} \eta^{\mu} 2^{\mu + \frac{1}{2}}}{\Gamma(\mu) (\eta - 1)^{2 \mu}}  \int_{0}^{\frac{(\eta^{2} - 1) \mu }{2 \eta \overline{\gamma}^{\alpha{/}2}}\gamma_{th}^{\alpha{/}2}} u^{\mu + \frac{\alpha^{2} + 8(1 - \alpha)}{2 \alpha^{2}} } e^{- \frac{\eta + 1}{\eta - 1} u } I_{\mu - \frac{1}{2}} (u) {\rm d} u. 
\end{equation}
Notably,   the above integral can be expressed in terms of the ILHIs. As a result \eqref{OP_a-h-m_1} is deduced, which completes the proof.  
\end{proof}
\begin{remark}
By recalling that 

\begin{equation}
P_{out} \triangleq F_{\gamma}(\gamma_{th})
\end{equation}
it immediately  follows from  \eqref{OP_a-h-m_1} that the CDF of the $\alpha{-}\eta{-}\mu$ distribution can be expressed as,

\begin{equation} \label{CDF_a-h-m_1}
F_{\gamma}(\gamma) =  \frac{\sqrt{\pi} 2^{\mu + \frac{1}{2}} \eta^{\mu} }{\Gamma(\mu) (\eta - 1)^{2 \mu}} Ie_{ \mu +\frac{1}{2} + \frac{4 (1 - \alpha)}{\alpha^{2}}, \, \mu - \frac{1}{2} }\left( \frac{ \mu (\eta^{2} - 1)  \gamma^{\alpha{/}2}}{2 \eta \overline{\gamma}^{\, \alpha{/}2}} \, ; \, \frac{\eta + 1}{\eta - 1} \right). 
\end{equation}
Furthermore, for the specific case that $\alpha = 2$, equation \eqref{OP_a-h-m_1} yields a closed-form expression for the  OP over $\eta{-}\mu$ fading channels, namely, 

\begin{equation} \label{OP_a-h-m_7}
P_{\rm out} = \frac{\sqrt{\pi} 2^{\mu + \frac{1}{2}} \eta^{\mu} }{\Gamma(\mu) (\eta - 1)^{2 \mu}} Ie_{ \mu -\frac{1}{2}, \, \mu - \frac{1}{2} }\left( \frac{ \mu (\eta^{2} - 1)  \gamma_{th}^{}}{2 \eta \overline{\gamma}^{\, }} \, ; \, \frac{\eta + 1}{\eta - 1} \right)
\end{equation}
which is valid for all values of $\eta$ and $\mu$. 
\end{remark}

\subsection{Outage Probability over $\alpha{-}\lambda{-}\mu$  Fading Channels } 

 The $\alpha{-}\lambda{-}\mu$ distribution has been also proposed as an accurate  fading model that  represents \textit{small scale} signal variations.  It is closely related to the $\alpha{-}\eta{-}\mu$ distribution and is  also known as its \textit{Format 2} while it includes as special cases the  $\lambda{-}\mu$, $\alpha{-}\mu$, Hoyt, Nakagami${-}m$ and Rayleigh  distributions. In terms of physical interpretation,  $\lambda$ is the correlation coefficient between the scattered-wave in-phase and quadrature components of each
cluster of multipath while  $\alpha$  and $\mu$ denote the  non-linearity parameter  and the number of multipath clusters, respectively.  
\begin{corollary}
For $-1 < \lambda < 1$  and  $\alpha$, $\mu$, $\overline{\gamma}$, $\gamma_{th} \in \mathbb{R}^{+} $, the OP  over i.i.d  $\alpha{-}\lambda{-}\mu$ fading channels can be expressed as,

\begin{equation} \label{OP_a-l-m_1}
P_{\rm out} =  \frac{ (-1)^{2 \mu}  \sqrt{\pi} (1 - \lambda)^{\mu} (1 + \lambda)^{\mu}}{\Gamma(\mu) 2^{\mu - \frac{1}{2}} \lambda^{2 \mu}}  Ie_{\mu + \frac{1}{2} + \frac{4(1 - \alpha)}{\alpha^{2}}, \, \mu - \frac{1}{2}}\left(\frac{2 \lambda \mu \gamma_{th}^{\alpha{/}2}}{(\lambda^{2} - 1) \overline{\gamma}^{\alpha{/}2}} ,  - \frac{1}{\lambda} \right).  
\end{equation}
\end{corollary}
\begin{proof}
The proof follows immediately by setting 

\begin{equation}
\eta =\frac{ 1 - \lambda}{1 + \lambda} 
\end{equation}
   in Corollary 5. 
\end{proof} 
\begin{remark}
It  readily follows from \eqref{OP_a-l-m_1} that the CDF of the $\alpha{-}\lambda{-}\mu$ distribution can be expressed as, 

\begin{equation} \label{OP_a-l-m_1}
F_{\gamma}(\gamma) = \frac{ (-1)^{2 \mu}  \sqrt{\pi} (1 - \lambda)^{\mu} (1 + \lambda)^{\mu}}{\Gamma(\mu) 2^{\mu - \frac{1}{2}} \lambda^{2 \mu}} Ie_{\mu + \frac{1}{2} + \frac{4(1 - \alpha)}{\alpha^{2}}, \, \mu - \frac{1}{2}}\left(\frac{2 \lambda \mu \gamma^{\alpha{/}2}}{(\lambda^{2} - 1) \overline{\gamma}^{\alpha{/}2}} ,  - \frac{1}{\lambda} \right)  
\end{equation}
whereas for the specific case that $\alpha = 2$, equation \eqref{OP_a-l-m_1} yields a closed-form expression for the  OP over $\lambda{-}\mu$ fading channels, namely, 

\begin{equation} \label{OP_a-l-m_2}
P_{\rm out} =  \frac{ (-1)^{2 \mu}  \sqrt{\pi} (1 - \lambda)^{\mu} (1 + \lambda)^{\mu}}{\Gamma(\mu) 2^{\mu - \frac{1}{2}} \lambda^{2 \mu}} Ie_{\mu - \frac{1}{2}, \, \mu - \frac{1}{2}}\left(\frac{2 \lambda \mu \gamma_{th}}{(\lambda^{2} - 1) \overline{\gamma}} ,  - \frac{1}{\lambda} \right)  
\end{equation}
which holds without restrictions on the value of $\lambda$ and $\mu$. 
\end{remark}

\subsection{Outage Probability over $\alpha{-}\kappa{-}\mu$  Fading Channels }

The $\alpha{-}\kappa{-}\mu$ distribution was also proposed as a remarkably accurate model for accounting for \textit{small scale} fading conditions. Its foundation is similar to that of $\alpha{-}\eta{-}\mu$ and $\alpha{-}\lambda{-}\mu$ distributions but it is differentiated in that it is complementary to these models while it characterizes efficiently \textit{line-of-sight} (LOS) communication scenarios. This is explicitly illustrated in \cite[Fig. 1]{Yacoub_NL_1}  which demonstrates the whole range of modelling capabilities of the aforementioned non-linear fading models. 
Thr  $\alpha{-}\kappa{-}\mu$ distribution includes as special cases the $\kappa{-}\mu$, $\alpha{-}\mu$, Rice, Nakagami${-}m$ and Rayleigh distributions, while in terms of physical interpretation, $\kappa$ denotes the   ratio between the in-phase dominant component and the quadrature dominant component, whereas $\alpha$ and $\mu$ parameters are defined as in  $\alpha{-}\eta{-}\mu$ and $\alpha{-}\lambda{-}\mu$ distributions \cite{Yacoub_NL_1}.  
\begin{definition}
For  $\alpha, \kappa, \mu,  \rho \in \mathbb{R}^{+}$, the normalized envelope  PDF of the  $\alpha{-}\kappa{-}\mu$ distribution is expressed as  follows,

\begin{equation} \label{OP_a-k-m_2}
p_{P}(\rho) =  \frac{\alpha \mu (1 + \kappa)^{\frac{\mu + 1}{2}}  I_{\mu - 1}\left( 2\mu \sqrt{\kappa (1 + \kappa)} \rho^{\alpha {/} 2} \right) }{ \kappa^{\frac{\mu - 1 }{2}}  e^{\kappa \mu } \rho^{1 - \frac{\alpha (1 + \mu)}{2} }  e^{ \mu(1 + \kappa)\rho^{a}}  }.  
\end{equation}
\end{definition}
\begin{corollary}
For  $\alpha$, $\kappa$, $\mu$, $\overline{\gamma}$, $\gamma_{th} \in \mathbb{R}^{+} $, the OP  over i.i.d $\alpha{-}\kappa{-}\mu$ fading channels can be expressed as follows,

\begin{equation} \label{OP_a-k-m_1}
P_{\rm out} =  T_{ \sqrt{\mu (1 + \kappa) \gamma_{th}^{\alpha{/}2} {/} \overline{\gamma}^{\alpha{/}2}} }\left(2\mu - 1, \mu - 1, \sqrt{\kappa \mu} \right). 
\end{equation}
\end{corollary}
\begin{proof}
Based on \eqref{OP_a-k-m_2}, the SNR  PDF of the $\alpha{-}\kappa{-}\mu$ distribution is expressed as \cite[eq. (6)]{Yacoub_NL_1},  

\begin{equation} \label{OP_a-k-m_3}
p_{\gamma}(\gamma) =     \frac{\alpha \mu (1+ \kappa)^{\frac{1 + \mu}{2}}}{2 \kappa^{\frac{\mu - 1}{2}} e^{\mu \kappa}}  \frac{\gamma^{\frac{\alpha(1 + \mu)}{4}-1}}{\overline{\gamma}^{\frac{\alpha(1 + \mu)}{4}}}  e^{- \mu (1 + \kappa) \frac{\gamma^{\alpha{/}2}}{\overline{\gamma}^{\alpha{/}2}}}I_{\mu - 1}\left( 2\mu \sqrt{\kappa (1 + \kappa)\frac{\gamma^{a{/}2}}{\overline{\gamma}^{\alpha{/}2}}}  \right).  
\end{equation}
Therefore, by substituting \eqref{OP_a-k-m_3} into \eqref{OP_a-h-m_4} it immediately follows that,   

\begin{equation} \label{OP_a-k-m_4}
P_{\rm out} =      \frac{\alpha \mu (1+ \kappa)^{\frac{1 + \mu}{2}}}{2 \kappa^{\frac{\mu - 1}{2}} e^{\mu \kappa} \,  \overline{\gamma}^{\frac{\alpha(1 + \mu)}{4}} }   \int_{0}^{\gamma_{th}}  \frac{\gamma^{\frac{\alpha(1 + \mu)}{4}-1}}{ e^{ \mu (1 + \kappa) \frac{\gamma^{\alpha{/}2}}{\overline{\gamma}^{\alpha{/}2}}}}  I_{\mu - 1}\left( 2\mu \sqrt{\kappa (1 + \kappa)\frac{\gamma^{a{/}2}}{\overline{\gamma}^{\alpha{/}2}}}  \right)   {\rm d}\gamma
\end{equation}
 which upon performing a necessary change of variables and carrying out  long but basic algebraic manipulations, it can be  expressed as follows, 

\begin{equation} \label{OP_a-k-m_5}
P_{\rm out} =    \frac{2 e^{- \kappa \mu}}{(\kappa \mu)^{\frac{\mu - 1}{2}}  } \int_{0}^{\sqrt{\mu (1 + \kappa) \gamma_{th}^{\alpha{/}2} {/} \overline{\gamma}^{\alpha{/}2}}} \frac{\gamma^{\mu}}{ e^{\gamma^{2}}} I_{\mu - 1} \left( 2  \sqrt{\kappa \mu} \gamma \right) {\rm d}\gamma.   
\end{equation}
It is evident that  \eqref{OP_a-k-m_5} can be equivalently expressed  as, 

\begin{equation} \label{OP_a-k-m_6}
P_{\rm out} =     2(\sqrt{\mu \kappa})^{(\mu - 1) - (2 \mu -1) + 1} e^{- \mu \kappa} \int_{0}^{\sqrt{\mu (1 + \kappa) \gamma_{th}^{\alpha{/}2} {/} \overline{\gamma}^{\alpha{/}2}} } \frac{\gamma^{2\mu - 1}}{\gamma^{1 - \mu} e^{\gamma^{2}} }I_{\mu - 1}\left(2 \sqrt{\mu \kappa} \gamma \right) {\rm d}\gamma  
\end{equation}
and thus,  the above representation can be expressed in closed-form in terms of the incomplete Toronto function. As a result,  equation \eqref{OP_a-k-m_1} is deduced, which completes the proof. 
\end{proof}
\begin{remark}
For the special case that $\alpha = 2$, equation \eqref{OP_a-k-m_1} reduces to the following closed-form expression for the OP over $\kappa{-}\mu$ fading channels, 

\begin{equation}   \label{OP_k-m}
P_{\rm out} =  T_{ \sqrt{\mu (1 + \kappa) \gamma_{th} {/} \overline{\gamma}} }\left(2\mu - 1, \mu - 1, \sqrt{\kappa \mu} \right). 
\end{equation}
which to the best of the Authors knowledge, it has not been previously reported in the open technical literature. 
\end{remark}

\subsection{Alternative Representations for the Outage Probability over $\eta{-}\mu$ and $\lambda{-}\mu$ Fading Channels} 

The $\eta{-}\mu$ fading model has been used extensively in the analysis of conventional and emerging communication systems over generalized multipath fading channels. The corresponding OP was firstly addressed in \cite{J:Jimenez2010, J:Paris2012_CommL, J:Paris} for specific cases. In what follows, we derive exact closed-form expressions for the $\eta{-}\mu$ and $\lambda{-}\mu$ fading models which are valid for both integer and half-integer values of $\mu$. 
 \begin{corollary}
For  $\eta$, $\mu$, $\overline{\gamma}$, $\gamma_{th} \in \mathbb{R}^{+} $, and $2\mu \in \mathbb{N}$, the OP  over i.i.d $ \eta{-}\mu$ fading channels can be expressed as follows,

\begin{equation} \label{OP_h-m_Alt}
\begin{split}
P_{out} &= \frac{2\sqrt{\pi} \eta^{\mu} \Gamma(2 \mu)\, _{2}F_{1}\left(\mu, \mu + \frac{1}{2}, \mu + \frac{1}{2}, \frac{(1 - \eta)^{2}}{(1 + \eta)^{2}} \right)}{\Gamma(\mu) \Gamma \left(\mu + \frac{1}{2} \right) (1 + \eta)^{2 \mu}} \\
&- \sum_{l = 0}^{2\mu - 1} \binom{2\mu - 1}{l} \frac{\sqrt{\pi} l! (1 + \eta)^{2\mu - l - 1} \gamma_{th}^{2\mu - l - 1} }{\overline{\gamma}^{2 \mu - l - 1}  \mu^{l - 2\mu }  \eta^{\mu - l - 1} \mu! \Gamma\left( \mu + \frac{1}{2}\right)  }  \frac{  \Phi_{1}\left(\mu, 1 + l, 2\mu; 1 - \eta, \frac{\mu (1 - \eta^{2})\gamma_{th}}{\overline{\gamma} \eta} \right)  }{ 2^{1 - 2\mu } e^{-\frac{\mu (1 + \eta) \gamma_{th}}{\overline{\gamma}\eta}}}
\end{split}
\end{equation}
 \end{corollary}
 \begin{proof}
 The proof follows  with the aid of Theorem 8 and Corollary 5. 
 \end{proof}
\begin{corollary}
For  $\mu$, $\overline{\gamma}$, $\gamma_{th} \in \mathbb{R}^{+} $, $-1 < \lambda < 1$ and $2\mu \in \mathbb{N}$, the OP  over i.i.d $ \lambda{-}\mu$ fading channels can be expressed as,

\begin{equation} \label{OP_l-m_Alt}
\begin{split}
P_{out} &= \frac{\sqrt{\pi} (1 - \lambda)^{\mu} (1 + \lambda)^{\mu} \Gamma(2\mu) \, _{2}F_{1}\left(\mu, \mu + \frac{1}{2}, \mu + \frac{1}{2}, \lambda^{2} \right) }{\Gamma(\mu) \Gamma\left( \mu + \frac{1}{2} \right) 2^{2\mu - 1} }  \\
& - \sum_{l = 0}^{2 \mu - 1} \binom{2\mu - 1}{l} \frac{\sqrt{\pi} l! \mu^{2\mu - l}   e^{- \frac{2 \mu \gamma_{th}  }{\overline{\gamma}(1 - \lambda)}} }{ (1 + \lambda)^{\mu} (1 - \lambda)^{\mu - l - 1}m! } \frac{ \Phi_{1}\left(\mu, 1 + l, 2\mu, \frac{2\lambda}{1 + \lambda}, \frac{4 \mu \lambda \gamma_{th}}{\overline{\gamma} (1 + \lambda)(1 - \lambda)} \right) }{ \overline{\gamma}^{2\mu - l - 1}  \gamma_{th}^{1 + l - 2 \mu } \Gamma \left(\mu + \frac{1}{2} \right) 2^{l}}. 
\end{split}
\end{equation}
 \end{corollary}
 \begin{proof}
 The proof follows immediately  by setting 
 
 \begin{equation}
 \eta = \frac{1 - \lambda}{1 + \lambda} 
 \end{equation}
 in Corollary 8. 
 \end{proof}

\subsection{ Truncated Channel Inversion with Fixed Rate Transmission over    Rician Fading Channels}

\begin{corollary}
For $n, \gamma_{0}, \gamma_{th}, \overline{\gamma}, B \in \mathbb{R}^{+}$, the spectral efficiency for truncated channel inversion with fixed-rate (TIFR)  policy over i.i.d. Rician fading channels can be expressed as follows, 

\begin{equation} \label{C-TIFR_1}
\frac{C_{\rm TIFR}}{B} =   {\rm log}_{2} \left(1 + \frac{\overline{\gamma}}{ 2 (1 + n^{2}) Q_{-1, 0} \left( n \sqrt{2}, \sqrt{ \frac{2 \gamma_{0} (1 + n^{2}) }{\overline{\gamma}} } \right) } \right) \left\lbrace 1 - T_{\sqrt{ \left( 1 + n^{2} \right) \gamma_{th} {/} \overline{\gamma} } } \left(1, 0, n \right) \right\rbrace  
\end{equation}
 where $n$ denotes the Nakagami${-}n$ parameter,   $B$ is the corresponding channel bandwidth and    $\gamma_{0}$ is the optimum cut-off SNR below which data transmission is suspended \cite{B:Alouini, J:Varaiya}. 
\end{corollary}

\begin{proof}
It is widely known that Rice distribution has been traditionally used for accounting for multipath fading in LOS communication scenarios. The corresponding SNR per symbol follows the non-central chi-square distribution with its PDF given by \cite[eq. (2.16)]{B:Alouini_2005},

\begin{equation} \label{OP_Rice_2}
p_{\gamma}(\gamma) = \frac{1 + n^{2}}{\overline{\gamma} e^{n^{2}}}  e^{ - (1 + n^{2}) \frac{\gamma}{\overline{\gamma}}}  I_{0} \left( 2n \sqrt{\frac{ (1 + n^{2}) \gamma }{\overline{\gamma}}} \right)   
\end{equation} 
where $n$ is related to the Ricean $K$ factor by $K = n^{2}$ and physically denotes the ratio of the LOS component to the average power of the scattered component \cite{B:Alouini}. It is also known that the inversion of the channel fading technique is based on adapting the transmitter power in order to maintain a constant SNR at the receiver. This technique often suffers a  capacity penalty which can be combated by inverting the channel fading only above a pre-determined fixed cut-off fade depth $\gamma_{0}$ \cite{B:Alouini}. Mathematically, the C-TIFR is given by \cite[eq. (15.36)]{B:Alouini},  namely, 

\begin{equation} \label{C-TIFR_2} 
C_{\rm TIFR} = B {\rm log}_{2} \left( 1 + \frac{1}{  \int_{\gamma_{0}}^{\infty} \frac{p_{\gamma}(\gamma)}{\gamma} d \gamma } \right) \left\lbrace 1 - P_{\rm out} \right\rbrace.  
\end{equation}
As a result, in the case of Rician fading one obtains straightforwardly, 

\begin{equation} \label{C-TIFR_3} 
  \int_{\gamma_{0}}^{\infty} \frac{p_{\gamma}(\gamma)}{\gamma} {\rm d}\gamma =  \frac{c}{e^{n^{2}}} \int_{\gamma_{0}}^{\infty} \frac{1}{\gamma} e^{-c \gamma} I_{0} \left(2n \sqrt{c \gamma} \right) {\rm d} \gamma 
\end{equation}
where 

\begin{equation}
c = \frac{(1 + n^{2})}{ \overline{\gamma}}.  
\end{equation}
Setting $y = \sqrt{2ax}$ and thus, $x = y^{2} {/} 2a$ and ${\rm d}y{/} {\rm d}x = \sqrt{a {/} 2x}$ and after some basic algebraic manipulations it follows that, 

\begin{equation} \label{C-TIFR_4}  
 \frac{c}{e^{n^{2}}} \int_{\gamma_{0}}^{\infty} \frac{1}{\gamma} e^{-c \gamma} I_{0} \left(2n \sqrt{c \gamma} \right) {\rm d} \gamma  = 2 \frac{1 + n^{2}}{\overline{\gamma}} \int_{\sqrt{ \frac{2  (1 + n^{2}) \gamma_{0} }{\overline{\gamma}} } }^{\infty}  \frac{ e^{- \frac{\gamma^{2} + 2n^{2}}{2} }}{\gamma} I_{0}\left( \sqrt{2} n \gamma \right) {\rm d} \gamma. 
\end{equation}
The above integral can be expressed in terms of the Nuttall $Q{-}$function. Furthermore, it is recalled that the Rice distribution constitutes a special case of the $\kappa{-}\mu$ distribution and thus, the corresponding OP can be readily deduced with the aid of \eqref{OP_k-m} yielding, 

\begin{equation} \label{OP_Rice_1}
P_{\rm out} = T_{\sqrt{ \left( 1 + n^{2} \right) \gamma_{th} {/} \overline{\gamma} } } \left(1, 0, n \right). 
\end{equation}
As a result, by substituting  \eqref{C-TIFR_4} and  \eqref{OP_Rice_1} in \eqref{C-TIFR_2} yields \eqref{C-TIFR_1}, which completes the proof. 
\end{proof}

The optimum cut-off fade depth below which the data transmission is suspended is given by \cite[eq. (15.5)]{B:Alouini},  namely, 

\begin{equation} \label{C-TIFR_5}
\underbrace{ \int_{\gamma_{0}}^{\infty} \frac{p_{\gamma}(\gamma)}{\gamma_{0}}  {\rm d} \gamma}_{\mathcal{I}_{12}} -    \underbrace{\int_{\gamma_{0}}^{\infty} \frac{p_{\gamma}(\gamma) }{\gamma} {\rm d} \gamma}_{\mathcal{I}_{11}}   \triangleq  1.   
\end{equation}
For the case of Rician fading we substitute \eqref{OP_Rice_2} in \eqref{C-TIFR_5} and by recalling that $\mathcal{I}_{11}$ can be expressed in closed-form according to \eqref{C-TIFR_4} it follows that, 

\begin{equation} \label{C-TIFR_6}  
 \underbrace{\frac{1 + n^{2}}{\overline{\gamma} e^{n^{2}}}\int_{\gamma_{0}}^{\infty} e^{- \frac{1 + n^{2}  }{\overline{\gamma}} \gamma } I_{0} \left(2n \sqrt{\frac{1 + n^{2}}{\overline{\gamma}} \gamma} \right) {\rm d} \gamma}_{\mathcal{I}_{12}}   - 2(1 + n^{2}) Q_{-1,0}\left(n \sqrt{2}, \sqrt{ \frac{ 2 \gamma_{0} (1 + n^{2})}{\overline{\gamma}}} \right)  = 1
\end{equation}
By setting

\begin{equation}
y = \sqrt{ \frac{2(1 + n^{2}) \gamma}{  \overline{\gamma}}} 
\end{equation}
 and thus, 

 \begin{equation}
 \gamma = \frac{\overline{\gamma} y^{2}}{ 2(1 + n^{2})}
 \end{equation}
and 

\begin{equation}
\frac{{\rm d}y}{ {\rm d}\gamma} = \sqrt{\frac{ 1 + n^{2}}{ 2 \gamma \overline{\gamma} }} 
\end{equation}
the $\mathcal{I}_{12}$ term can be  expressed in closed-form in terms of the Marcum $Q{-}$function, namely, 

\begin{equation} \label{C-TIFR_7}
\mathcal{I}_{12} =  Q_{1} \left( n \sqrt{2}, \sqrt{ \frac{2 \gamma_{0} (1 + n^{2})}{\overline{\gamma}} } \right).  
\end{equation}
Therefore, by substituting \eqref{C-TIFR_7} in \eqref{C-TIFR_6} and after performing basic algebraic manipulations, the optimum cut-off SNR can be finally expressed as follows, 

\begin{equation} \label{Cut-off}
\gamma_{0} = \frac{\overline{\gamma} Q_{1} \left( n \sqrt{2}, \sqrt{ \frac{2 \gamma_{0} (1 + n^{2})}{\overline{\gamma}} } \right) }{\overline{\gamma} + 2 (1 + n^{2}) Q_{-1, 0} \left( n \sqrt{2}, \sqrt{ \frac{2 \gamma_{0} (1 + n^{2}) }{\overline{\gamma}} } \right) }. 
\end{equation}
The above expression can be used in determining $\gamma_{0}$ numerically with the aid of popular software packages such as MATLAB and MATHEMATICA.

\subsection{Truncated Channel Inversion with Fixed Rate Transmission in MIMO Systems over  Rician Fading Channels}

In   multiple-input multiple-output spatial multiplexing communications, truncated channel inversion with fixed rate can be applied to each eigen-mode in order to transform the fading eigen-modes into a set of parallel AWGN channels with the same average SNR \cite{J:Varaiya}. This is expressed as

 $$  \frac{1}{m \int_{\gamma_{0}}^{\infty} \frac{p_{\gamma}(\gamma)}{ \gamma} {\rm d} \gamma } $$ 
 where, $p_{\gamma}(\gamma)$ is the PDF of the  corresponding fading statistics and $m$ are the non-zero positive real eigenvalues of the non-central Wishart-type random matrix $\textbf{HH}^{\mathcal{H}}$, with $\mathcal{H}$ denoting the Hermitian operator. Also, $\gamma_{0}$ is the predetermined SNR threshold which is selected accordingly for either guaranteeing a required OP or for maximizing the achievable fixed transmission rate of the eigen-mode truncated channel inversion (em-ti) policy with  capacity:    

\begin{equation} \label{TIFR_General}
C_{\rm em-tifr}^{m, n} = m {\rm log}_{2}\left(1 +  \frac{1}{m \int_{\gamma_{0}}^{\infty} \frac{p_{\gamma}(\gamma)}{\gamma} {\rm d} \gamma }\right)  \int_{\gamma_{0}}^{\infty} p_{\gamma}(\gamma) {\rm d} \gamma    
\end{equation}
where $n$ are non-zero mean circularly symmetric Gaussian random variables whose sum denote the non-zero eigenvalue $\lambda$ \cite{J:Aissa, J:Alouini_Goldsmith}. 
%

\begin{corollary}
For $K, n, \gamma_{0}, \overline{\gamma} \in \mathbb{R}^{+}$ and $\gamma_{th} \in \mathbb{R}^{+} $, the em-tifr  capacity of MISO/SIMO communication systems over uncorrelated Rician fading channels can be expressed according to 

\begin{equation} \label{MISO_em-ti_1}
\begin{split}
C_{{\rm em-tifr}}^{1, n} =& {\rm log}_{2} \left( 1 + \frac{\overline{\gamma} \, (2 K \textbf{m}^{\mathcal{H}}\textbf{m})^{\frac{n-1}{2}}}{  2 (K+1)  Q_{n - 2, n - 1}\left( \sqrt{2 K \textbf{m}^{\mathcal{H}}\textbf{m}}, \sqrt{ \frac{2(K + 1)\gamma_{0}}{\overline{\gamma}} } \right)} \right)  \\
& \times  \left\lbrace 1 -T_{\sqrt{\frac{ (K+1) \gamma_{0}}{\overline{\gamma}}}} \left(2n - 1, n - 1, \sqrt{K \textbf{m}^{\mathcal{H}} \textbf{m}} \right) \right\rbrace 
\end{split}
 \end{equation}
with  $K$ denoting the Rician $K{-}$factor  and $\textbf{m}$ being the   $N{-}$dimensional deterministic vector that accounts for the corresponding LOS component. 
\end{corollary}

\begin{proof}
With the aid of the SNR  PDF for uncorrelated Rician fading channels in  \cite[eq. (29)]{J:Aissa} and recalling that 

\begin{equation}
\int_{\gamma_{0}}^{\infty} p_{\gamma}(\gamma) {\rm d} \gamma = 1 - P_{\rm out}
\end{equation}
 it immediately follows that,  

\begin{equation} \label{Rice_pdf_SIMO}
P_{\rm out}^{1, n}  = \frac{(K+1)^{\frac{n+1}{2}} e^{-K \textbf{m}^{\mathcal{H}}\textbf{m} } }{ \overline{\gamma}^{\frac{n+1}{2}} \left( K \textbf{m}^{\mathcal{H}} \textbf{m} \right)^{\frac{n - 1}{2}}} \int_{0}^{\gamma_{0}}  \frac{ \gamma ^{\frac{n - 1}{2}}}{ e^{   \frac{(K + 1) \gamma}{2}}}  I_{n - 1} \left( 2 \sqrt{\frac{ (K + 1) K\textbf{m}^{\mathcal{H}}\textbf{m} \gamma}{\overline{\gamma}}} \right) {\rm d} \gamma. 
\end{equation}
The above representation can be expressed in terms of the incomplete Toronto function, namely,   

\begin{equation} \label{Rice_OP_SIMO}
P_{\rm out}^{1, n}  = T_{\sqrt{\frac{ (K+1) \gamma_{0}}{\overline{\gamma}}}} \left(2n - 1, n - 1, \sqrt{K \textbf{m}^{\mathcal{H}} \textbf{m}} \right). 
\end{equation}
By substituting \eqref{Rice_OP_SIMO} in \cite[eq. (38)]{J:Aissa}, one obtains \eqref{MISO_em-ti_1}, which completes the proof. 
\end{proof}

The optimal cutoff threshold for \eqref{MISO_em-ti_1} has to satisfy \cite[eq. (34)]{J:Aissa}, namely, 

\begin{equation} \label{optimal_cutoff}
\gamma_{0} = Q_{n}\left( \sqrt{2K \textbf{m}^{\mathcal{H}}\textbf{m}}, \sqrt{2 \mu_{K} \gamma_{0}} \right) - \frac{2^{\frac{3 - n}{2}}\mu_{K} \gamma_{0}}{( K \textbf{m}^{\mathcal{H}}\textbf{m})^{\frac{n-1}{2}}}  Q_{n-2,n-1}\left( \sqrt{2K \textbf{m}^{\mathcal{H}}\textbf{m}}, \sqrt{2 \mu_{K} \gamma_{0}} \right) 
\end{equation}
where 

\begin{equation}
\mu_{K} = \frac{K + 1}{\overline{\gamma}}  
\end{equation}
 The above expression can be further elaborated and  an exact closed-form expression for $\gamma_{0}$ can be deduced. 

\begin{lemma}
For $K, \overline{\gamma}, n \in \mathbb{R}^{+}$ and with  $ Q_{m,n}^{-1}(a,b)$ denoting the inverse  Nuttall $Q{-}$function, the following closed-form expression holds for the optimal cut-off threshold in \eqref{MISO_em-ti_1},

\begin{equation} \label{cutoff_1}
 \gamma_{0}  =  \frac{\left[ Q_{n-2,n-1}^{-1}\left(\sqrt{2K \textbf{m}^{\mathcal{H}}\textbf{m}}, - \frac{(2K \textbf{m}^{\mathcal{H}}\textbf{m})^{\frac{n-1}{2}}}{\mu_{K}} \right) \right]^{2}}{2\mu_{K}}. 
\end{equation} 
\end{lemma}

\begin{proof}
By taking the first derivative of \eqref{optimal_cutoff} w.r.t. $\gamma_{0}$ it immediately follows that, 

\begin{equation} \label{cutoff_2}
 \frac{\partial Q_{n}\left(\mathcal{A}, \sqrt{2 \mu_{K} \gamma_{0}} \right)}{\partial \gamma_{0}}  - \frac{2\mu_{K} \gamma_{0}}{\mathcal{A}^{n-1}}  \frac{\partial  Q_{n-2,n-1}\left(\mathcal{A}, \sqrt{2 \mu_{K} \gamma_{0}} \right)}{\partial \gamma_{0}} - \frac{2 \mu_{K}}{\mathcal{A}^{n-1}} Q_{n-2,n-1}\left(\mathcal{A}, \sqrt{2 \mu_{K} \gamma_{0}} \right) = 1 
\end{equation}
where 

\begin{equation}
\mathcal{A} = \sqrt{2K \textbf{m}^{\mathcal{H}}\textbf{m}}. 
\end{equation}
 After performing the above derivatives it   follows that, 

\begin{equation}\label{cutoff_3}
\mathcal{A}^{1 - n}  \frac{2\mu_{K} \gamma_{0} \,  I_{n-1}(\mathcal{A}\sqrt{\mu_{K} \gamma_{0}}) }{ (2\mu_{K} \gamma_{0})^{1 - \frac{n}{2}} e^{ \frac{2 \mu_{K} \gamma_{0} + \mathcal{A}^{2}}{2}} } - \frac{2\mu_{K}}{\mathcal{A}^{n-1}} Q_{n-2,n-1} \left(\mathcal{A}, \sqrt{2\mu_{K} \gamma_{0}} \right) - \frac{ (2\mu_{K} \gamma_{0})^{ \frac{n}{2}} I_{n-1}(\mathcal{A}\sqrt{\mu_{K} \gamma_{0}}) }{\mathcal{A}^{n-1}   e^{ \frac{2 \mu_{K} \gamma_{0} + \mathcal{A}^{2}}{2}} }   = 1
\end{equation}
which after some basic algebraic manipulations becomes,   

\begin{equation} \label{cutoff_4}
 Q_{n-2,n-1} \left(\mathcal{A}, \sqrt{2\mu_{K} \gamma_{0}} \right) = - \frac{\mathcal{A}^{n-1}}{ \mu_{K}}.
\end{equation}
With the aid of the inverse Nuttall $Q{-}$function, it immediately follows that,  

\begin{equation} \label{cutoff_5}
\sqrt{2\mu_{K} \gamma_{0}}  =  Q_{n-2,n-1}^{-1}\left(\sqrt{2K \textbf{m}^{\mathcal{H}}\textbf{m}}, - \frac{\mathcal{A}^{n-1}}{\mu_{K}} \right). 
\end{equation}
As a result, by solving w.r.t $\gamma_{0}$ one obtains   \eqref{cutoff_1}  which completes the proof. 
\end{proof}

 In MIMO communication scenarios over uncorrelated Rician fading channels, the random matrix $\textbf{HH}^{\mathcal{H}}$ follows a Wishart type distribution. The corresponding PDF of a single unordered eigenvalue $\lambda $ was   given in \cite[Corollary 1]{J:Lozano_New}, and then in \cite[eq. (69)]{J:Aissa}, which with the aid of $ \gamma = \lambda\overline{\gamma}$ it is expressed as, 

\begin{equation} \label{Rice_pdf_MIMO_1} 
p_{\gamma}(\gamma)  =  \frac{K_{m, n}^{\omega_{t}}}{m} \sum_{i = 1}^{m}   \sum_{j=1}^{m-t}  \frac{ c_{ij}^{(t)} \gamma^{d  + i + j - 2}}{  \mu_{K}^{1 - d - i - j}  e^{\mu_{K} \gamma}} + K_{m, n}^{\omega_{t}} \sum_{i = 1}^{m} \sum_{j = m - t + 1}^{m} \frac{c_{ij}^{(t)} \mu_{K} ^{ d +  i }  \,_{0}F_{1} (d + 1;\mu_{K} \gamma \omega_{j}) }{    m \, d!  \gamma^{1 - d  - i } e^{\mu_{K} \gamma}}
\end{equation} 
where $d = n - m$ whereas $c_{ij}^{(t)}$ and $K_{m, n}^{\omega_{t}}$ are given by \cite[eq. (69)]{J:Aissa} and  \cite[eq. (71)]{J:Aissa}, respectively. Also, $\omega_{m-t+1},\dots, \omega_{m}$ are $t$ distinct eigenvalues of the non-centrality parameter of the distribution that can be represented  in the form of the following as a column vector, 

\begin{equation}
\omega_{t} = [\omega_{m-t+1}, \dots, \omega_{m}]^{T}. 
\end{equation}
To this effect, the overall capacity of the em-ti transmission policy in MIMO systems is given by \cite[eq. (47)]{J:Aissa}, namely, 

\begin{equation} \label{MIMO_em-ti_1}
C^{m, n}_{{\rm em-ti}} = m {\rm log}_{2}\left(1 + \kappa^{m, n}_{{\rm em-ti}} (\gamma_{0}) \right) \times \left( 1 - P_{\rm out}^{\rm em} \right) 
\end{equation}
where  $\kappa^{m, n}_{{\rm em-ti}}(\gamma_{0}) $, $P_{\rm out}^{\rm em} $ are given in \eqref{MIMO_em-ti_2} and \eqref{MIMO_em-ti_3}, respectively, along with the optimal cut-off  $\gamma_{0}$   in \cite[eq. (44)]{J:Aissa} which is expressed as,
 
\begin{equation} \label{MIMO_em-ti_2}
\kappa^{m, n}_{{\rm em-ti}}(\gamma_{0}) = \frac{\mu_{K}{/}K_{m,n}^{\omega_{t}}}{ \sum\limits_{i = 1}^{m} \left[\sum\limits_{j=1}^{m-t}  c_{ij}^{(t)}\Gamma(d+i+j-2,\mu_{K} \gamma_{0}) + \sum\limits_{j=m-t+1}^{m}c_{ij}^{(t)}\frac{   Q_{d+2i-3, d} \left(\sqrt{2 \omega_{j}}, \sqrt{2 \mu_{K} \gamma_{0}} \right) }{ 2^{i-2 + \frac{d}{2}} \omega_{j}^{d{/}2} e^{-\omega_{t}}}   \right]}
\end{equation}
and

\begin{equation} \label{MIMO_em-ti_3}
P_{\rm out}^{\rm em} = 1 -  \sum _{i = 1}^{m} \left[\sum _{j=1}^{m-t} \frac{\Gamma(d+i+j-2,\mu_{K} \gamma_{0})}{m \left[ K_{m,n}^{\omega_{t}} c_{ij}^{(t)} \right]^{-1}} + \sum _{j=m-t+1}^{m} \frac{  Q_{d+2i-3, d} \left(\sqrt{2 \omega_{j}}, \sqrt{2 \mu_{K} \gamma_{0}} \right) e^{ \omega_{t}} }{m \left[K_{m,n}^{\omega_{t}} c_{ij}^{(t)}\right]^{-1} 2^{i-2 + \frac{d}{2}} \omega_{j}^{d{/}2} }   \right]
\end{equation}

\begin{equation}  \label{MIMO_em-ti_4} 
\begin{split}
\gamma_{0} &=  \sum_{i = 1}^{m} \sum_{j=1}^{m-t} K_{m, n}^{\omega_{t}} c_{ij}^{(t)} \left[ \Gamma(d + i + j - 1, \mu_{K} \gamma_{0}) - \mu_{K} \gamma_{0} \Gamma(d + i + j - 1, \mu_{K} \gamma_{0}) \right] \\
& + \sum_{i=1}^{m} \sum_{j=m-t+1}^{m}  \frac{ Q_{d + 2i - 1, d} \left( \sqrt{2 \omega_{j}}, \sqrt{2\mu_{K} \gamma_{0}}  \right) - \sqrt{2\mu_{K} \gamma_{0}}  Q _{d + 2i - 2, d} \left( \sqrt{2 \omega_{j}}, \sqrt{2 \mu_{K} \gamma_{0}}\right)}{\left[ K_{m, n}^{\omega_{t}} c_{ij}^{(t)} \right]^{-1}   2^{i -1 +  \frac{d}{2}} \omega_{j}^{d{/}2}  e^{-\omega_{j}}}.
\end{split}
\end{equation}
Evidently, the performance measures in \eqref{MIMO_em-ti_2}${-}$\eqref{MIMO_em-ti_4} can be computed accurately and straightforwardly with the aid of the proposed expressions for the $ Q_{m,n}(a,b)$ function in Sec.~II.

\section{Conclusions}
New analytic expressions were derived for a set of important special functions in wireless communication theory, namely, the Nuttall $Q{-}$function, the incomplete Toronto function, the Rice $Ie{-}$function and the incomplete Lipschitz Hankel integrals. These expressions include closed-form expressions for general and specific cases as well as tight upper and lower bounds, polynomial representations and approximations. Explicit relationships in terms of these functions were also provided for specific cases of the Humbert $\Phi_1$ and Kamp$\acute{e}$ de F${\it \acute{e}}$riet special functions. The derived expressions  are rather useful both analytically and computationally because although the considered functions have been used widely in   analyses relating to wireless communications, they are neither tabulated nor  built-in functions in popular mathematical software packages such as MATLAB, MATHEMATICA and MAPLE. As an  example, the offered results were  indicatively employed  in deriving  novel analytic expressions for the outage probability over  $\alpha{-}\eta{-}\mu$, $\alpha{-}\lambda{-}\mu$ and $\alpha{-}\kappa{-}\mu$  fading channels as well as for the  truncated capacity with channel inversion in single-antenna and multi-antenna   communications under    Rician multipath fading conditions.

\section*{Appendix A \\ Proof of Theorem 1}

Equation \eqref{Nuttall_1} can be alternatively  written as,

\begin{equation} \label{KdF_2}
Q_{m,n}(a,b) = \underbrace{  e^{- \frac{ a^{2}}{2} }  \int_{0}^{\infty} x^{m} e^{- \frac{ x^{2}}{2} } I_{n}(ax) {\rm d}x}_{\mathcal{G}} -  e^{- \frac{ a^{2}}{2}  }\int_{0}^{b} x^{m} e^{- \frac{ x^{2}}{2} } I_{n}(ax) {\rm d}x. 
\end{equation}
By utilizing \cite[eq. (8.406.3)]{B:Tables} and \cite[eq. (6.621.1)]{B:Tables} in $\mathcal{G}$ it follows that, 

\begin{equation}  \label{KdF_3}
\mathcal{G} = \frac{ a^{n} \Gamma\left(  \frac{m + n + 1}{2} \right) \, _{1}F_{1} \left( \frac{m + n + 1}{2}, 1 + n, \frac{a^{2}}{2} \right)}{n!e^{\frac{a^{2}}{2}} 2^{ \frac{n - m + 1}{2} }}. 
\end{equation}
Substituting \eqref{KdF_3} in \eqref{KdF_2} and expanding the $I_{n}(x)$ function according to \cite[eq. (8.445)]{B:Tables} one obtains, 

\begin{equation} \label{KdF_4}
Q_{m,n}(a,b) = \mathcal{G} - \sum_{l = 0}^{\infty} \frac{a^{n + 2l} e^{ - \frac{a^{2}}{2}}  }{l! \Gamma(n + l + 1) 2^{n + 2l}} \int_{0}^{b} x^{m + n + 2l} e^{ - \frac{x^{2}}{2} } {\rm d}x.  
\end{equation}
Both $I_{n}(x)$ and ${\rm exp}(x)$ are entire functions and the limits of \eqref{KdF_4} are finite. Thus, substituting \cite[eq. (1.211.1)]{B:Tables}  in \eqref{KdF_4}, the resulting integral can be straightforwardly evaluated analytically yielding 

\begin{equation} \label{KdF_5}
Q_{m,n}(a,b) = \mathcal{G}  -  \sum_{l = 0}^{\infty} \sum_{i = 0}^{\infty} \frac{(-1)^{i} a^{n + 2l} b^{m + 2l +2i + n + 1} e^{ - \frac{a^{2}}{2}} }{l! i! \Gamma(n + l + 1)2^{n + i + 2l} (m + 2l +2i + n + 1) }.  
\end{equation}
To this effect and using the Pochhammer symbol 

\begin{equation}
(a)_{n} = \frac{\Gamma(a + n)}{\Gamma(a)}
\end{equation}
 while  recalling that 
 
\begin{align}
a &= \frac{a!}{(a-1)!} \\
&= \frac{\Gamma(a+1)}{\Gamma(a)} 
\end{align}
  one obtains 

\begin{equation} \label{KdF_6}
Q_{m,n}(a,b) = \mathcal{G}  -  \frac{a^{n}b^{n}}{ 2^{n} e^{  \frac{a^{2}}{2} }}  \sum_{l = 0}^{\infty} \sum_{i = 0}^{\infty} \frac{(-1)^{i} a^{2l} b^{ m +  2l + 2i + 1} (m + 2i + n + 1)_{2l} (m + n + 1)_{2i} (m + n)!  }{l! i! (n + l)! 2^{ i + 2l} (m + 2i + n + 2)_{2l} (m + n + 2)_{2i} \Gamma(m + n + 2) }. 
\end{equation} 
which upon using the identity,

\begin{equation}
(2x)_{2n} = 2^{2n} \left( x \right)_{n} \left( x + \frac{1}{2} \right)_{n}
\end{equation}
 it leads to 

\begin{equation} \label{KdF_7}
Q_{m,n}(a,b) = \mathcal{G} -  \frac{ a^{n} e^{ - \frac{a^{2}}{2} }}{m + n + 1}  \sum_{l = 0}^{\infty} \sum_{i = 0}^{\infty} \frac{(-1)^{i} a^{2l} b^{m + n + 2l + 2i + 1}  \left( \frac{m + n + 1}{2} + i \right)_{l} \left( \frac{m + n + 1}{2} \right)_{i}    }{l! i!  (n + 1)_{l} \left(  \frac{m + n + 1}{2} + 1 + i \right)_{l} \left( \frac{m + n + 1}{2} + 1 \right)_{i}   2^{n + i + 2l}   }.   
\end{equation}
By subsequently expressing each term of the form $(a + m)_{n}$ as follows, 

\begin{equation}
 (a + m)_{n} = \frac{\Gamma(a + m + n)}{\Gamma(a + m)} = \frac{(a)_{m + n}}{(a)_{m}}
\end{equation}
 and  performing some basic algebraic manipulations \eqref{KdF_7} can be re-written according to

\begin{equation} \label{KdF_8}
Q_{m,n}(a,b) = \mathcal{G} -   \frac{ e^{- \frac{a^{2}}{2} } a^{n} b^{m + n + 1} }{(m + n + 1) n! 2^{n}} \sum_{l = 0}^{\infty} \sum_{i = 0}^{\infty} \frac{\left( \frac{m + n + 1}{2} \right)_{l + i}}{(n + 1)_{l} \left( \frac{m + n + 3}{2}  \right)_{l + i}} \frac{\left( \frac{a^{2} b^{2}}{4} \right)^{l}}{l!} \frac{ \left(- \frac{b^{2}}{2} \right)^{i}}{i!}. 
\end{equation}
Importantly, this double series representation can be expressed in terms of the KdF function \cite{B:Exton, B:Prudnikov}. Therefore, by performing the necessary change of variables and substituting in \eqref{KdF_8}, eq. \eqref{KdF_1} is deduced thus, completing the proof.

\section*{Appendix B \\ Proof of Proposition 1}

Tight polynomial approximations for the modified Bessel function of the first kind were derived by Gross \textit{et al} in \cite{J:Gross_2}. Based on this,  by substituting \cite[eq. (19)]{J:Gross_2} in \eqref{Nuttall_1}   one obtains, 

\begin{equation} \label{Nuttall_Polynomial_1}
Q_{m, n} (a, b) \simeq  \sum_{l = 0}^{p} \frac{\Gamma(p + l) p^{1 - 2l} a^{n + 2l} }{l! (p - l)! (n + l)! } \int_{b}^{\infty} \frac{x^{2l + m +n}}{ 2^{n + 2l} e^{ \frac{x^{2} + a^{2}}{2}}} {\rm d}x. 
\end{equation}
By recalling that,

\begin{equation}
\Gamma(x) \triangleq (x-1)!
\end{equation}
 which holds for when $x \in \mathbb{R}^{+}$, and expressing the above integral according to \cite[eq. (8.350.2)]{B:Tables} yields \eqref{Nuttall_Polynomial}.  To this effect and for the specific case that $\frac{m + n + 1}{2}  \in \mathbb{N}$, the $\Gamma(a, x)$ function can be expressed in terms of a finite series according to \cite[eq. (8.352.4)]{B:Tables}. Therefore, by performing the necessary change of variables and substituting in \eqref{Nuttall_Polynomial} yields \eqref{Nuttall_Integer_Indices}, which completes the proof.

\section*{Appendix C \\ Proof of Theorem 2}

By expressing the $I_{n}(x)$ function   according to \cite[eq. (8.467)]{B:Tables}  and substituting in \eqref{ITF_Definition} yields,  

\begin{equation} \label{ITF_1}
\begin{split}
T_{B}(m,n,r) =& \sum_{k = 0}^{n - \frac{1}{2}} \frac{r^{n - m - k + \frac{1}{2}} \left( n + k - \frac{1}{2} \right)!}{k! \, \sqrt{\pi} \left( n - k - \frac{1}{2} \right)! 2^{2k}  e^{r^{2}} }  \\
& \times  \left\lbrace  \int_{0}^{B} (-1)^{k} t^{m - n - k - \frac{1}{2}} e^{-t^{2}}  e^{2rt} {\rm d}t   + \int_{0}^{B}  (-1)^{n + \frac{1}{2}} t^{m - n - k - \frac{1}{2}} e^{- t^{2}} e^{-2rt}   {\rm d}t \right\rbrace  
\end{split}
\end{equation}
which can be equivalently expressed as follows, 

\begin{equation} \label{ITF_1b}
\begin{split}
T_{B}(m, n, r) =& \sum_{k = 0}^{n - \frac{1}{2}} \frac{\left( n + k - \frac{1}{2} \right)! r^{n - m -k + \frac{1}{2}}}{\sqrt{\pi} k! \left( n - k - \frac{1}{2} \right)! 2^{2k}}   \\ & \times  \left\lbrace  \underbrace{ \int_{0}^{B}   (-1)^{n + \frac{1}{2}} t^{L -k } e^{- (t + r)^{2}} {\rm d}t}_{\mathcal{I}_{3}} + \underbrace{\int_{0}^{B} (-1)^{k} t^{L - k} e^{-(t - r)^{2}} {\rm d}t}_{\mathcal{I}_{4}} \right\rbrace 
\end{split}
\end{equation}

\noindent
where $L = m - n - \frac{1}{2}$. The $\mathcal{I}_{3}$ and $\mathcal{I}_{4}$  can be also expressed in terms of  \cite[eq. (1.3.3.18)]{B:Prudnikov} which yields 

\begin{equation}  \label{ITF_2}
\begin{split}
T_{B}(m, n, r) =& \sum_{k = 0}^{n - \frac{1}{2}} \sum_{l = 0}^{L} \frac{ r^{-(2k +l)} \, \left(n + k - \frac{1}{2}\right)! \, (L - k)!  }{\sqrt{\pi} \, k! \, l!\left( n - k - \frac{1}{2} \right)! (L-k-l)!}    \\
& \times  \left\lbrace  (-1)^{m - l}   \underbrace{\int_{0}^{B + r} \frac{t^{l} e^{-t^{2}}}{ \,2^{2k}} {\rm d}t }_{\mathcal{I}_{5}} + (-1)^{k} \underbrace{\int_{0}^{B - r} \frac{t^{l} e^{-t^{2}}}{ \,2^{2k}} {\rm d}t}_{\mathcal{I}_{6}} \right\rbrace .
\end{split}
\end{equation}
 Importantly, the $\mathcal{I}_{5}$ and $\mathcal{I}_{6}$ integrals can be expressed in terms of the $\gamma(a,x)$ function. Hence, by making the necessary change of variables  and substituting in \eqref{ITF_2} yields \eqref{ITF_Closed}, which completes the proof.

\section*{Appendix D \\ Proof of Lemma 2}

By performing a change of variables and the integral limits in $T_{B}(m ,n, r)$ and utilizing \cite[eq. (8.406.3)]{B:Tables} and \cite[eq. (6.631.1)]{B:Tables} one obtains, 

\begin{equation} \label{Nuttall_Toronto_1}
T_{B}(m, n, r) = \frac{ \Gamma \left( \frac{m + 1}{2} \right) \, _{1}F_{1} \left( \frac{m + 1}{2}, 1 + n, r^{2}  \right) }{ r^{m - 2n - 1}  \Gamma(n + 1) e^{r^{2}} }- \frac{2 e^{ - r^{2}}  }{r^{m - n - 1}}   \int_{B}^{\infty} t^{m - n} e^{ - t^{2}} I_{n}( 2rt) {\rm d}t. 
\end{equation}
By setting $u = \sqrt{2} t$ and performing long but basic algebraic representations, the above integral can be expressed in closed-form in terms of the Nuttall $Q{-}$function which yields \eqref{Nuttall_Toronto}.  This completes the proof.

\section*{Appendix E \\ Proof of Theorem 4}

The ${\rm exp}(x)$ and $I_{n}(x)$ functions in \eqref{ITF_Definition} are entire and can be expanded since the integration interval is finite. To this end, by making the necessary variable transformation in \cite[eq. (1.211.1)]{B:Tables} and \cite[eq. (8.445)]{B:Tables}, respectively, and substituting in \eqref{ITF_Definition} it follows that, 

\begin{equation} \label{ITF/KdF_2}
T_{B}(m, n, r) = 2 r^{n - m + 1} e^{- r^{2}} \sum_{l = 0}^{\infty} \sum_{i = 0}^{\infty} \frac{(-1)^{i} r^{n + 2l} B^{m  + 2l + 2i + 1} }{ l! i! \Gamma(n + l + 1) (m + 2l + 2i + 1)}. 
\end{equation}
By recalling that 

\begin{align}
(x + y) &= \frac{(x + y)!}{(x + y - 1)!} \\
&= \frac{(x + y)!}{\Gamma(x + y )} 
\end{align}
and that

\begin{equation}
\Gamma(x + n) = (x)_{n} \Gamma(x)
\end{equation}
and subsequently substitute  in \eqref{ITF/KdF_2} yields  

\begin{equation} \label{ITF/KdF_3}
T_{B}(m, n, r) = 2 r^{n - m +1}   e^{ - r^{2}} \sum_{l = 0}^{\infty} \sum_{i = 0}^{\infty} \frac{(-1)^{i} r^{n + 2l} B^{m + 2l + 2i + 1} (m + 2i + 1)_{2l} (m + 1)_{2i} \Gamma(m+1) }{l! i! (n + 1)_{l} \Gamma(n+1) (m + 2i + 2)_{2l} (m + 2)_{2i} \Gamma(m + 2) }. 
\end{equation}
With the aid of the identity, 

\begin{equation}
(2x)_{2n} = 2^{2n} (x)_{n} \left(x + \frac{1}{2}\right)_{n} 
\end{equation}
and after some basic algebraic manipulations  \eqref{ITF/KdF_3} can be alternatively expressed according to  
 
\begin{equation} \label{ITF/KdF_4}
T_{B}(m, n, r) = \frac{2 r^{n - m + 1}  e^{-r^{2}} }{ (m+1)} \sum_{l = 0}^{\infty} \sum_{i = 0}^{\infty} \frac{(-1)^{i} r^{n + 2l} B^{m + 2l + 2i + 1} \left( \frac{m + 1}{2} + i \right)_{l} \left( \frac{m + 1}{2} \right)_{i}}{l! i! \Gamma(n + l + 1) \left( \frac{m + 1}{2} + 1 +i \right)_{l} \left( \frac{m + 1}{2} + 1 \right)_{i} }. 
\end{equation}
Notably, each term of the form $(x + i)_{l}$ can be equivalently expressed as follows, 

\begin{equation}
\frac{\Gamma(x + i + l)}{\Gamma(x + i)} =  \frac{(x)_{l + i}}{ (x)_{i}}. 
\end{equation}
To this effect, by substituting accordingly in \eqref{ITF/KdF_4} one obtains, 

\begin{equation} \label{ITF/KdF_5}
T_{B}(m, n, r) = \frac{2 r^{2n - m + 1} B^{m + 1}}{n! (m + 1) e^{r^{2}}} \sum_{l = 0}^{\infty} \sum_{i = 0}^{\infty} \frac{(-1)^{i} r^{2l} B^{2l + 2i} \left( \frac{m + 1}{2} \right)_{l + i} }{l! i! (n + 1)_{l} \left( \frac{m + 1}{2}  + 1 \right)_{l + i}}.
\end{equation}
The above series  can be expressed in closed-form in terms of the KdF Function in \cite{B:Tables, B:Exton, B:Prudnikov}. Therefore, by making the necessary change of variables and substituting in \eqref{ITF/KdF_5}, one obtains \eqref{ITF/KdF_1}, which completes the proof.

\section*{Appendix F \\  Proof of Theorem 5}

The $I_{n}(x)$ function is monotonically decreasing with respect to its order $n$. Therefore, for an arbitrary positive real quantity $a \in \mathbb{R}^{+}$, it can be claimed  straightforwardly that $I_{n \pm a}(x) \lessgtr I_{n}(x)$. By applying this identity in \eqref{Rice_Alternative} the $Ie(k,x)$ can be upper bounded as follows,  

\begin{equation} \label{Rice_UB_2}
Ie(k,x) < 1 - e^{-x}I_{0}(kx) + k \int_{0}^{x} e^{-t} I_{\frac{1}{2}}(kt){\rm d}t. 
\end{equation}

\noindent
With the aid of the closed-form expression in \cite[eq. (8.467)]{B:Tables} it follows that 

\begin{equation}
I_{\frac{1}{2}} (kt) =  \frac{e^{kt} - e^{-kt}}{\sqrt{2 \pi k t}}
\end{equation}
 By substituting this in \eqref{Rice_UB_2} one obtains,

\begin{equation} \label{Rice_UB_3}
Ie(k,x) < 1 - e^{-x}I_{0}(kx) + k \underbrace{\int_{0}^{x} e^{-t} \left[\frac{e^{kt} - e^{-kt}}{\sqrt{2\pi kt}}\right] {\rm d}t  }_{\mathcal{I}_{8}}. 
\end{equation}

\noindent
The $\mathcal{I}_{8}$ integral can be expressed in closed-form in terms of the ${\rm erf}(x)$ function, namely, 

\begin{equation} \label{Rice_UB_3}
\mathcal{I}_{8} = \sqrt{\frac{k}{2}} \left[ \frac{{\rm erf}(\sqrt{x}\sqrt{1 - k})}{\sqrt{1 - k}} - \frac{{ \rm erf}(\sqrt{x} \sqrt{1 + k})}{\sqrt{1 + k}} \right]. 
\end{equation}

\noindent
By substituting \eqref{Rice_UB_3} in \eqref{Rice_UB_2}, equation \eqref{Rice_UB_1} is deduced. Likewise, based on the aforementioned monotonicity property of the $I_{n}(x)$ function, it is easily shown that $I_{\frac{3}{2}}(x) < I_{1}(x)$. To this effect and by performing the necessary change of variables and substituting in \eqref{Rice_Alternative}, one obtains the following inequality,

\begin{equation} \label{Rice_LB_2}
Ie(k,x) > 1 - e^{-x}I_{0}(kx) + k \int_{0}^{x} e^{-t} I_{\frac{3}{2}}(kt){\rm d}t. 
\end{equation}

\noindent
It is noted that a similar inequality can be obtained by exploiting the monotonicity properties of the Marcum $Q{-}$function which is strictly increasing w.r.t  $m$. Based on this it follows that $Q_{1}(a,b) > Q_{0.5}(a,b)$, which upon substituting in \cite[eq. (2c)]{J:Pawula} yields, 

\begin{equation} \label{Rice_LB_3}  
Ie(k,x) > \frac{1}{\sqrt{1 - k^{2}}} \left[2Q_{\frac{1}{2}}(a, b) - e^{-x}I_{0}(kx) - 1 \right]. 
\end{equation}

\noindent 
Importantly, according to \cite[eq. (27)]{J:Karagiannidis},  

\begin{equation}
Q_{0.5}(a,b) = Q(b+a) + Q(b-a). 
\end{equation}
 Therefore, by substituting in \eqref{Rice_LB_3} and applying the identity 

 \begin{equation}
 {\rm erf}(x) = 1 - 2Q(x \sqrt{2}) 
 \end{equation}
 equation \eqref{Rice_LB_1} is deduced, which completes the proof.

\section*{Appendix G \\  Proof of Theorem 6}

By changing the integral limits in \eqref{Rice_Definition} and expressing the $I_{n}(x)$ function according to \cite[eq. (9.238.2)]{B:Tables} it follows that,

 \begin{equation} \label{Rice_Humbert_2}
Ie(k, x) =  \frac{1}{\sqrt{1 - k^{2}}} - \int_{x}^{\infty} e^{-t(1 + k)} \,_{1}F_{1} \left( \frac{1}{2}, 1, 2kt \right) {\rm d}t
\end{equation}
where the

 $$\int_{0}^{\infty} {\rm exp}(-t) I_{0}(kt) {\rm d}t$$
  integral was expressed in closed-form with the aid of \cite[eq. (8.406.3)]{B:Tables} and \cite[eq. (6.621.1)]{B:Tables}. By subsequently setting 
 
  \begin{equation}
  u = 2kt - 2kx
  \end{equation}
  and therefore, $t = (u + 2kx){/}2k$ and ${\rm d}u{/}{\rm d}t = 2k$, it follows that, 

\begin{equation} \label{Rice_Humbert_3}
Ie(k, x) =  \frac{1}{\sqrt{1 - k^{2}}}- \frac{e^{-(1+k)x}}{2k} \int_{0}^{\infty} e^{-\frac{(1+k)}{2k}u} \,_{1}F_{1} \left( \frac{1}{2}, 1,   u + 2kx \right) {\rm d}u. 
\end{equation}
The above integral can be expressed in terms of the confluent Appell function or Humbert function $\Phi_{1}$ with the aid of \cite[eq. (3.35.1.9)]{B:Prudnikov_4}. Based on this, by making the necessary variable transformation and substituting in \eqref{Rice_Humbert_3} yields \eqref{Rice_Humbert}   thus, completing the proof.

\section*{Appendix H \\ Proof of Theorem 7}

By making the necessary change of variables in \cite[eq. (8.467)]{B:Tables} and substituting in \eqref{ILHIs_Definition} one obtains,  

\begin{equation} \label{ILHIs_CF_2}
\begin{split}
Ie_{m,n}(x;a) =& \sum_{k=0}^{n-\frac{1}{2}} \frac{\left(n+k- \frac{1}{2}\right)! 2^{-k - \frac{1}{2}}}{\sqrt{\pi} k!\left(n-k-\frac{1}{2} \right)!}  \\
& \times \left\lbrace (-1)^{k}\int_{0}^{x} y^{\mathcal{P} } e^{y} e^{-ay} {\rm d}y + (-1)^{n + \frac{1}{2}} \int_{0}^{x} x^{\mathcal{P} } e^{-y} e^{ay} {\rm d}y \right\rbrace  
\end{split}
\end{equation}

\noindent
where $\mathcal{P} = m-k+\frac{1}{2}$. Both integrals in \eqref{ILHIs_CF_2} can be expressed in closed-form in terms of the $\gamma(a,x)$ function. Hence, after basic algebraic manipulations one obtains \eqref{ILHIs_CF_1}, which completes the proof.

\section*{Appendix I \\ Proof of Theorem 8}

The $I_{n}(x)$ based representation in \eqref{ILHIs_Definition}  can be equivalently expressed as follows, 

\begin{equation} \label{m+n_Integer_2}
Ie_{m, n}(x;a) = \int_{0}^{\infty} y^{m} e^{-ay}I_{n}(y){\rm d}y - \int_{x}^{\infty} y^{m} e^{-ay}I_{n}(y){\rm d}y. 
\end{equation}
The first integral in \eqref{m+n_Integer_2} can be expressed in closed-form according to \cite[eq. (8.406.3)]{B:Tables} and \cite[eq. (6.621.1)]{B:Tables}. To this effect and by re-writing  the second integral by applying \cite[eq. (9.238.2)]{B:Tables} one obtains, 

\begin{equation} \label{m+n_Integer_3}
Ie_{m,n}(x;a)  = \frac{(m+n)!\,_{2}F_{1}\left( \frac{m + n + 1}{2}, \frac{m+n}{2} +1; 1 + n; \frac{1}{a^{2}} \right) }{2^{n} a^{m + n + 1} n!}-  \frac{1}{2^{n}n!} \int_{x}^{\infty} \frac{\,_{2}F_{1} \left( n + \frac{1}{2}, 1 + 2n, 2y\right)}{ y^{-m-n} e^{y(1+a)}}  {\rm d}y. 
\end{equation}
The above integral can be  expressed as,

\begin{equation} \label{m+n_Integer_4}
\int_{x}^{\infty} \frac{y^{m+n}}{e^{y(1+a)}} \,_{2}F_{1} \left( n + \frac{1}{2}, 1 + 2n, 2y\right) {\rm d}y  =   \int_{0}^{\infty} \frac{(y+x)^{m+n} \,_{2}F_{1} \left( n + \frac{1}{2}, 1 + 2n, 2(x + y)\right) }{ e^{y(1+a)} e^{x(1+a)}}  {\rm d}y
\end{equation}
By expanding the $(y+x)^{m+n}$ term according to \cite[eq. (1.111)]{B:Tables} and substituting in  \eqref{m+n_Integer_3} yields 

\begin{equation} \label{m+n_Integer_5}
\begin{split}
Ie_{m,n}(x;a)   &= \frac{(m+n)!\,_{2}F_{1}\left( \frac{m + n + 1}{2}, \frac{m+n}{2} +1; 1 + n; \frac{1}{a^{2}} \right) }{2^{n} a^{m + n + 1} n!}    \\
&  -  \frac{e^{-x(1+a)}}{2^{n}n!} \sum_{l = 0}^{m + n} \binom{m + n}{l} x^{m + n - l}  \int_{0}^{\infty} y^{l}  e^{-y(1+a)} \,_{2}F_{1} \left( n + \frac{1}{2}, 1 + 2n, 2(x + y)\right) {\rm d}y. 
\end{split}
 \end{equation}
Importantly, the integral in \eqref{m+n_Integer_5} can be expressed in closed-form with the aid of \cite[eq. (3.35.1.9)]{B:Prudnikov_4}. As a result \eqref{m+n_Integer_1} is deduced, which completes the proof.

\section*{Appendix J \\  Proof of Theorem 9}

By reversing the integral limits and applying \cite[eq. (9.238.2)]{B:Tables} it follows that,    

\begin{equation} \label{-n=n_2}
Ie_{-n, n}(x;a)  =  \int_{0}^{\infty}  \frac{ \, _{1}F_{1} \left( n + \frac{1}{2}, 1 + 2n, 2y \right)}{ 2^{n} n! e^{ y(1+a)}} {\rm d}y  -  \int_{x}^{\infty}  \frac{\, _{1}F_{1} \left( n + \frac{1}{2}, 1 + 2n, 2y \right)}{ 2^{n} n! e^{ y(1+a)} } {\rm d}y. 
\end{equation}
The first integral can be evaluated analytically with the aid of \cite[eq. (7.521)]{B:Tables}. To this effect and by setting in the second integral $u = 2y + x $ and carrying out some basic algebraic manipulations one obtains,  

\begin{equation} \label{-n=n_3} 
Ie_{-n, n}(x;a)  = \frac{\, _{2}F_{1}\left(n + \frac{1}{2}, 1; 1 + 2n; \frac{2}{1 + a} \right) }{2^{n} n! (1 + a)} -  \int_{0}^{\infty} \frac{\,_{1}F_{1}\left( n + \frac{1}{2}, 1 + 2n, u + 2x \right)}{ 2^{n+1} n!  e^{x(1+a)} e^{ (1+a) \frac{u}{2} }} {\rm d}u. 
\end{equation}
The above integral can be expressed in terms of the Humbert function $\Phi_{1}$ according to
\cite[eq. (3.35.1.9)]{B:Prudnikov_4}. To this effect, one obtains the closed-form expression in \eqref{-n=n_1}.  

For the special case that $m = n = 0$ in \eqref{ILHIs_Definition} and setting  $u = ay \Rightarrow y = u{/}a$ and ${\rm d}{/}{\rm d}t = a$ one obtains, 

\begin{equation} \label{ILHIs_Zeros_3}
Ie_{0,0}(x;a) = \int_{0}^{x} e^{-ay} I_{0}(y){\rm d}y    
\end{equation}
which can be equivalently expressed as,

\begin{equation} \label{ILHIs_Zeros_4}
Ie_{0,0}(x;a) =  \int_{0}^{ax}  \frac{e^{-u}}{a}  I_{0}\left(\frac{u}{a}\right){\rm d}u.   
\end{equation}
The above integral can be expressed in closed-form according to \cite[eq. (2c)]{J:Pawula}. Thus, by substituting in \eqref{ILHIs_Zeros_4} and after   basic algebraic manipulations, eq. \eqref{ILHIs_Zeros} is deduced. This completes the proof.

\section*{Appendix K \\ Proof of Proposition 6}

By making the necessary variable transformation in \cite[eq. (19)]{J:Gross_2} and substituting in \eqref{ILHIs_Polynomial_1} one obtains,

\begin{equation} \label{ILHIs_Polynomial_2}
Ie_{m,n}(x;a)   \simeq  \sum_{l=0}^{L}\int_{0}^{x} \frac{\Gamma(L+l)L^{1 - 2l}y^{m+n+2l} }{l!(L-l)!(n + l)!2^{n + 2l}  e^{ay}} {\rm d}y. 
\end{equation}
The above integral can be expressed in closed-form according to  \cite[eq. (3. 381.3)]{B:Tables}, namely,

\begin{equation} \label{ILHIs_Polynomial_3}
\int_{0}^{x} y^{m+n+2l} e^{-ay} {\rm d}y = \frac{\gamma(m + n + 2l + 1, ax)}{a^{m + n + 2l + 1}}. 
\end{equation}
By substituting \eqref{ILHIs_Polynomial_3} into \eqref{ILHIs_Polynomial_2} equation \eqref{ILHIs_Polynomial_1} is deduced. To this effect and as $L \rightarrow \infty$, the terms

$$
\frac{\Gamma(L+l)L^{1 - 2l}}{(L-l)!}
$$
 vanish and \eqref{ILHIs_Polynomial_1} becomes  the exact infinite series in \eqref{ILHIs_Infinite_Series_1}, which completes the proof.

\bibliographystyle{IEEEtran}

\thebibliography{99}

\bibitem{J:Marcum_1} 
J. I. Marcum, 
``A statistical theory of target detection by pulsed radar: Mathematical appendix," 
\emph{RAND Corp.}, Santa Monica, Research memorandum, CA, 1948. 

\bibitem{J:Marcum_2} 
J. I. Marcum, 
``Table of $Q{-}$functions, U.S. Air Force project RAND Res. Memo. M${-}$339, ASTIA document AD 116545," 
\emph{RAND Corp.}, Santa Monica, CA, 1950.

\bibitem{J:Marcum_3}  
J. I. Marcum,
``A statistical theory of target detection by pulsed radar," 
\emph{IRE Trans. on Inf. Theory}, vol. IT-6, pp. 59${-}$267, April 1960.

\bibitem{J:Swerling} 
P. Swerling, 
``Probability of detection for fluctuating targets," 
\emph{IRE Trans. on Inf. Theory}, vol. IT-6, pp. 269${-}$308, April 1960.

\bibitem{B:Proakis_Book}  
J. G. Proakis,
\emph{Digital Communications}, 3rd ed. New York: McGraw - Hill, 1995.

\bibitem{J:Shnidman1989} 
D. A. Shnidman,
``The calculation of the probability of detection and the generalized Marcum $Q{-}$function," 
\emph{IEEE Trans. Inf. Theory}, vol. 35, Issue no. 2, pp. 389${-}$400, Mar. 1989.

\bibitem{J:Helstrom1992} 
C. W. Helstrom,
``Computing the generalized Marcum $Q{-}$function,"
\emph{ IEEE Trans. Inf. Theory}, vol. 38, no. 4, pp. 1422${-}$1428, Jul. 1992. 

\bibitem{Simon_Marcum} 
M. K. Simon, 
``A new twist on the Marcum $Q{-}$function and its application," 
\emph{IEEE Commun. Lett.}, vol. 2, no. 2, pp. 39${-}$41, Feb. 1998.

\bibitem{J:Nuttall} 
A. H. Nuttall, 
``Some integrals involving the $Q{-}$function," \emph{Naval underwater systems center}, New London Lab, New London, CT, 4297, 1972.

\bibitem{B:Alouini} 
M. K. Simon and M. S. Alouini,
``Exponential${-}$type bounds on the generalized Marcum $Q{-}$function with application to error probability over fading channels,"
\emph{IEEE Trans. Commun.}, vol. 48, no. 3, pp. 359${-}$366, Mar. 2000.



\bibitem{J:Simon_2} 
M. K. Simon, 
``The Nuttall $Q{-}$function${-}$its relation to the Marcum $Q{-}$function and its application in digital communication performance evaluation," 
\emph{IEEE Trans. Commun.}, vol. 50, no. 11, pp. 1712${-}$1715, Nov. 2002.

\bibitem{J:Hatley} 
A. H. Heatley,
``A short table of the Toronto functions,
 \emph{Trans. Roy. Soc. (Canada)}, vol. 37, Sec. III. 1943.

\bibitem{J:Sagon} 
H. Sagon,
``Numerical calculation of the incomplete Toronto function," 
\emph{Proceedings of the IEEE}, vol. 54, Issue 8, pp. 1095${-}$1095, Aug. 1966.

\bibitem{J:Fisher} 
R. A. Fisher,
``The general sampling distribution of the multiple correlation coefficient," 
\emph{Proc. Roy. Soc. (London)}, Dec. 1928.

 \bibitem{J:Rice_1}  
S. O. Rice,
``Statistical properties of a sine wave plus random noise," 
\emph{Bell Syst. Tech.} J., 27, pp. 109${-}$157, 1948.

\bibitem{B:Roberts} 
J. H. Roberts,
\emph{Angle Modulation}, Peregrinus, Stevenage, UK, 1977.

\bibitem{J:Rice_2} 
R. F. Pawula, S. O. Rice and J. H. Roberts,
``Distribution of the phase angle between two vectors perturbed by Gaussian noise," 
\emph{IEEE Trans. Commun}. vol. COM-30, pp. 1828${-}$1841, Aug. 1982.

\bibitem{J:Tan} 
B. T. Tan, T. T. Tjhung, C. H. Teo and P. Y. Leong,
``Series representations for Rice's $Ie{-}$function," 
\emph{IEEE Trans. Commun}. vol. COM-32, no. 12, Dec. 1984.

\bibitem{J:Agrest1971} 
M.M. Agrest,
``Bessel function expansions of incomplete Lipschitz-Hankel integrals," \emph{USSR Computational Mathematics and Mathematical Physics}, vol.  11, no. 5, pp. 40${-}$54, May 1971.  

\bibitem{B:Maksimov} 
M. M. Agrest and M. Z. Maksimov,
\emph{Theory of Incomplete Cylindrical Functions and Their Applications}, Springer-Verlag, New York, 1971.

\bibitem{J:Miller1989} 
A. R. Miller, 
``Incomplete Lipschitz-Hankel integrals of Bessel functions," 
\emph{J. Math. Anal. Appl.}, vol. 140, pp. 476${-}$484, 1989.

\bibitem{J:Dvorak} 
S. L. Dvorak,
``Applications for incomplete Lipschitz-Hankel integrals in electromagnetics," 
\emph{IEEE Antennas Prop. Mag.} vol. 36, no. 6, pp. 26${-}$32, Dec. 1994.

\bibitem{B:Alouini_2005} 
M. K. Simon and M.-S. Alouni, 
\emph{Digital Communication over Fading Channels}, 2nd ed. New York: Wiley, 2005.

 \bibitem{J:Loskot2009}  
P. Loskot and N. C. Beaulieu, 
``Prony and polynomial approximations for evaluation of the average probability of error over slow-fading channels," 
\emph{ IEEE Trans. Veh. Technol.}, vol. 58, no. 3, pp. 1269${-}$1280, Mar. 2009. 

\bibitem{J:Cheng}  
X. Li and J. Cheng,
``Asymptotic error rate analysis of selection combining on generalized correlated Nakagami${-}m$ channels," 
\emph{IEEE Trans. Commun.}, vol. 60, no. 7, pp. 1765${-}$1771, July 2012. 

 \bibitem{J:Tellambura2013} 
 S. P. Herath, N. Rajatheva and C. Tellambura,
 ``Energy detection of unknown signals in fading and diversity reception,"
\emph{IEEE Trans. Commun.}, vol. 59, no. 9, pp. 2443${-}$2453, Sep. 2011. 

\bibitem{J:Paris2009_EL} 
J. F. Paris,
``Nakagami${-}q$ (Hoyt) distribution function with applications,"
\emph{Electron. Lett.}, vol. 45, no. 4, pp. 210${-}$211, Feb. 2009.

\bibitem{J:Paris_Hoyt_2010}
J. F. Paris, and D. Morales-Jimenez,
``Outage probability analysis for Nakagami${-}q$ (Hoyt) fading channels under Rayleigh interference,"
\emph{IEEE Trans. Wirel. Commun.}, vol. 9, no. 4, pp. 1272${-}$1276, Apr. 2010. 


\bibitem{J:Jimenez2010} 
D. Morales-Jimenez and Jose F. Paris,
``Outage probability analysis for $\eta{-}\mu$ fading channels,"
\emph{IEEE Commun. Lett.}, vol. 14, no. 6, pp. 521${-}$523, June 2010.

\bibitem{C:Alsusa} 
J-C. Shen, E. Alsusa and D. K. C. So,
``Performance bounds on cyclostationary feature detection over fading channels," 
\emph{Proc. IEEE WCNC 2013}, pp. 1${-}$5, Shanghai, China,  7${-}$10 Apr. 2013. 

\bibitem{J:Farina} 
G. Cui, A. De Maio, M. Piezzo and A. Farina,
``Sidelobe blanking with generalized Swerling${-}$chi fluctuation models," 
\emph{IEEE Trans. Aerosp. Electron. Syst.}, vol. 49, no. 2, pp. 982${-}$1005, April 2013. 

\bibitem{J:Tulino2010} 
G. Alfano, A. De Maio, and A. Tulino, 
``A theoretical framework for LMS MIMO communication systems performance analysis," \emph{IEEE Trans. Inf. Theory}, vol. 56, no. 11, pp. 5614${-}$5630, Nov. 2010.

\bibitem{J:Paris2012_CommL} 
D. Morales-Jimenez, J. Paris, and A. Lozano, 
``Outage probability analysis for MRC in $\eta{-}\mu$ fading channels with co-channel interference," 
\emph{IEEE Commun. Lett.}, vol. 16, no. 5, pp. 674${-}$677, May 2012.

\bibitem{J:Labao}
Juan P. Pena-Martın, J. M. Romero-Jerez, and C. Tellez-Labao, 
``Performance of TAS${/}$MRC wireless systems under Hoyt fading channels,"
\emph{IEEE Trans. Wirel. Commun.}, vol. 12, no. 7, pp. 3350${-}$3359, July 2013. 
 

\bibitem{J:Paris2013_TCOM} 
F. J. Lopez-Martinez, D. Morales-Jimenez, E. Martos-Naya, and J. F.
Paris, 
``On the bivariate Nakagami${-}m$ cumulative distribution function: closed-form expression and applications,"
\emph{ IEEE Trans. Commun.}, vol. 61, no. 4, pp. 1404${-}$1414, Apr. 2013.

\bibitem{J:Baricz2009_EL} 
Y. Sun, A Baricz, M. Zhao, X. Xu, and S. Zhou, 
``Approximate average bit error probability for DQPSK over fading channels,"
\emph{Electron. Lett.}, vol. 45, no. 23, pp. 1177${-}$1179, Nov. 5, 2009.

\bibitem{J:Affes2009} 
I. Trigui, A. Laourine, S. Affes, and A. Stephenne, 
``Performance analysis of mobile radio systems over composite fading${/}$shadowing channels with co-located interference,"
\emph{IEEE Trans. Wireless Commun.}, vol. 8, no. 7, pp. 3448${-}$3453, July 2009.

\bibitem{J:Ferrari2002_T-IT} 
G. E. Corazza and G. Ferrari, 
``New bounds for the Marcum $Q{-}$function,"
\emph{IEEE Trans. Inf. Theory}, vol. 48, pp. 3003${-}$3008, Nov. 2002.

\bibitem{C:Kam2006_ISIT} 
R. Li and P. Y. Kam, 
``Computing and bounding the generalized Marcum $Q{-}$function via a geometric approach," 
\emph{Proc. IEEE ISIT 2006}, pp. 1090${-}$1094, Seattle, USA, 9${-}$14 July 2006.

\bibitem{C:Li2006} 
P. Y. Kam and R. Li, 
``Generic exponential bounds and erfc-bounds on the Marcum $Q{-}$function via the geometric approach,"
\emph{ Proc. IEEE Globecom 2006}, pp. 1${-}$5, San Francisco, CA, USA, 27 Nov. ${-}$ 1 Dec., 2006.

\bibitem{J:Zhao2008_EL} 
X. Zhao, D. Gong and Y. Li,
``Tight geometric bound for Marcum $Q{-}$function," 
\emph{Electron. Letters}, vol. 44, no. 5, pp. 340${-}$341, Feb. 2008. 

\bibitem{J:Kam2008} 
P. Y. Kam and R. Li,
``Computing and bounding the first-order Marcum $Q{-}$function: A geometric approach," 
\emph{IEEE Trans. Commun.}, vol. 56, no. 7, pp. 1101${-}$1110, July 2008. 

\bibitem{C:Sun2008_Globecom} 
Y. Sun and S. Zhou, 
``Tight bounds of the generalized Marcum $Q{-}$function based on log-concavity,"
\emph{Proc. IEEE Globecom 2008}, New Orleans, LA, pp. 1${-}$5, 30 Nov. ${-}$ 4 Dec. 2008.

\bibitem{J:Baricz2008} 
Y. Sun and A. Baricz,
``Inequalities for the generalized Marcum $Q{-}$function," 
\emph{Applied Mathematics and Computation}, vol. 203, pp. 134${-}$141, 2008.

\bibitem{J:Baricz} 
A. Baricz and Y. Sun, ``New bounds for the generalized Marcum $Q{-}$function,"
\emph{IEEE Trans. Inf. Theory}, vol. 55, no. 7, pp. 3091${-}$3100, July 2009.

\bibitem{J:Karagiannidis} 
V. M. Kapinas, S. K. Mihos and G. K. Karagiannidis,
``On the monotonicity of the generalized Marcum and Nuttall $Q{-}$functions," 
\emph{IEEE Trans. Inf. Theory}, vol. 55, no. 8, pp. 3701${-}$3710, Aug. 2009.

\bibitem{J:Baricz2009_JMAA} 
A. Baricz, 
``Tight bounds for the generalized Marcum $Q{-}$function,"
\emph{J. Math. Anal. Appl.}, vol. 360, pp. 265${-}$277, 2009.

\bibitem{J:Abreu} 
G. T. F.  de Abreu,
``Jensen-Cotes upper and lower bounds on the Gaussian $Q{-}$function and related functions,"
\emph{IEEE Trans. Commun.}, vol. 57, no. 11, pp. 3328${-}$3338, Nov. 2009. 

\bibitem{J:Kam201_T-IT} 
R. Li, P. Y. Kam, and H. Fu, 
``New representations and bounds for the generalized Marcum $Q{-}$Function via a geometric approach, and an application,"
\emph{IEEE Trans. Commun.}, vol. 58, no. 1, pp. 157${-}$169, Jan. 2010.

\bibitem{J:Baricz2010_IT} 
Y. Sun, A. Baricz and S. Zhou,
``On the monotonicity, log-concavity, and tight bounds of the generalized Marcum and
Nuttall $Q{-}$functions,"
\emph{IEEE Trans. Inf. Theory}, vol. 56, no. 3, pp. 1166${-}$1186, Mar. 2010. 

\bibitem{C:Kam2011} 
H. Fu and P. Y. Kam,
``Exponential-type bounds on the first-order Marcum $Q{-}$function,"
\emph{Proc. IEEE Globecom 2011}, pp. 1${-}$5, Houston, Texas, USA,  5${-}$9 Dec. 2011.

\bibitem{J:Koh2013_IET} 
I-S. Koh and S. P. Chang,
``Uniform bounds of first-order Marcum $Q{-}$function,"
\emph{IET Commun.}, vol. 7, no. 13, pp. 1331${-}$1337, July 2013.

\bibitem{C:Kam2006_VTC} 
P. Y. Kam and R. Li, 
``A new geometric view of the first-order Marcum $Q{-}$function and some simple tight erfc${-}$bounds," 
\emph{Proc. IEEE $63^{\rm rd}$ VTC 2006 Spring}, pp. 2553${-}$2557, Melbourne, Australia, 7${-}$10 May 2006.

 \bibitem{C:Ding2008} 
 N. Ding and H. Zhang, 
 ``A flexible method to approximate Marcum $Q{-}$function based on geometric way of thinking,"
\emph{ Proc. ISCCSP 2008}, pp. 1351${-}$1356,  St. Julian's, Malta  12${-}$14 Apr. 2008.

\bibitem{J:Andras2011} 
S. Andras, A. Baricz and Y. Sin,
``The generalized Marcum $Q{-}$function: an orthogonal polynomial approach,"
\emph{Acta Univ. Sapientiae, Mathematica}, vol. 3, no. 1, pp. 60${-}$76, Jan. 2011.

\bibitem{J:Marcum_New_1} 
M. Z. Bocus, C. P. Dettmann and J. P. Coon, 
``An approximation of the first order Marcum Q${-}$function with application to network connectivity analysis," 
\emph{IEEE Commun. Lett.}, vol. 17, no. 3, pp. 499${-}$502, Mar. 2013.

\bibitem{J:Brychkov2012} 
Yu. A. Brychkov,
``On some properties of the Marcum $Q{-}$ function,"
\emph{Integral Transforms and Special Functions}, vol. 23, no. 3, pp. 177${-}$182, Mar. 2012.

\bibitem{B:Vasilis_PhD}
V. M. Kapinas, 
\emph{Optimization and Performance Evaluation of Digital Wireless Communication Systems with Multiple Transmit and Receive Antennas}, Ph.D. dissertation, Aristotle University of Thessaloniki, Thessaloniki, Greece, 2014.


\bibitem{J:Lozano2013} 
D. Morales-Jimenez, F. J. Lopez-Martinez, E. Martos-Naya, J. F. Paris and A. Lozano,
``Connections between the generalized Marcum $Q{-}$function and a class of hypergeometric functions,"
 \emph{IEEE Trans. Inf. Theory}, vol. 60, no. 2, pp. 1077${-}$1082, Feb. 2014.





\bibitem{J:Karasawa}
Q. Shi and Y. Karasawa,
``An intuitive methodology for efficient evaluation of the Nuttall $Q{-}$Function and performance analysis of energy detection in fading channels," 
\emph{IEEE Wirel. Commun. Lett.}, vol. 1, no. 2, pp. 109${-}$112, Apr. 2012.

\bibitem{J:Brychkov2013}
Yu. A. Brychkov,
``On some properties of the Nuttall function $Q_{\mu, \nu}(a, b)$," 
\emph{Integral Transforms and Special Functions}, vol. 25, no. 1, pp. 34${-}$43, Jan. 2014. 






\bibitem{J:Pawula} 
R. F. Pawula,
``Relations between the Rice $Ie{-}$function and the Marcum $Q{-}$function with applications to error rate calculations," \emph{Elect. Lett}. vol. 31, no. 24, pp. 2078${-}$2080, Nov. 1995.

\bibitem{B:Tables} 
I. S. Gradshteyn and I. M. Ryzhik, 
\emph{Table of integrals, series, and products}, in $7^{th}$ ed.  Academic, New York, 2007.




\bibitem{J:Paris}  
J. F. Paris, E. Martos-Naya, U. Fernandez-Plazaola and J. Lopez-Fernandez
``Analysis of adaptive MIMO transmit beamforming under channel prediction errors based on incomplete Lipschitz-Hankel integrals," 
\emph{IEEE Trans. Veh. Technol.}, vol. 58, no. 6, pp. 2815${-}$2824, July 2009.


\bibitem{B:Sofotasios} 
P. C. Sofotasios,
\emph{On Special Functions and Composite Statistical Distributions and Their Applications in Digital Communications over Fading Channels}, Ph.D. Dissertation, University of Leeds, England, UK, 2010. 


\bibitem{C:Sofotasios_1}
P. C. Sofotasios, and S. Freear, 
``A novel representation for the Nuttall Q-function,''
\emph{Proc. IEEE ICWITS `10}, pp. 1${-}$4, Hawaii, HI, USA,  28 Aug.  ${-}$ 3 Sep.  2010. 


\bibitem{C:Sofotasios_2}
P. C. Sofotasios, and S. Freear, 
``Novel expressions for the one and two dimensional Gaussian $Q-$functions,''
\emph{Proc. IEEE ICWITS `10}, pp. 1${-}$4, Hawaii, HI, USA,  28 Aug.  ${-}$ 3 Sep.  2010. 


\bibitem{C:Sofotasios_3}
P. C. Sofotasios, and S. Freear, 
``Novel expressions for the Marcum and one dimensional $Q-$functions,''
\emph{Proc. ISWCS `10}, pp. 736${-}$740, York, UK,  18${-}$21 Sep.  2010. 


\bibitem{C:Sofotasios_4}
P. C. Sofotasios, and S. Freear, 
``New analytic results for the incomplete Toronto function and incomplete Lipschitz-Hankel Integrals,''
\emph{Proc. IEEE VTC `11 Spring}, pp. 1${-}$4, Budapest, Hungary, 15${-}$18 May 2011.  


\bibitem{C:Sofotasios_5}
S, Harput, P. C. Sofotasios, and S. Freear, 
``A Novel Composite Statistical Model For Ultrasound Applications," 
\emph{Proc. IEEE IUS `11}, pp. 1${-}$4, Orlando, FL, USA, 8${-}$10 Oct. 2011.  

\bibitem{C:Sofotasios_6}
P. C. Sofotasios, and S. Freear, 
``Simple and accurate approximations for the two dimensional Gaussian $Q{-}$Function,''
\emph{Proc. SBMO/IEEE IMOC  `11}, pp. 44${-}$47, Natal, Brazil, 29${-}$31 Oct. 2011.  


\bibitem{C:Sofotasios_7}
P. C. Sofotasios, and S. Freear, 
``Upper and lower bounds for the Rice $Ie-$function,''
\emph{IEEE ATNAC `11}, pp. 1${-}$4, Melbourne, Australia, 9${-}$11 Nov.  2011. 

\bibitem{C:Sofotasios_8}
P. C. Sofotasios, and S. Freear, 
``Analytic expressions for the Rice $Ie-$function and the incomplete Lipschitz-Hankel integrals,''
\emph{IEEE INDICON `11}, pp. 1${-}$6, Hyderabad, India,  16${-}$18 Dec.  2011. 

\bibitem{C:Sofotasios_9}
P. C. Sofotasios, K. Ho-Van, T. D. Anh, and H. D. Quoc,
``Analytic results for efficient computation of the Nuttall$-Q$ and incomplete Toronto functions,''
\emph{Proc. IEEE ATC `13}, pp. 420${-}$425, HoChiMinh City, Vietnam, 16${-}$18 Oct. 2013.  

\bibitem{C:Sofotasios_10}
P. C. Sofotasios, M. Valkama, T. A. Tsiftsis, Yu. A. Brychkov, S. Freear, and G. K. Karagiannidis, 
\emph{Proc. CROWNCOM `14}, pp. 260${-}$265, Oulu, Finland, 2${-}$4 June 2014.  



\bibitem{B:Abramowitz}
M. Abramowitz and I. A. Stegun, 
\emph{Handbook of Mathematical Functions with Formulas, Graphs, and Mathematical tables.}, New York: Dover, 1974.



\bibitem{J:Bailey}
W. N. Bailey, 
``Appell's hypergeometric functions of two variables," 
\emph{Ch. 9 in generalised hypergeometric series. Cambridge, England.}  Cambridge University Press, pp. 73${-}$83 and 99${-}$101, 1935.

\bibitem{KdF_1}
V. L. Deshpande,
``Expansion theorems for the Kampe de Feriet function," \emph{Indagationes Mathematicae (Proceedings)}, vol. 74, pp. 39${-}$46, 1971. 


\bibitem{B:Exton}
H. Exton, 
\emph{Handbook of Hypergeometric Integrals}, Mathematics and its applications, Ellis Horwood Ltd., 1978. 

\bibitem{J:Srivastava}
H. M. Srivastava and P. W. Karlsson, 
\emph{Multiple Gaussian Hypergeometric Series. Mathematics and its Applications,} 
Halsted Press, Ellis Horwood Ltd.  Chichester, UK:  1985.

\bibitem{J:Watson}
E. T. Whittaker and G. N. Watson, 
\emph{A Course in Modern Analysis,} 
Cambridge University Press, 4th ed. Cambridge, England:  1990.

\bibitem{KdF_2}
H. Exton,
`` Transformations of certain generalized Kampe de Feriet functions II," \emph{Journal of Applied Mathematics and Stochasti Analysis}, vol. 10, no. 3, pp. 297${-}$304, 1997. 



\bibitem{B:Prudnikov} 
A. P. Prudnikov, Yu. A. Brychkov, and O. I. Marichev, 
\emph{Integrals and Series}, 3rd ed. New York: Gordon and Breach Science, vol. 1, Elementary Functions, 1992.


\bibitem{J:Gross_2} 
L- L. Li, F. Li and F. B. Gross,
``A new polynomial approximation for $J_{m}$ Bessel functions,"
 \emph{Elsevier Journal of Applied Mathematics and Computation}, vol. 183, pp. 1220${-}$1225, 2006. 



\bibitem{J:Humbert}
P. Humbert, 
``Sur les fonctions hypercylindriques", 
\emph{[J] C. R. 171}, ISSN 0001${-}$4036, pp. 490${-}$492,  1920.  

\bibitem{J:Brychkov}
Yu. A. Brychkov, Nasser Saad,
``Some formulas for the Appell function $F_{1} (a, b, b′; c; w, z)$," 
\emph{Integral Transforms and Special Functions}, vol. 23, no. 11, pp. 793${-}$802, Nov. 2012. 

\bibitem{B:Prudnikov_4} 
A. P. Prudnikov, Yu. A. Brychkov and O. I. Marichev, 
\emph{Integrals and Series}, Gordon and Breach Science, vol. 4, Direct Laplace Transforms, 1992.


\bibitem{Yacoub_NL_1}
G. Fraidenraich,  M. D. Yacoub, 
``The $\alpha{-}\eta{-}\mu$ and $\alpha{-}\kappa{-}\mu$ fading distributions,"
\emph{ Proc. 9$^{ th}$ IEEE ISSTA '06}, pp. 16${-}$20, Manaus-Amazon, Brazil,  28${-}$31 Aug. 2006. 

\bibitem{Yacoub_NL_2}
R. Cogliatti, R. A. Amaral de Souza, 
``A near${-}100\%$ efficient algorithm for generating $\alpha{-}\kappa{-}\mu$ and $\alpha{-}\eta{-}\mu$ variates,"
\emph{ Proc. IEEE VTC-Fall '13}, pp. 1${-}$5, Las Vegas, NV, USA, 2${-}$5 Sep. 2013. 







\bibitem{Yacoub_a-m}
M. D. Yacoub,
``The $\alpha{-}\mu$ distribution: A physical fading model for the Stacy distribution," 
\emph{IEEE Trans. Veh. Technol.}, vol. 56, no. 1, pp. 27${-}$34, Jan. 2007. 

\bibitem{Yacoub_h-m_k-m}
M. D. Yacoub,
``The $\kappa{-}\mu$ distribution and the $\eta{-}\mu$ distribution," 
\emph{IEEE Antennas Propag. Mag.}, vol. 49, no. 1, pp. 68${-}$81, Feb. 2007.

\bibitem{Yacoub_l-m}
G. Fraidenraich,  M. D. Yacoub, 
``The $\lambda{-}\mu$ general fading distribution," \emph{Proc. SMBO${/}$IEEE MTT-S IMOC '03}, pp. 49${-}$54, Foz do Iguacu, Brazil, 20${-}$23 Sep. 2003.  



\bibitem{Additional_1}
P. C. Sofotasios, M. K. Fikadu, K. Ho-Van, M. Valkama, and G. K. Karagiannidis, 
``The area under a receiver operating characteristic curve over enriched multipath fading conditions,''
\emph{in IEEE Globecom '14}, Austin, TX, USA, Dec. 2014, pp. 3090$-$3095. 


\bibitem{Additional_2}
M. K. Fikadu, P. C. Sofotasios, M. Valkama, and Q. Cui, 
``Analytic performance evaluation of $M-$QAM based decode-and-forward relay networks over enriched multipath fading channels,'' 
\emph{in IEEE WiMob '14}, Larnaca, Cyprus,  Oct. 2014, pp. 194$-$199.


\bibitem{Additional_3}
P. C. Sofotasios, T. A. Tsiftsis, K. Ho-Van, S. Freear, L. R. Wilhelmsson, and M. Valkama, 
``The $\kappa-\mu$/inverse-Gaussian composite statistical distribution in RF and FSO wireless channels,''
\emph{in IEEE VTC '13 - Fall}, Las Vegas, USA,  Sep. 2013, pp. 1$-$5.


\bibitem{Additional_4}
P. C. Sofotasios, T. A. Tsiftsis, M. Ghogho, L. R. Wilhelmsson and M. Valkama, 
``The $\eta-\mu$/inverse-Gaussian Distribution: A novel physical multipath/shadowing fading model,''
\emph{in IEEE ICC '13}, Budapest, Hungary, June 2013. 


\bibitem{Additional_5}
P. C. Sofotasios, and S. Freear, 
``The $\alpha-\kappa-\mu$/gamma composite distribution: A generalized non-linear multipath/shadowing fading model,''
\emph{IEEE INDICON  `11}, Hyderabad, India, Dec. 2011. 


\bibitem{Additional_6}
P. C. Sofotasios, and S. Freear,
``The $\alpha-\kappa-\mu$ extreme distribution: characterizing non linear severe fading conditions,'' 
\emph{ATNAC `11}, Melbourne, Australia, Nov. 2011. 


\bibitem{Additional_7}
P. C. Sofotasios, and S. Freear, 
``The $\eta-\mu$/gamma and the $\lambda-\mu$/gamma multipath/shadowing distributions,'' 
\emph{ATNAC  `11}, Melbourne, Australia, Nov. 2011. 


\bibitem{Additional_8}
P. C. Sofotasios, and S. Freear, 
``On the $\kappa-\mu$/gamma composite distribution: A generalized multipath/shadowing fading model,'' 
\emph{IEEE IMOC `11},  Natal, Brazil, Oct. 2011, pp. 390$-$394.


\bibitem{Additional_9}
P. C. Sofotasios, and S. Freear, 
``The $\kappa-\mu$/gamma extreme composite distribution: A physical composite fading model,''
\emph{IEEE WCNC  `11},  Cancun, Mexico, Mar. 2011, pp. 1398$-$1401.


\bibitem{Additional_10}
P. C. Sofotasios, and S. Freear,
``The $\kappa-\mu$/gamma composite fading model,''
\emph{IEEE ICWITS  `10}, Honolulu, HI, USA, Aug. 2010, pp. 1$-$4.


\bibitem{Additional_11}
P. C. Sofotasios, and S. Freear, 
``The $\eta-\mu$/gamma composite fading model,''
\emph{IEEE ICWITS `10}, Honolulu, HI, USA, Aug. 2010, pp. 1$-$4. 



\bibitem{J:Varaiya}
A. Goldsmith and P. Varaiya, 
``Capacity of fading channels with channel side information," 
\emph{IEEE Trans. Inform. Theory}, vol. 43, no. 11, pp. 1986${-}$1992, Nov. 1997.

\bibitem{J:Alouini_Goldsmith}
M.-S. Alouini and A. J. Goldsmith,
``Capacity of Rayleigh fading channels under different adaptive transmission and diversity -combining techniques," \emph{IEEE Trans. Veh. Technol.}, vol. 48, no. 4, pp. 1165${-}$1181, July 1999.

\bibitem{J:Jayaweera2005}
S. Jayaweera and H. Poor, 
``On the capacity of multiple-antenna systems in Rician fading,"
\emph{ IEEE Trans. Wireless Commun.}, vol. 4, no. 3, pp. 1102${-}$1111, May 2005.

\bibitem{J:Aissa}
A. Maaref and S. Aissa,
``Capacity of MIMO Rician fading channels with transmitter and receiver channel state information," 
\emph{IEEE Trans. Wireless Commun. }, vol. 7, no. 5, pp. 1687${-}$1698, May 2008.

\bibitem{J:Lozano_New}
G. Alfano, A. Lozano, A. Tulino and S. Verdu, 
``Mutual information and eigenvalue distribution of MIMO Ricean channels," 
\emph{Proc. IEEE Int. Symp. Information Theory $\&$ its Applications (ISITA '04)}, pp. 1${-}$6, Parma, Italy, 10${-}$13 Oct. 2004.


\end{document}